\documentclass[10pt]{article}
\usepackage[square,authoryear]{natbib}
\usepackage{graphicx}
\usepackage{marsden_article}
\usepackage{stackrel}
\usepackage{ulem}
\usepackage{stmaryrd}
\usepackage{mathrsfs}
\usepackage{latexsym,amsmath,amscd,amssymb,graphics,amsthm}
\usepackage{enumerate,framed}
\usepackage{framed}
\usepackage{here} 
\usepackage{subfigure}
\usepackage{amsmath}
\usepackage{color}
\usepackage{multirow}
\RequirePackage{color}

\newtheorem{remark}[theorem]{Remark}
\DeclareFontFamily{OT1}{pzc}{}
\DeclareFontShape{OT1}{pzc}{m}{it}{<-> s * [1.100] pzcmi7t}{}
\DeclareMathAlphabet{\mathpzc}{OT1}{pzc}{m}{it}

\title{Resonance, symmetry, and bifurcation of periodic orbits in perturbed Rayleigh-B\'enard convection}
\begin{document}
\date{}
\author{\hspace{-1cm}
\begin{tabular}{c}
$^{\dag}$Masahito Watanabe and  Hiroaki Yoshimura\footnote{Email: $^{\dag}$masa.watanabe@aoni.waseda.jp; yoshimura@waseda.jp}\\[5mm]  
Department of Applied Mechanics and Aerospace Engineering \\[1.5mm] 
School of Fundamental Science and Engineering\\[1.5mm]
Waseda University \\ [4mm] 
Okubo, Shinjuku, Tokyo 169-8555, Japan\\ [5mm]  
\end{tabular}\\
}

\maketitle

\begin{center}
\abstract{\vspace{2mm}
This paper investigates the global structures of periodic orbits that appear in Rayleigh-B\'enard convection, which is modeled by a two-dimensional perturbed Hamiltonian model, by focusing upon resonance, symmetry and bifurcation of the periodic orbits. 
First, we show the global structures of periodic orbits in the extended phase space by numerically detecting the associated periodic points on the Poincar\'e section. Then, we illustrate how resonant periodic orbits appear and specifically clarify that there exist some symmetric properties of such resonant periodic orbits which are projected on the phase space; namely, the period $m$ and the winding number $n$ become odd when an $m$-periodic orbit  is symmetric with respect to the horizontal and vertical center lines of a cell. Furthermore, the global structures of bifurcations of periodic orbits are depicted when the amplitude $\varepsilon$ of the perturbation is varied, since in experiments the amplitude of the oscillation of the convection gradually increases when the Rayleigh number is raised.}
\end{center}

\tableofcontents

\section{Introduction}
\label{Sec:introduction}
\paragraph{Background.} 
In the fields of meteorology, oceanography, and chemical engineering, much concern has been focused on the prediction and control of the spread of oil and chemical spills as well as the measurement of air pollutant concentrations. In particular, the natural convection in a horizontal fluid layer with heated bottom and cooled top planes called Rayleigh-Benard convection has been well known as a typical phenomenon of such fluid transport that exists in nature (see \cite{Ch1961}) and it is crucial to study the global fluid transport associated with natural convection. So far, the fluid transport in perturbed Rayleigh-Benard convection has been actively investigated; when the temperature difference of the two planes is relatively small, or, Rayleigh number $Ra$ is relatively small, 
multiple convection rolls with steady velocity fields may appear in parallel in the layer. 
When the flow in the direction of the roll axes is negligible, it may be considered as a two-dimensional steady convection from that direction. On the other hand, it was clarified by \cite{ClBu1974} and \cite{BoBuCl1986}
that the parallel convection rolls may start to wave slightly by the even oscillatory instability 
when $Ra$ is set slightly above a critical number $ Ra_{t}$ by increasing the temperature difference. 
Since the wave propagates along the roll axes almost periodically, 
the two-dimensional velocity field observed from the direction of the roll axes is perturbed. 

One of the important remarks is that although the velocity field of such oscillatory convection seems to be stable in  Eulerian description, some fluid particles can be transported chaotically in Lagrangian description; see \cite{Ot1989}.
Furthermore, increasing $Ra$ by raising the temperature difference, 
the amplitude of the oscillation enlarges and the fluid transport become very complicated. 
Since very rich dynamics such as Lagrangian chaotic fluid transport can be observed in the perturbed Rayleigh-Benard convection, 
the fluid transport in this convection has been actively studied by numerical and experimental methods. 
Amongst such past researches on the study of chaotic fluid transport in the perturbed Rayleigh-B\'enard convection, \cite{SoGo1988b} has been well known as a pioneer work, where the diffusion of impurities in the convection was studied by optical absorption techniques, in which the convection was modeled as a two-dimensional perturbed Hamiltonian system following the experimental results and it was numerically clarified that the basic mechanism of fluid transport is chaotic advection around cell boundaries rather than molecular diffusion.
\cite{GoSo1989} also made some numerical analysis to show some evidence of chaotic transport in the perturbed Hamiltonian model in the sense of being sensitive to the initial condition. 
In addition, \cite{OuMoHoMo1991} numerically studied the diffusion constant of the model and \cite{OuMo1992} showed that some anomalous diffusion is caused by the accelerator-mode islands of KAM tori around cell boundaries.
Furthermore, \cite{InHi1998, InHi2000} investigated the mixing patterns of another perturbed Hamiltonian model with different perturbations by analyzing Poincar\'e maps and the degree of mixing, and showed how the chaotic structures vary when the amplitude or the frequency of the oscillation is changed. 
\cite{SoMe2003} explored experimentally and numerically the uniform mixing of weakly three-dimensional and 
weakly time-periodic vortex flow by using magnetohydrodynamic techniques.

From the viewpoint of dynamical systems theory, \cite{CaWi1991a,CaWi1991b} investigated 
the stable and unstable manifolds of the perturbed Hamiltonian model of Rayleigh-B\'enard convection to clarify 
the mechanism of chaotic transport by the so-called "lobe dynamics"; see also \cite{Wi1992}. On the other hand,
\cite{SoToWa1996, SoToWa1998} experimentally detected some lobes by observing the transport of impurities in a fluid layer with a chain of horizontal vortices that are oscillated by magnetohydrodynamic forcing.
In addition, \cite{MaMeWi1998} studied the patchiness of the model with stable and unstable manifolds, 
where a patch is a region that has a considerably different average velocity compared to the surrounding region. 
\cite{ShLeMa2005} and \cite{LeShMa2007} numerically clarified the Lagrangian coherent structures (LCSs) 
of the perturbed Hamiltonian model, where LCS corresponds to the stable and unstable manifolds of non-autonomous systems; see also \cite{HaYu2000}. 

As mentioned in the above, most of the past works have been focused on the chaotic region of the fluid transport 
in perturbed Rayleigh-B\'enard convection rather than exploring the stable region of periodic orbits, 
or some of them have locally detected some elliptic periodic points in the perturbed Hamiltonian model with some fixed parameters.
For the sake of understanding the global structures of such fluid transport, it is quite essential to 
find both elliptic and hyperbolic periodic points in the perturbed Hamiltonian model for some range 
of parameters and investigate how the periodic orbits appear and bifurcate in the Rayleigh-B\'enard convection. 
Especially, it is crucial to investigate the resonance and symmetry of periodic orbits, 
since they are one of the important topological properties of periodic orbits. 
Furthermore, needless to say, it is necessary to clarify how the transport becomes complicated when $Ra$ is increased. In other words, how the periodic transport varies to chaos when the amplitude of the perturbation is increased. Although the transition of Rayleigh-B\'enard convection from steady to oscillatory and chaotic flow, 
and the resonance of quasi-periodic Rayleigh-B\'enard convection are investigated in the Eulerian description in many studies 
such as \cite{GoBe1980}, \cite{LiFaLa1983}, and \cite{EcKe1988}, 
the resonance and symmetry of periodic orbits and the global structures of bifurcations from periodic to chaotic orbits 
have not been clarified enough in the Lagrangian description in the perturbed Rayleigh-B\'enard convection.


\paragraph{Contributions and the organization of this paper.}
The main goals of this paper are to clarify the global structures of periodic orbits, the symmetric properties of resonant orbits, as well as the global bifurcations of the periodic orbits appeared in the two-dimensional Hamiltonian model of the perturbed Rayleigh-B\'enard convection.  To do this, we first introduce an autonomous Hamiltonian model in the extended phase space from the two-dimensional Hamiltonian model of the perturbed Rayleigh-B\'enard convection. Then, we numerically detect the elliptic and hyperbolic periodic points on the Poincar\'e section and investigate the structures of the associated periodic orbits in the extended phase space of the autonomous system. In particular, we consider the projection of the periodic orbits onto the original phase space to investigate the resonances and symmetry of the orbits. 
Lastly, we show the global structures of $\varepsilon$-parameter bifurcations of the periodic orbits, where $\varepsilon$ denotes the amplitude of the perturbation.

The organization of this paper consists of the following sections: In \S\ref{Sec:Poincare}, 
the model of the two-dimensional perturbed Hamiltonian system for the oscillatory Rayleigh-B\'enard convection 
is described together with symmetric properties. Then, numerical analysis is made by the Poincar\'e map to 
detect the periodic points on the Poincar\'e section $\Sigma^{\theta_0}$ and also to clarify 
the structures of periodic orbits and KAM tori in the extended phase space. 
In \S\ref{Sec:symmetry}, symmetries of resonant orbits are demonstrated by projecting the $m$-periodic orbits 
onto the phase space and a theorem for the resonant orbits is given that the period $m$ and the winding number $n$ are odd numbers, 
if the projection is symmetric with respect to the horizontal and vertical center lines of a cell. 
In \S\ref{Sec:bifurcation}, the $\varepsilon$-parameter bifurcations of periodic points 
are illustrated and, in particular, the classification of the bifurcations is made into fold or flip bifurcations 
according to the multipliers of the periodic points.
Finally in \S\ref{Sec:conclusion}, the conclusions of this paper are described.


\section{Poincar\'e map and structures of periodic orbits}
\label{Sec:Poincare}
In order to investigate the two-dimensional Rayleigh-B\'enard convection whose velocity field is perturbed by the even oscillatory instability, we employ the two-dimensional perturbed Hamiltonian system, which was originally developed by \cite{SoGo1988b}, see also \cite{CaWi1991a}. 
Then, we investigate the global structures of such periodic orbits that appear in the perturbed Hamiltonian system by Poincar\'e maps. 

\subsection{Model of perturbed Rayleigh-B\'enard convecton}
\paragraph{Hamiltonian system of steady Rayleigh-B\'enard convection.}
By assuming the stress-free boundary condition, it follows from the Navier-Stokes equations with the Boussinesq approximation that
two-dimensional steady Rayleigh-B\'enard convection can be modeled by a Hamiltonian system as 
\begin{equation}\label{RBCModel}
\begin{split}
\frac{dx}{dt}&=-\frac{\partial H_0(x,z)}{\partial z}=-\frac{A \pi}{k}{\sin(kx)\cos(\pi z)},\\[3mm]
\frac{dz}{dt}&=\frac{\partial H_0(x,z)}{\partial x}=A{\cos(kx)\sin(\pi z)},
\end{split}
\end{equation}
where $H_0(x,z)$ is a Hamiltonian, given by the stream function
\begin{equation*}
H_0(x,z) =  \frac{A}{k}{\sin(kx) \sin(\pi z)};
\end{equation*}
see \cite{Ch1961}.
In the above,  $x \in \mathbb{R}$ and $z \in U=[0,1] \subset \mathbb{R}$ are the horizontal and vertical coordinates respectively, and hence we define the phase space $M=\mathbb{R} \times U$. Further, $A$ denotes the amplitude of the velocity in $z$ direction and $k$ is the wave number of the cell pattern in $x$ direction. In this Hamiltonian system, we have the hyperbolic equilibrium points $p_{i,0}^{\pm}=(x_{i,0},z_i^{\pm})$ as
\begin{equation*}
(x_{i,0},z_i^{\pm})=\left( \frac{i \pi}{k}, z_i^{\pm}\right), \quad (i=0,\pm1, \; \pm2, \;...),
\end{equation*}
where $z^-_i=0$ and $z^+_i=1$, and it is noticed that there exist heteroclinic connections between  $p_{i,0}^{+}$ and $p_{i,0}^{-}$ along the roll boundaries.

\paragraph{Hamiltonian model of perturbed Rayleigh-B\'enard convecton.}
Now we consider the case in which a time-periodic term $\varepsilon \cos(\omega t)$ is added to $x$ in the Hamiltonian $H_0(x,z)$ for the steady Rayleigh-B\'enard convection. Then, it follows that a time-dependent Hamiltonian on the extended phase space $ M \times \mathbb{R} $ is given in coordinates $(x,z,t) \in  M  \times\mathbb{R} $ as
$$
H(x,z,t):=H_0(x,z)+\varepsilon H_1(x,z,t), 
$$
where Taylor expansion is applied to the sinusoidal term as
$$
H_1(x, z,t)=A \cos(\omega t){\cos(kx)\sin(\pi z)}. 
$$
Note that $A$ denotes some given constant of the magnitude. 
Then, we get a non-autonomous Hamiltonian vector field 
$X_{H}: M \times  \mathbb{R} \to T M$, locally given by 
\begin{equation}\label{PerHamEq}
\begin{split}
\frac{dx}{dt}&=-\frac{\partial H(x, z, t)}{\partial z}=-\frac{\partial H_0(x, z)}{\partial z}+\varepsilon \frac{\partial H_1(x, z,t)}{\partial z},\\[3mm]
\frac{dz}{dt}&=\frac{\partial H(x, z, t)}{\partial x}=\frac{\partial H_0(x, z)}{\partial x}+\varepsilon \frac{\partial H_1(x, z,t)}{\partial x}.
\end{split}
\end{equation}
In the above, $\varepsilon \in \mathbb{R}$ is a given magnitude of the perturbation and the perturbed terms $\frac{\partial H_1(x, z,t)}{\partial z}$ and $\frac{\partial H_1(x, z,t)}{\partial x}$ are respectively given by the periodic function:
\begin{equation*}
\begin{split}
\frac{\partial H_1(x, z,t)}{\partial z}&=- A \pi \cos (\omega t) {\cos(kx)\cos(\pi z)},\\[2mm]
\frac{\partial H_1(x, z,t)}{\partial x}&=-Ak \cos(\omega t) {\sin(kx) \sin(\pi z)}.
\end{split}
\end{equation*}
Fig.\ref{fig:rbflow} illustrates a schematic figure of this model, 
where the wavy dashed lines indicate the perturbed cell boundaries.

\begin{figure}[htb]
\begin{center}
\includegraphics[scale=0.45]{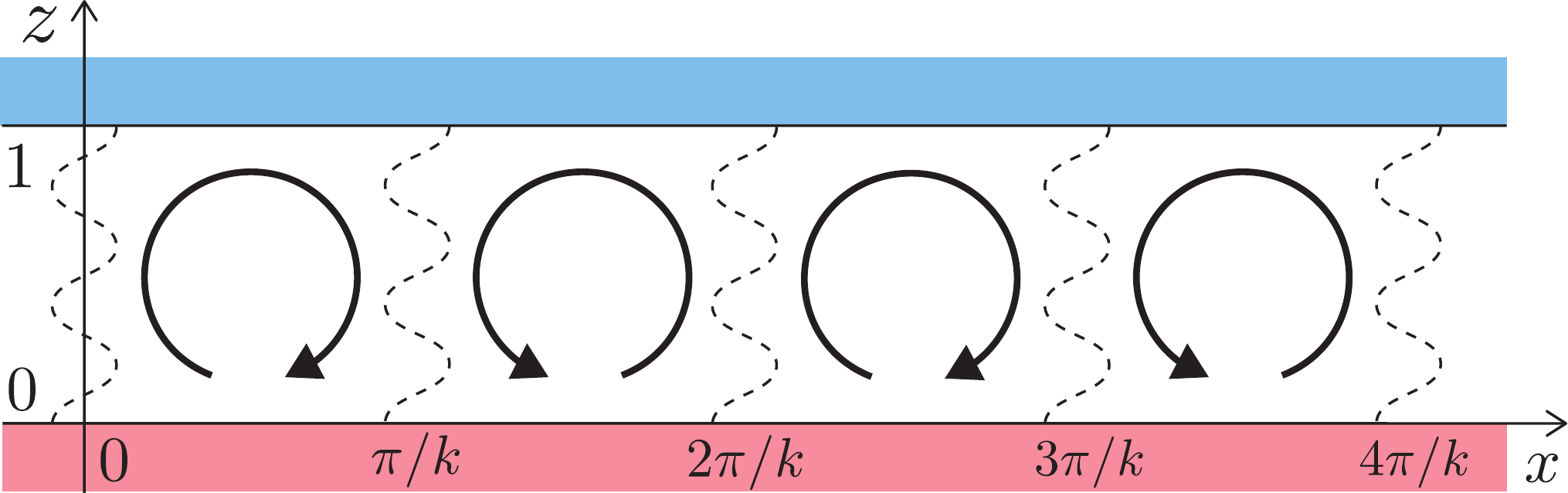}
\caption{Two-dimensional Rayleigh-B\'enard convection with perturbations}
\label{fig:rbflow}
\end{center}
\end{figure}

\paragraph{Symmetric properties of the model.}
Recall from \cite{CaWi1991a} that the perturbed Hamiltonian system \eqref{PerHamEq} is invariant under the following coordinate transformations:
\begin{equation*}
\begin{split}
&\textrm{i)}\qquad \displaystyle  x \mapsto x + \frac{2a\pi}{k},\quad z \mapsto -z+1, \quad t \mapsto -t+bT,\\[3mm]
&\textrm{ii)}\qquad \displaystyle  x \mapsto x + \frac{(2a+1)\pi}{k},\quad z \mapsto z, \quad t \mapsto -t+bT,\\[3mm]
&\textrm{iii)}\qquad \displaystyle  x \mapsto x + \frac{2a\pi}{k},\quad z \mapsto z, \quad t \mapsto t+bT,
\end{split}
\end{equation*}
where $T(=2\pi/\omega)$ is the period of the perturbation and $a,b \in \mathbb{Z}$.
Note that there exists two more symmetries associated with the following transformation:
\begin{eqnarray*}
\begin{split}
&\textrm{iv)}\qquad x \mapsto -x + \frac{(2a+1)\pi}{k},\quad z \mapsto z,\quad t \mapsto -t + \biggl(b+\frac{1}{2} \biggr)T,\\[3mm]
&\textrm{v)}\qquad x \mapsto -x + \frac{(2a+1)\pi}{k},\quad z \mapsto -z + 1,\quad t \mapsto t + \biggl(b+\frac{1}{2} \biggr)T,
\end{split}
\end{eqnarray*}
which will be used for investigating symmetric properties of periodic orbits in \S\ref{Sec:symmetry}.

\subsection{Structures of periodic points}
%
%
In this subsection, we numerically compute Poincar\'e maps to detect periodic points on a Poincar\'e section. To do this, we transform the perturbed Hamiltonian system that is a non-autonomous system on $M=\mathbb{R} \times U$ with local coordinates $(x,z)$ into the setting of an autonomous system by introducing the extended phase space $\mathcal{M}=M \times S^1$ with  local coordinates $(x,z,\theta)$ and then define a Poincar\'e map $P^{\theta_{0}}_{\varepsilon}: \Sigma^{\theta_{0}} \to \Sigma^{\theta_{0}}$, where $\Sigma^{\theta_{0}}\subset \mathcal{M}$ is a chosen Poincar\'e section. 

\paragraph{Autonomous Hamiltonian systems.}
By introducing an angle variable $\theta:=\omega t+\theta_0 \in S^1$, where $\theta_0 \in [0,2\pi)$, the Hamiltonian can be rewritten on the extended phase space $\mathcal{M} = M \times S^1$ as
$$
H(x,z,\theta):=H_0(x,z)+\varepsilon H_1(x,z,\theta), 
$$
where 
$$
H_1(x, z,\theta)=A \cos(\theta-\theta_0){\cos(kx)\sin(\pi z)}.
$$
Then, the vector field for the non-autonomous Hamiltonian system given in \eqref{PerHamEq}  can be transformed into the form of the vector field $X_{H}: \mathcal{M} \to T\mathcal{M}$ of the autonomous system on the extended phase space $\mathcal{M}$, which can be described by using local coordinates $(x,z,\theta)$: 
\begin{equation}\label{ExtHamEq}
\begin{split}
\frac{dx}{dt}&=-\frac{\partial H(x, z,\theta)}{\partial z}=-\frac{\partial H_0(x, z)}{\partial z}+\varepsilon \frac{\partial H_1(x, z,\theta)}{\partial z},\\[3mm]
\frac{dz}{dt}&=\frac{\partial H(x, z,\theta)}{\partial x}=\frac{\partial H_0(x, z)}{\partial x}+\varepsilon \frac{\partial H_1(x, z,\theta)}{\partial x},\\[3mm]
\frac{d\theta}{dt}&=\omega,
\end{split}
\end{equation}
and the perturbed terms are given as
\begin{equation*}
\begin{split}
\frac{\partial H_1(x, z,\theta)}{\partial z}&=- A \pi \cos (\theta-\theta_0) {\cos(kx)\cos(\pi z)},\\[2mm]
\frac{\partial H_1(x, z,\theta)}{\partial x}&=-Ak \cos(\theta-\theta_0) {\sin(kx) \sin(\pi z)}.
\end{split}
\end{equation*}

\paragraph{Poincar\'e map of the model.}
Associated with the autonomous Hamiltonian system described in \eqref{ExtHamEq}, let $\phi^\varepsilon:  \mathbb{R} \times \mathcal{M}  \to \mathcal{M}; ~ (t,x,z,\theta) \mapsto \phi^\varepsilon (x,z,\theta)$ be the flow, where $t \in \mathbb{R}$ indicates a time interval.  Hence, for some fixed $t$ and given $\varepsilon$, we define the diffeomorphism on the extended phase space as
 \begin{equation}\label{Flow}
\phi^\varepsilon_{t}: \mathcal{M} \to \mathcal{M};\; (x,z,\theta) ~\mapsto~  \phi^\varepsilon_{t}(x,z,\theta).
 \end{equation}
Let $(x(t),z(t),\theta(t))$ be an integral curve of the Hamiltonian system in \eqref{ExtHamEq}. For some fixed $\theta_0$, the angle variable $\theta(t)$ may be written as a periodic function with period $T=2\pi/\omega$ such that $\theta(t)=\theta_{0}+\omega t=\theta_{0}+2\pi t/T$. 
For each discrete time $t=kT,\; k \in \mathbb{Z}$, we can identify $\theta$ with $\theta_0+2\pi k$  and the equivalent class $[\theta]$ of $S^1$ is given by $[\theta]:=\{\theta \in S^1 \mid \theta= \theta_0+2\pi k\}$. 
Choose a representative $\theta_0$ for the equivalent class to define a Poincar\'e section $\Sigma^{\theta_{0}} $ by setting
\begin{equation}\label{PoincareSection}
\Sigma^{\theta_{0}}:=\left\{ (x,z, \theta_0) \in \mathcal{M}/S^1 \mid (x, z) \in M, \;\;\theta_0 \in [\theta]  \right\}.
\end{equation}
Then, for some fixed parameter $\varepsilon \in \mathbb{R}$, we define a Poincar\'e map $P^{\theta_{0}}_{\varepsilon}$ on $\Sigma^{\theta_{0}}$ by
\begin{equation}\label{PoincareMap}
P^{\theta_{0}}_{\varepsilon}:= \phi^\varepsilon_T \Bigr\rvert_{\Sigma^{\theta_{0}}} : \Sigma^{\theta_{0}} \to \Sigma^{\theta_{0}},
\end{equation}
which is locally given by
%
\begin{eqnarray*}
(x(kT), z(kT), \theta(kT)=\theta_{0}+2\pi k \equiv \theta_0) \hspace{60mm}\\
 \hspace{30mm} \mapsto (x ((k+1)T), z ((k+1)T), \theta((k+1)T)=\theta_{0}+2\pi(k+1) \equiv \theta_0).
\end{eqnarray*}
Note that  one special choice for $\theta_0$ may be $\theta_0=0$ and then the Poincar\'e section $\Sigma^{\theta_{0}}$ is locally isomorphic to $M \cong \mathcal{M}/S^1$. Hence, we note that a point on $\Sigma^{\theta_{0}}$ is mapped by $P^{\theta_{0}}_{\varepsilon}: \Sigma^{\theta_{0}} \to \Sigma^{\theta_{0}}$ to another point on $\Sigma^{\theta_{0}}$
during the period $T$.

\paragraph{Periodic points.}
A fixed point of the Poincar\'e map corresponds to a {\it periodic orbit} with period $T$ for the flow, and an $m$-{\it periodic point}, which corresponds to the periodic orbit with period $mT$ ($m \in \mathbb{Z}^{+}$), namely the {\it $m$-periodic orbit}, 
is the fixed point $\mathbf{x}_0 \in \Sigma^{\theta_{0}}$ such that
$$
(P^{\theta_{0}}_{\varepsilon})^{m} ( \mathbf{x}_0 ) = \mathbf{x}_0 \;\;\textrm{for}\;\; m \ge1,\;\; \textrm{while}\;\;  (P^{\theta_{0}}_{\varepsilon})^{\ell}(\mathbf{x}_0) \ne \mathbf{x}_0 \;\; \textrm{for}\;\; 1\le \ell \le m-1, ~ m \ge 2,
$$
where
$$
(P^{\theta_{0}}_{\varepsilon})^{m}=\underbrace{(P^{\theta_{0}}_{\varepsilon})\circ \cdots \circ (P^{\theta_{0}}_{\varepsilon})}_{m}.
$$

Since the Poincar\'e section $\Sigma^{\theta_0}$ is two-dimensional, it is apparent that
the Jacobian matrix of the Poincar\'e $m$-return map
\begin{equation*}
J_\varepsilon (\mathbf{x}) := \frac{\partial (P^{\theta_{0}}_{\varepsilon})^{m} (\mathbf{x})}{\partial \mathbf{x}} 
\end{equation*}
have two eigenvalues. Especially, the eigenvalues of the Jacobian matrix evaluated at an periodic point 
is called the \textit{multipliers}. Let $\mu_1$ and $\mu_2$ be the two multipliers of an $m$-periodic point 
$\mathbf{x_0} \in \Sigma^{\theta_0}$, where $|\mu_1| \leq |\mu_2|$. Since $| J_\varepsilon(\mathbf{x}_0) | =1$,
the multipliers have the product $\mu_1\mu_2 = 1$. The $m$-periodic points are classified according to the conditions of the associated multipliers as follows (see \cite{GuHo1983}): 
\begin{itemize}
\item hyperbolic: $|\mu_1|<1<|\mu_2|$
\item elliptic: $|\mu_i|=1$ but $\mu_i \neq \pm1 ~(i=1,2)$
\item parabolic: $\mu_i=\pm1 ~(i=1,2)$
\end{itemize}
The periodic orbits are stable when the associated periodic points are elliptic, 
while they are unstable when the associated ones are hyperbolic.

%
%
%
%
%
%
%
%

\paragraph{Numerical algorithm for detecting periodic points.}
Now we compute the image of the Poincar\'e map $P^{\theta_{0}}_{\varepsilon}: \Sigma^{\theta_{0}} \to \Sigma^{\theta_{0}}$ 
in order to detect periodic points, each of which corresponds to a periodic orbit in $\mathcal{M}$ through itself. 
First, we describe our numerical algorithm for detecting $m$-periodic points for some {\it fixed} amplitude $\varepsilon$
of the perturbation. Define a map $F_{\varepsilon}:\Sigma^{\theta_{0}} \to \mathbb{R}^2$ as
\begin{eqnarray*}
F_{\varepsilon}({\bf x}):={\bf x} - (P_\varepsilon^{\theta_0})^m({\bf x}),\;\;\textrm{for ${\bf x}=(x,z) \in \Sigma^{\theta_0}$}.
\end{eqnarray*}
For detecting $m$-periodic points, we shall numerically compute the kernel of the map $F_{\varepsilon}$ to find a solution ${\bf x}$ for $F_{\varepsilon}({\bf x})={\bf 0}$, where we employ Newton's method as follows:
\medskip

\begin{framed}\paragraph{\textsf{Numerical algorithm for detecting an $m$-perodic point:}\vspace{2mm}}
\begin{itemize}
\item[(1)]
Set $k=0$ with an initial approximation  ${\bf x}^{(0)}$ for the required $m$-periodic point.
\item[(2)] Set $k:=k+1$ and compute the $k$-th approximation ${\bf x}^{(k)}$ by Newton's method as
\begin{equation*}
\begin{split}
{\bf x}^{(k)} &:= {\bf x}^{(k-1)} - \left( \frac{\partial F_{\varepsilon}({\bf x})}{\partial \mathbf{x}} \Biggr|_{\mathbf{x}=\mathbf{x}^{(k-1)}} \right) ^{-1} F_{\varepsilon}({\bf x}^{(k-1)}),
\end{split}
\end{equation*}
where
\begin{equation*}
\frac{\partial F_{\varepsilon}({\bf x})}{\partial \mathbf{x}} \Biggr|_{\mathbf{x}=\mathbf{x}^{(k-1)}} = \mathbf{I} - J_\varepsilon (\mathbf{x}^{(k-1)}).
\end{equation*}
Here, $\mathbf{I}$ is the unit matrix and the Jacobian matrix
\begin{equation*}
J_\varepsilon (\mathbf{x}^{(k-1)}) = \frac{\partial (P^{\theta_{0}}_{\varepsilon})^{m} (\mathbf{x})}{\partial \mathbf{x}} \Biggr|_{\mathbf{x}=\mathbf{x}^{(k-1)}}
\end{equation*}
is numerically obtained by using the central difference scheme. 
\item[(3)]
If $|F_{\varepsilon}({\bf x}^{(k)})| < \delta$, where the convergence radius is set to $\delta=10^{-10}$, then 
the computation ends up and the $m$-periodic point is to be detected as ${\bf x}={\bf x}^{(k)}$. 
\item[(4)] Otherwise, return to (2) in order to iterate the computation until convergence.
\end{itemize} 
\end{framed}


\begin{remark}\rm
Since the approximation value of the periodic points are unknown, 
we cover the Poincar\'e section with a small grid spacing and set each grid point as the initial condition ${\bf x}^{(0)}$. In our computation the grid spacing is set to $0.005$. The Poincar\'e maps are computed with 7th-order Runge-Kutta method 
with double precision floating point, which are the same through this study. 
\end{remark}



\paragraph{Periodic points at $\varepsilon=0.1$.}
Let us consider to detect the periodic points for the case $\varepsilon=0.1$. For numerical computations, throughout the paper, we fix other parameters of the convection to $A=\pi,k=\pi$ and $T=1/\pi$.
Now we illustrate in Fig.\ref{fig:pm_prd} the image of the Poincar\'e section by the Poincar\'e map 
and the detected periodic points in a cell which range from $x=0$ to $x=1(=\pi/k)$, 
where the elliptic and hyperbolic periodic points with period $m \leq 15$ are depicted. 
The color and the shape of the plots denote the period $m$ and the symbols of plots, i.e., $\bullet$ and $\star$, 
indicate elliptic and hyperbolic respectively. 
The number of recurrences due to the Poincar\'e map is set to $N=1000$ and the initial condition for $\theta$ is $\theta_0=0$. 
Note that there is no loss of generality to investigate only one single cell, 
since there is a topological isomorphism among cells. 
The left figure in Fig.\ref{fig:pm_prd} shows an enlarged view of the squared section in the right figure of Fig.\ref{fig:pm_prd}.
In conjunction with symmetry, it is observed that the periodic points appear symmetrically with respect to $z=1/2$, 
which is consistent with the symmetric property i) of the non-autonomous system in \eqref{PerHamEq}.

\paragraph{The periodic points and KAM curves.}
It is apparent from the Poincar\'e map that there exists one large island in the middle of the cell 
which we denote by label $I_1$, while there are three small islands surrounding the main island $I_1$, 
each of which is respectively denoted by labels $I_2$, $I_3$, and $I_4$ as in Fig.\ref{fig:pm_prd}. 
As is well known, inside the islands, there exist quasi-periodic points, 
while outside the islands there is a chaotic sea where points correspond to chaotic orbits. 
We can see that the elliptic and hyperbolic periodic points inside the islands appear alternately along the KAM curves in the perturbed Hamiltonian systems as is well known; see \cite{GuHo1983} and \cite{DoOt1988}.

In particular, it is observed in Fig.\ref{fig:pm_prd} that the elliptic periodic points appear at the center of islands, which is surrounded by KAM curves. For example, the elliptic 3-periodic points exist at the center of islands $I_2$, $I_3$, and $I_4$ in Fig.\ref{fig:pm_prd}, and the elliptic 5, 7, 8, and 13-periodic points appear at the center of the small islands in $I_1$.  
The relation between the elliptic periodic points and the islands will be discussed in detail in \S\ref{subsec:kam}.
In contrast, it is observed that the hyperbolic periodic points appear in the chaotic regions. 
This is because the stable and unstable manifolds associated with the hyperbolic periodic points form complicated homoclinic tangles around them and the points in the neighborhood are to be transported chaotically. Further, we note that some of the elliptic and hyperbolic periodic points do not appear as mentioned above, since not all of the islands and chaotic regions can be numerically detected in Fig.\ref{fig:pm_prd}. Especially, the chaotic regions between KAM curves in the islands cannot be observed in details.

\begin{figure}[H]
\begin{center}
\includegraphics[scale=0.3]{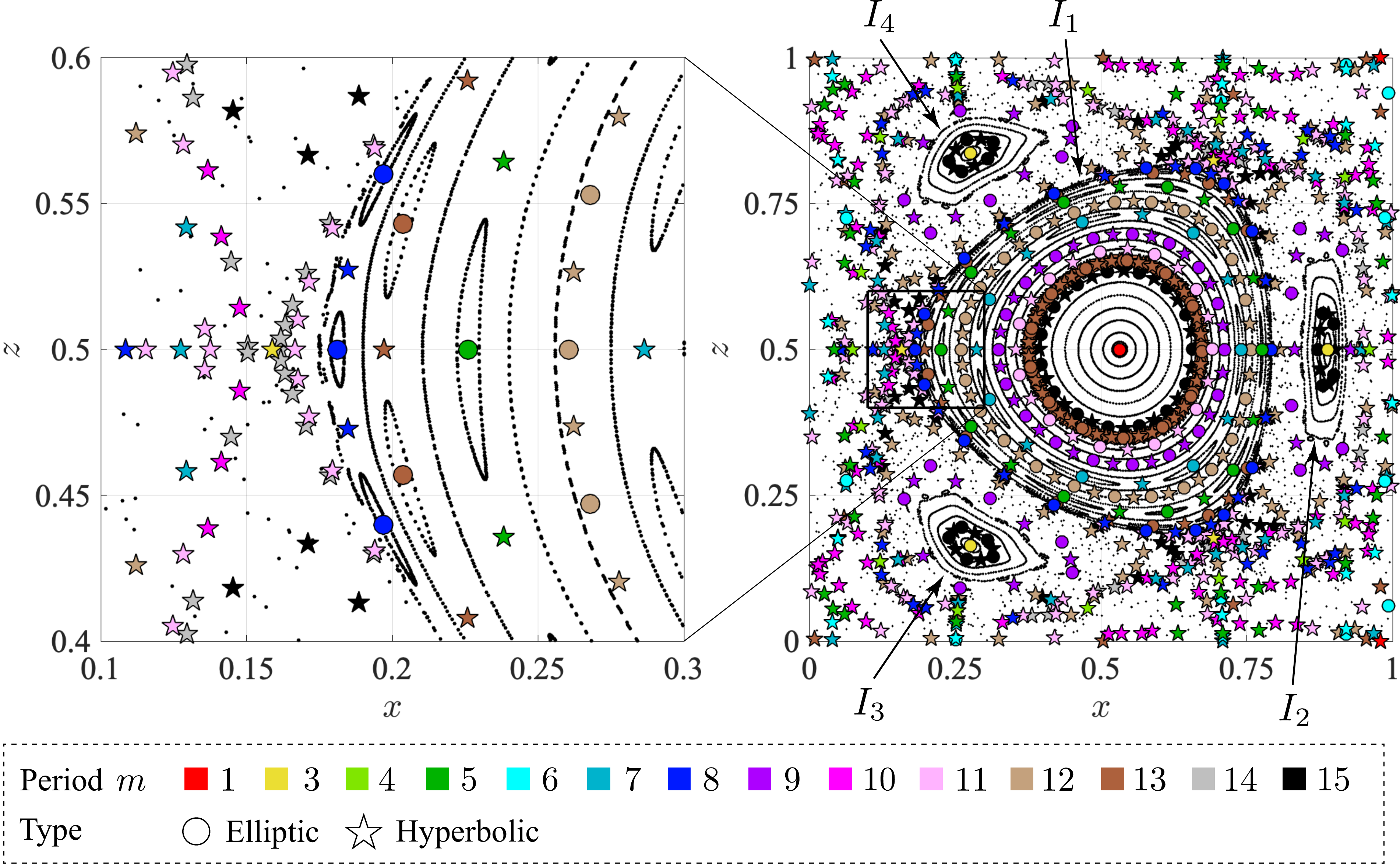}
\caption{Structure of elliptic and hyperbolic periodic points ($\varepsilon=0.1$)}
\label{fig:pm_prd}
\end{center}
\end{figure}

\subsection{Structures of periodic orbits and KAM tori.}\label{subsec:kam}
As we have shown in Fig.\ref{fig:pm_prd}, the elliptic periodic points appear at the center of the islands of KAM tori. 
In this subsection, we investigate the structures of periodic orbits and KAM tori
in the extended phase space $\mathcal{M}=M \times S^1$, which are associated with elliptic periodic points.
Here, we especially focus on those associated with the elliptic 3-periodic points at the center of islands 
$I_2$, $I_3$, and $I_4$ in Fig.\ref{fig:pm_prd}.

\paragraph{Twisted structures of periodic orbits and KAM tori.}
Fig.\ref{fig:prd_point_MP3} illustrates the elliptic 3-periodic points at the center of islands $I_2$, $I_3$, and $I_4$
on the Poincar\'e section $\Sigma^{\theta_0}$. 
Their 3-periodic orbit and the associated KAM torus in the extended phase space $\mathcal{M}$ 
are shown in Fig.\ref{fig:torus} in yellow and blue respectively. 
The Poincar\'e section $\Sigma^{\theta_0}$ given in \eqref{PoincareSection} is depicted in gray, 
where it is restricted to $U \times U \subset M$ and where  we choose $\theta_0=0$ for $[\theta]=\theta_0 +2\pi k$. 
The intersection of the KAM torus and the Poincar\'e section $\Sigma^{\theta_0}$ corresponds to the 
KAM curve of the island, and those of the periodic orbit 
$\widetilde c \in \mathcal{M}$ and $\Sigma^{\theta_0}$ corresponds to the elliptic 3-periodic points.
It is apparent that the periodic orbit and the associated KAM tori for the 3-periodic points are connected 
with each other and thus they globally have a twisted structure. 
Generally, this implies that KAM tori for elliptic periodic points whose period is more than two have twisted structures in the extended phase space $\mathcal{M}$ and also that the orbit of the elliptic periodic points goes through the center of it. 
Note that such KAM torus do not appear around the orbits of hyperbolic periodic points. 


\begin{figure}[H]
\begin{center}
\subfigure[Elliptic 3-peirodic points on $\Sigma^{\theta_0}$]
{\includegraphics[scale=0.25]{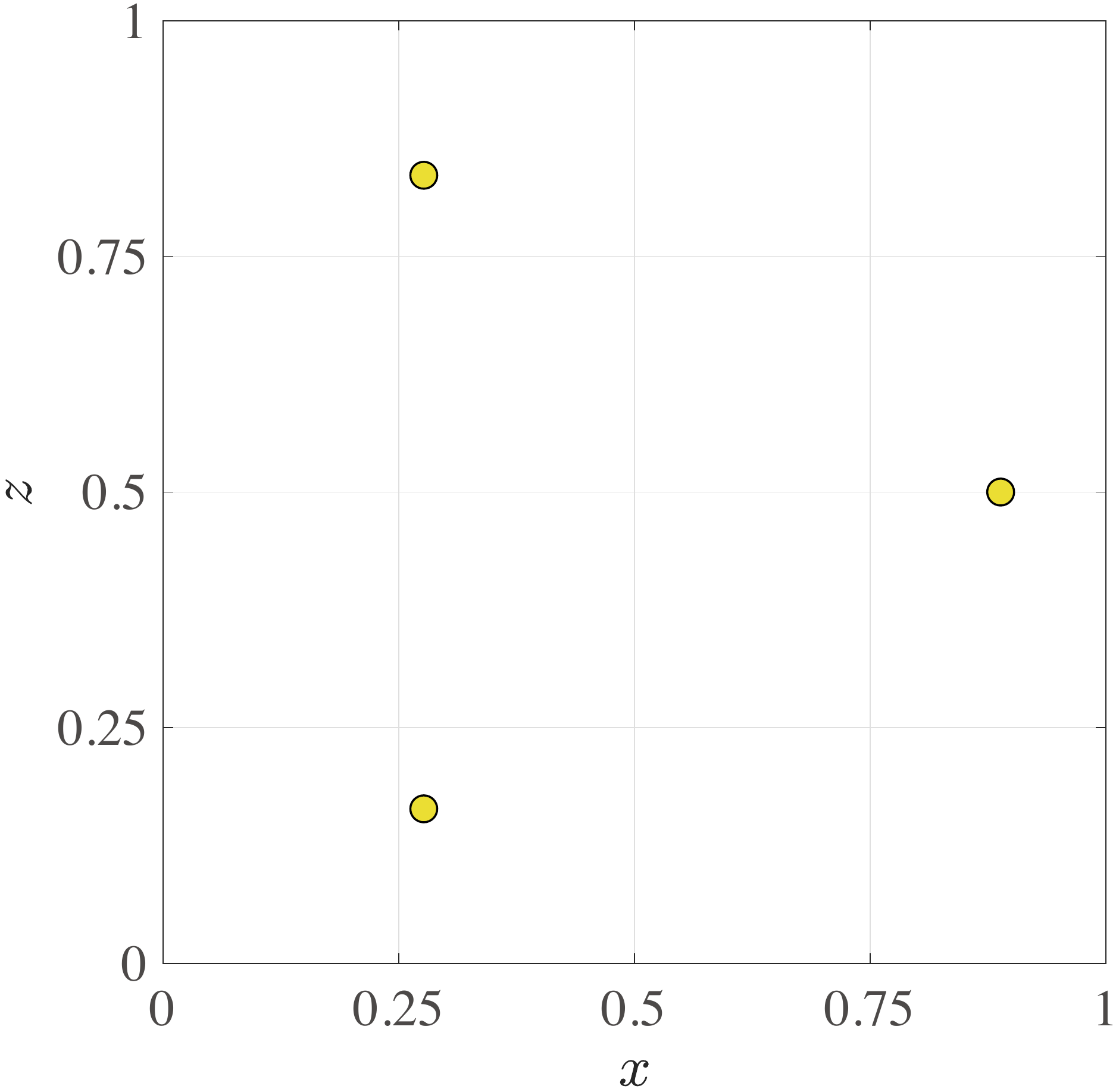}
\label{fig:prd_point_MP3}}
\hspace{13mm}
\subfigure[3-periodic orbit and KAM torus in $\mathcal{M}$]
{\includegraphics[scale=0.26]{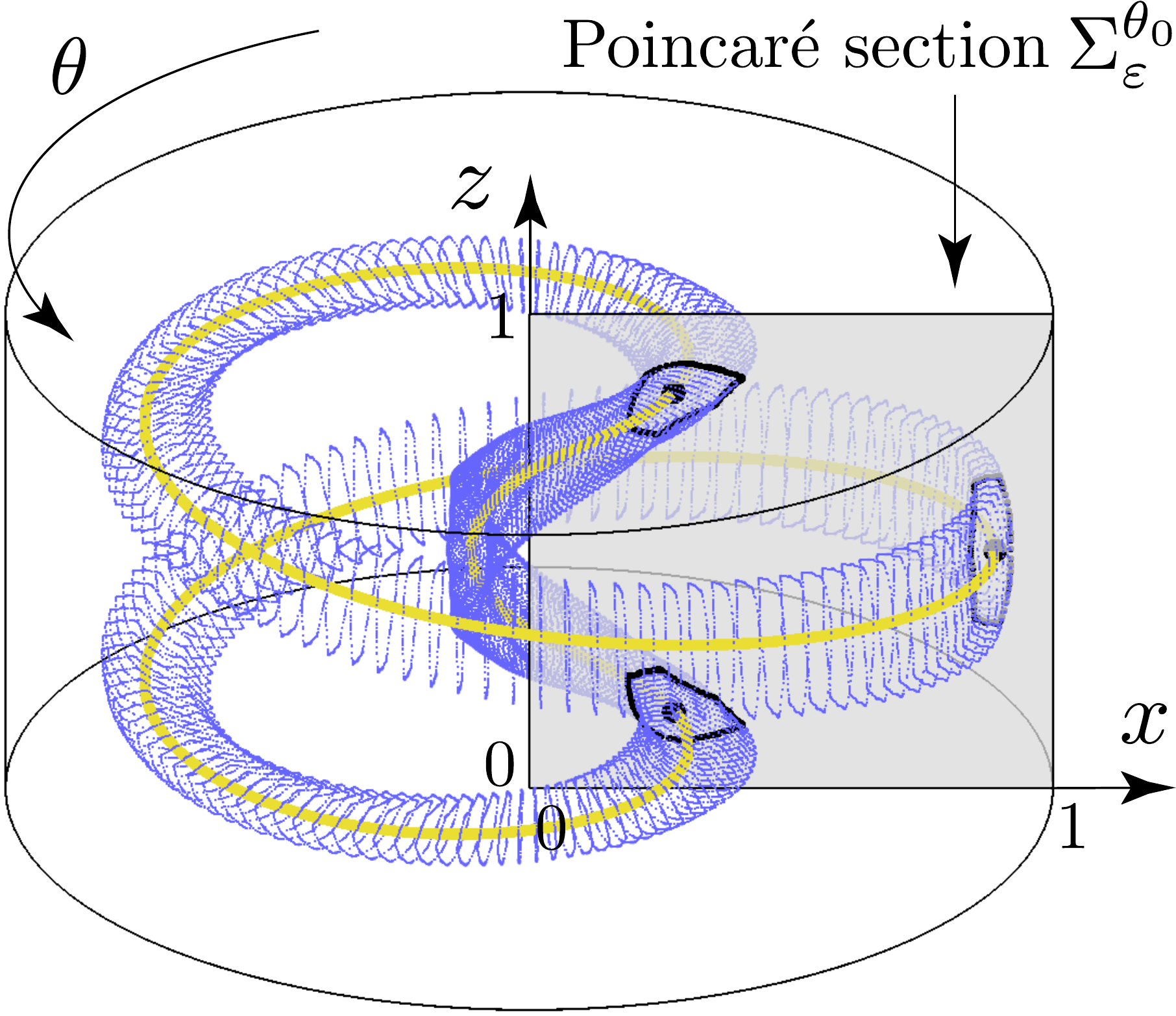}
\label{fig:torus}}
\caption{Elliptic 3-periodic points and their orbit with the associated KAM torus}
\label{fig:prd_torus}
\end{center}
\end{figure}

\paragraph{Periodic transport of islands.}
Since we have seen in Fig.\ref{fig:torus} that the KAM tori for each island are connected with each other, 
we next investigate the images of the island regions by Poincar\'e map $P^{\theta_0}_\varepsilon$. 
Let us denote the closed regions of island $I_i$ as $R_i \subset U \times U$ for $i=2,3,$ and $4$. 
Fig.\ref{fig:island_map} shows the initial position and the image of the regions mapped by $P^{\theta_0}_\varepsilon$. 
In order to easily recognize the deformation of the regions, each of them is illustrated in four colors. 
The elliptic 3-periodic points are indicated in yellow plots.
We can see that the regions of $I_2$, $I_3$, and $I_4$ are mapped to $I_3$, $I_4$, and $I_2$ respectively in order with the 3-periodic points as
\begin{equation*}
P^{\theta_0}_\varepsilon(R_2)=R_3,~~~P^{\theta_0}_\varepsilon(R_3)=R_4,~~~P^{\theta_0}_\varepsilon(R_4)=R_2.
\end{equation*}
It follows that the region of each island is mapped to the same island after three times of Poincar\'e maps as
\begin{equation*}
(P^{\theta_0}_\varepsilon)^3(R_i)=R_i.
\end{equation*}
Of course, this implies that the region $R$ of an island associated with an $m$-periodic point is mapped to the same island 
after $m$ times of Poincar\'e maps as
\begin{equation*}
(P^{\theta_0}_\varepsilon)^m(R)=R.
\end{equation*}

Furthermore, Fig.\ref{fig:island_map} indicates that the regions of the islands rotate around the 3-periodic points 
when they are mapped. It follows from the physical point of view that fluid in the region of an island is 
transported periodically as a sort of vortex by the Lagrangian transport as a whole, though each point is transported 
quasi-periodically. The KAM curve around the region, which is an invariant manifold, seems to act as a barrier and 
enclose the fluid inside. Notice that these vortex structures do not appear in a vortex field in the Eulerian description. 
It seems that these structures are quite relevant with the "Lagrangian vortices" or "Lagrangian eddies", 
which are regions that are transported stably as rotating regions; see \cite{HaBe2013}, {\cite{BlHa2014}, and \cite{FaHa2016}. However, we will seek for the relevance with Lagrangian vortices in details in future works.

\begin{figure}[H]
\begin{center}
\subfigure[Initial position of the regions]
{\includegraphics[scale=0.25]{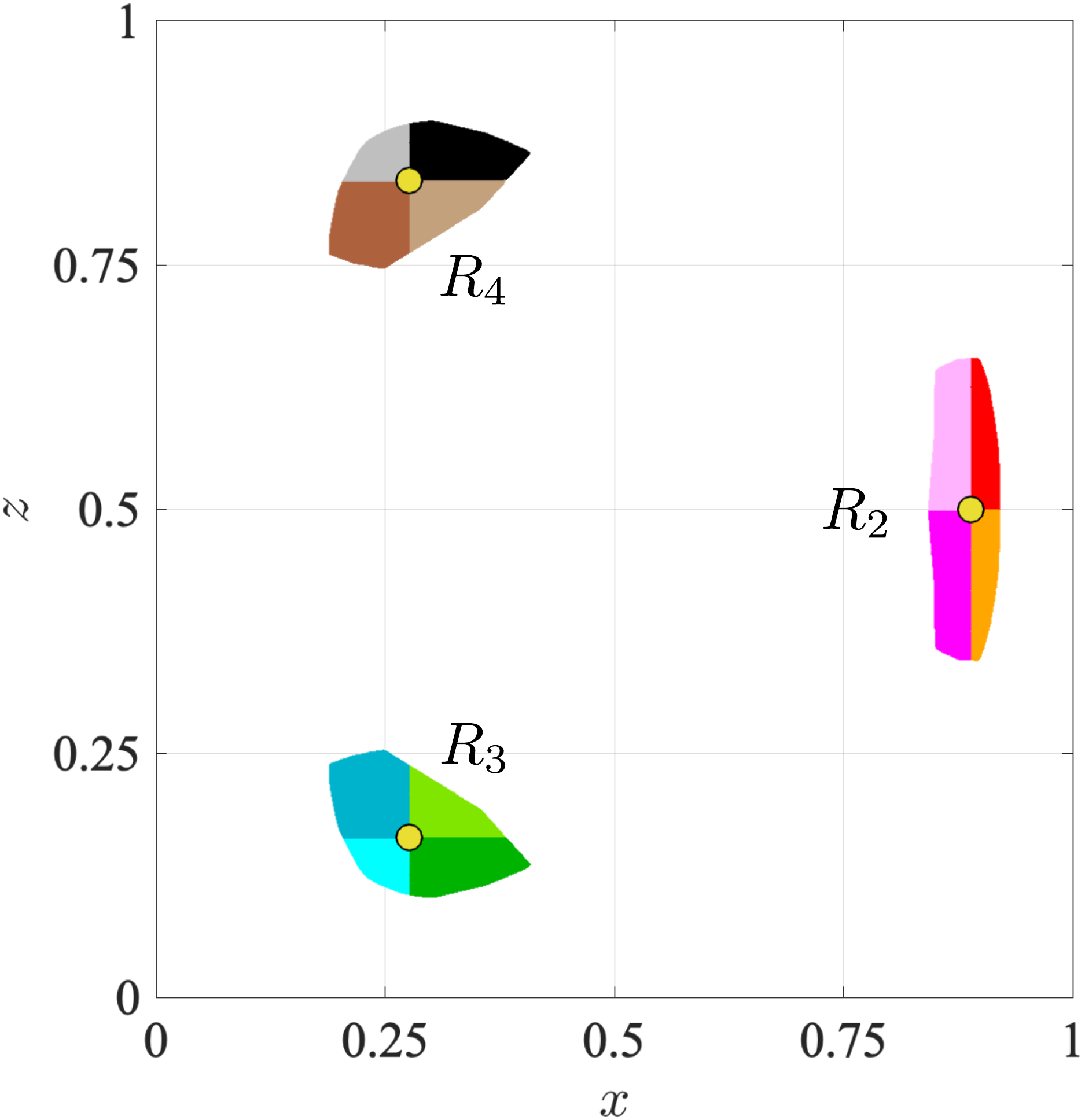}}
\hspace{10mm}
\subfigure[Position of the mapped regions]
{\includegraphics[scale=0.25]{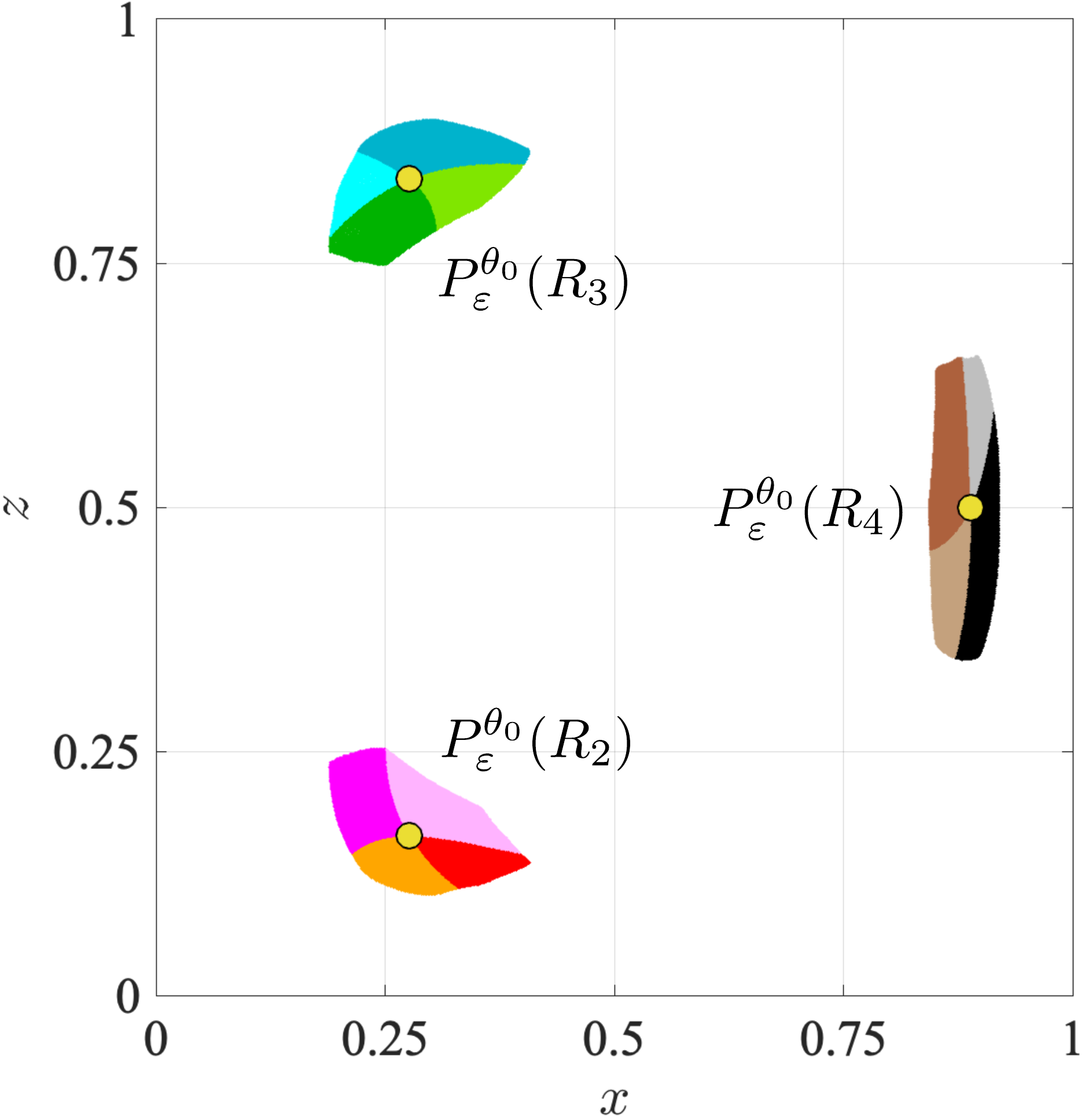}}
\caption{Mapping of the regions of the islands by Poincar\'e map}
\label{fig:island_map}
\end{center}
\end{figure}

\section{Resonances and symmetries of periodic orbits}
\label{Sec:symmetry}
In this section, we investigate the resonances and symmetric properties of periodic orbits which is a solution curve passing through periodic points. To do this, we consider the projection of the $m$-periodic orbits in the extended phase space $\mathcal{M}=M \times S^1$ to the original phase space $M$ and analyze the winding number $n$ of the projected orbits around the center of a cell. 

\subsection{Resonances of periodic orbits}
\paragraph{Periodic solutions.}
Let us investigate the resonance of periodic orbits by introducing a projection. Let $\widetilde{c}(t):=(x(t),z(t),\theta(t)),\,t \in \mathbb{I} \subset \mathbb{R}$ be a periodic solution of the perturbed Hamiltonian system in \eqref{ExtHamEq}, which is given by a curve on the extended phase space $\mathcal{M}$, and let $\pi: \mathcal{M} \to M; (x,z,\theta) \mapsto (x,z)$ be the natural projection. Then, from the periodic solution $\widetilde{c}(t)$, the projected curve $c(t)$ can be defined on $M$ as 
$$
c(t):=\pi(\widetilde{c}(t))=(x(t),z(t)),
$$
which can be identified with the solution curve of the non-autonomous Hamiltonian system in \eqref{PerHamEq} on $M$. 

Fig.\ref{fig:prd_orbit_MP3} illustrates the projection of the 3-periodic orbit in Fig.\ref{fig:torus} by $\pi$ onto $M$. 
It follows that the projection is a closed orbit and that it goes around the center of the cell $(x,z)=(1/2,1/2)$ once.

\paragraph{Winding number of periodic orbits.}
In order to analyze the number of times that a projected periodic orbit goes around the center of a cell, 
let us introduce the concept of {\it winding number} $n$ of a projected orbit. 

\begin{figure}[H]
\begin{center}
\includegraphics[scale=0.25]{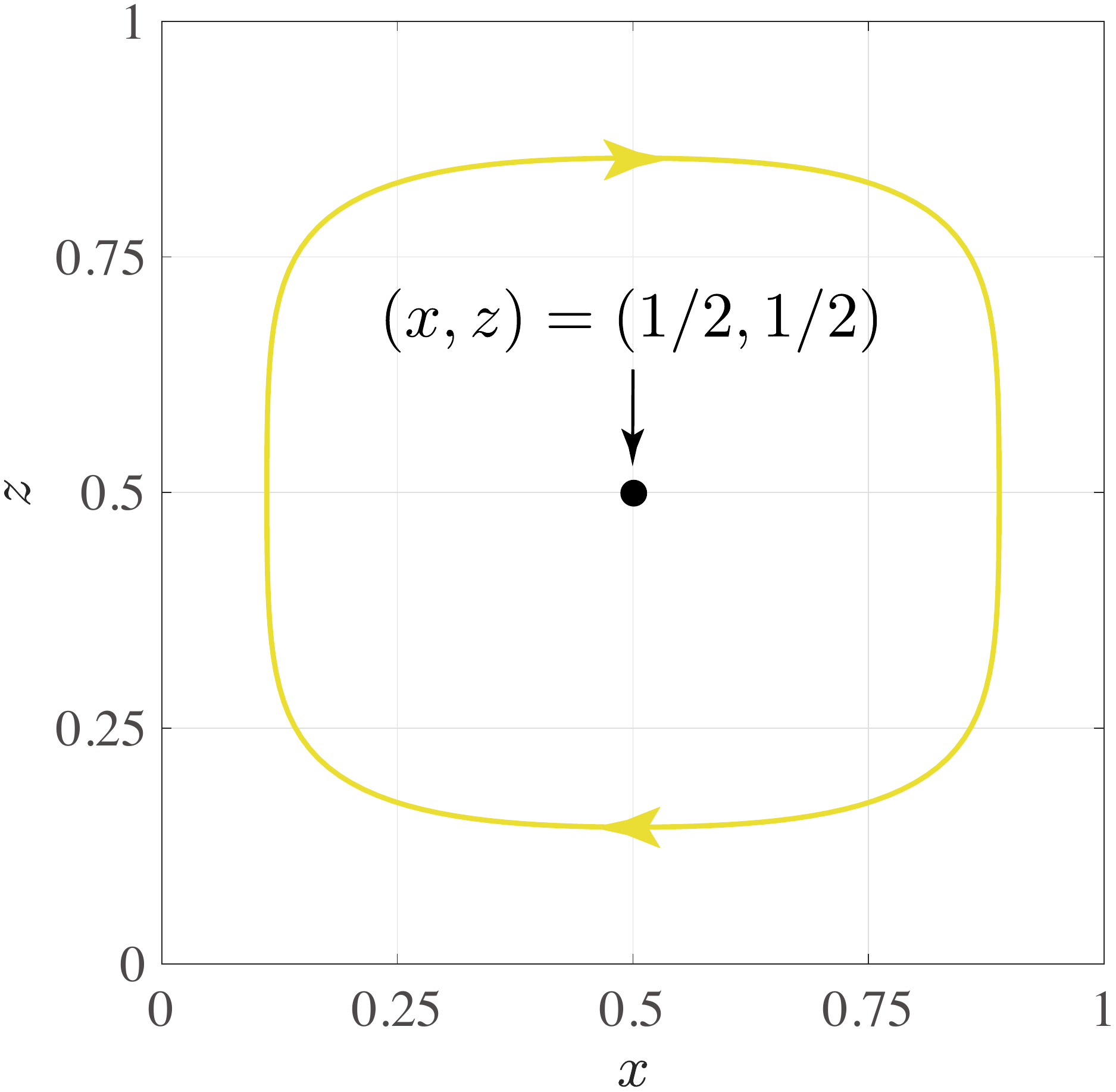}
\vspace{-2mm}
\caption{3-periodic orbit projected by $\pi$ on $M$}
\label{fig:prd_orbit_MP3}
\end{center}
\end{figure}

\vspace{-5mm}

\begin{definition}\rm
Consider an $m$-periodic orbit $\widetilde{c}(t):=(x(t),z(t),\theta(t)), t \in\mathbb{I} \subset \mathbb{R}$ on the extended phase space $M \times S^1$. Then we can define the periodic curve on $M$ by $c(t)=(x(t), z(t)):=\pi (\tilde{c}(t))$. Then, the {\it winding number} of $c(t)$ is given by
\begin{equation}\label{eq:definition_n}
n= \frac{1}{2\pi i} \oint_c \frac{dw}{w-w_c},
\end{equation}
where $w=x+iz \in \mathbb{C}$ is a point on $c(t)$ and $w_c=x_c+iz_c \in \mathbb{C}$ is 
a point on $M$ such that $w_c \notin c(t)$. Regarding the winding number, see \cite{Fl1983}.
\end{definition}

The absolute value of the winding number $n$ corresponds to the number of times that the orbit goes around the center of a cell, while it could take both positive and negative values in general according to the direction. Namely, the winding number is positive when the orbit goes in counter-clockwise direction,
while it is negative when it goes in clockwise direction. 
For example, the winding number of the projection of the 3-periodic orbit shown in Fig.\ref{fig:prd_orbit_MP3} 
is $n=-1$ when $x_c=z_c=1/2$, since the orbit goes around $(x,z)=(1/2,1/2)$ once in clockwise direction. 

\vspace{-1mm}

\paragraph{Resonant periodic orbits.}
Fig.\ref{fig:prd_point_orbit_MP1_3_5_7_11} -- Fig.\ref{fig:prd_point_orbit_MP5_11} illustrate 
some of the periodic points in Fig.\ref{fig:pm_prd} and the projection of the associated periodic curves onto $M$, where they are classified according to the symmetry with respect to the horizontal and 
vertical center lines of the cell, namely $x=\pi/(2k)$ and $z=1/2$.
As can be seen, the projected orbits go around the center of the cell once or several times. 
For example, the projection of the 7-periodic orbit in Fig.\ref{fig:prd_point_orbit_MP1_3_5_7_11} goes around the center 
of the cell three times in clockwise direction, which means that the winding number is $n=-3$ when $x_c=z_c=1/2$. 
It follows that $m$-periodic orbits can be considered as {\it resonant orbits} in the sense that those with winding number $n$ 
go around the center of a cell $n$ times when they are projected on $M$, 
while they go around $m$ times in $\theta$ direction in the extended phase space $\mathcal{M}$.

\vspace{-1mm}

\paragraph{Resonance condition of periodic orbits.}
Let us define the {\it resonance condition} of an $m$-periodic orbit with winding number $n$ as $|n/m|$.
The resonance conditions of the detected periodic orbits are indicated besides each orbit 
in Fig.\ref{fig:prd_point_orbit_MP1_3_5_7_11} - Fig.\ref{fig:prd_point_orbit_MP5_11}, where $x_c=z_c=1/2$. 
Fig.\ref{fig:prd_point_orbit_MP1_3_5_7_11} - Fig.\ref{fig:prd_point_orbit_MP5_11} shows that there are many kinds of 
periodic orbits with different resonance conditions. It follows that the resonance conditions of the orbits 
in the middle of the cell tend to be larger than that of those in the outer area, since the absolute value of the winding number of 
those in the middle tend to be larger. It is also observed that some of the $m$-periodic orbits have different resonance conditions
even when their periods are the same. For example, we can see two different kinds of 7-periodic orbits with $|n/m|=3/7$ and 2/7, 
and also 11-periodic orbits with $|n/m|=5/11$ and 3/11. 

\begin{figure}[H]
\begin{center}
\subfigure[Peirodic points on $\Sigma^{\theta_0}$]
{\includegraphics[scale=0.25]{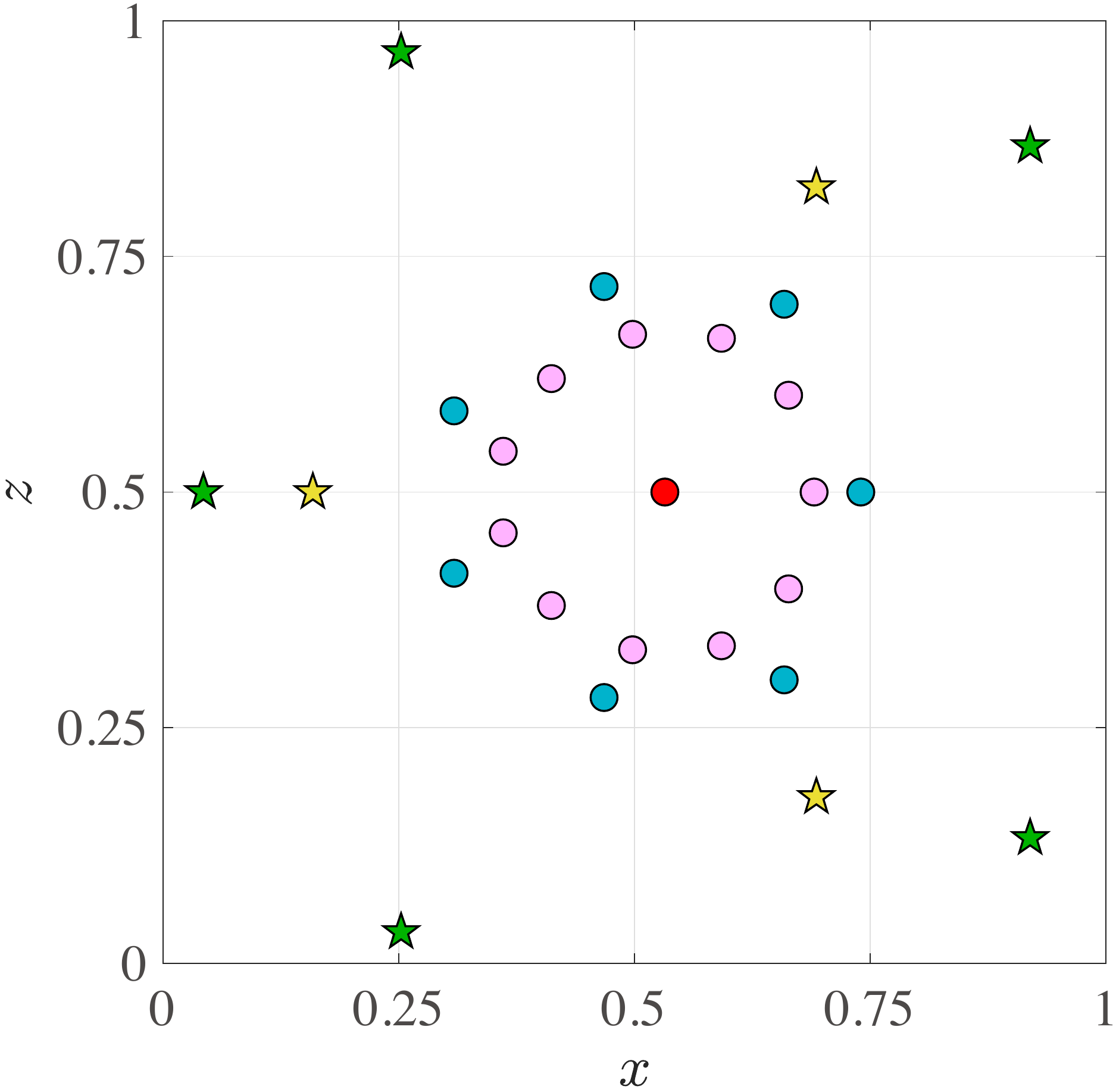}
\label{fig:prd_point_MP1_3_5_7_11}}
\hspace{5mm}
\vspace{-1mm}
\subfigure[The projection of periodic orbits onto $M$]
{\includegraphics[scale=0.25]{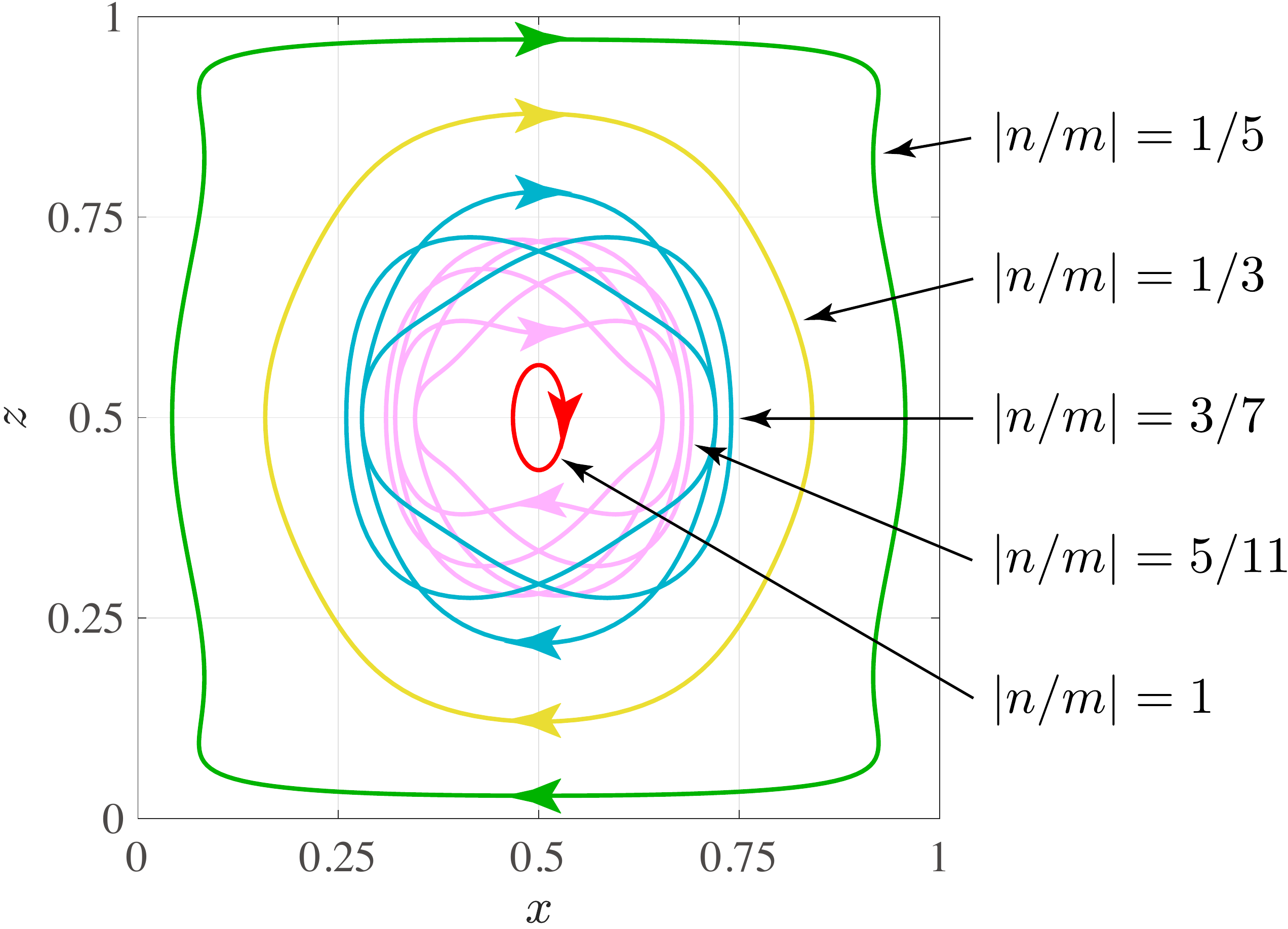}
\label{fig:prd_orbit_MP1_3_5_7_11}}
\subfigure{\includegraphics[scale=0.3]{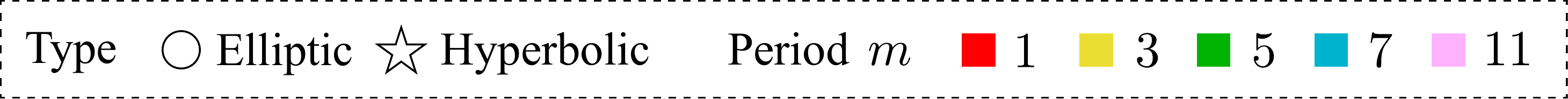}}
\vspace{-3mm}
\caption{Orbits symmetric with respect to $x=\pi/(2k)$ and $z=1/2$}
\label{fig:prd_point_orbit_MP1_3_5_7_11}
\end{center}
\end{figure}

\vspace{-7mm}

\begin{figure}[H]
\begin{center}
\subfigure[Peirodic points on $\Sigma^{\theta_0}$]
{\includegraphics[scale=0.25]{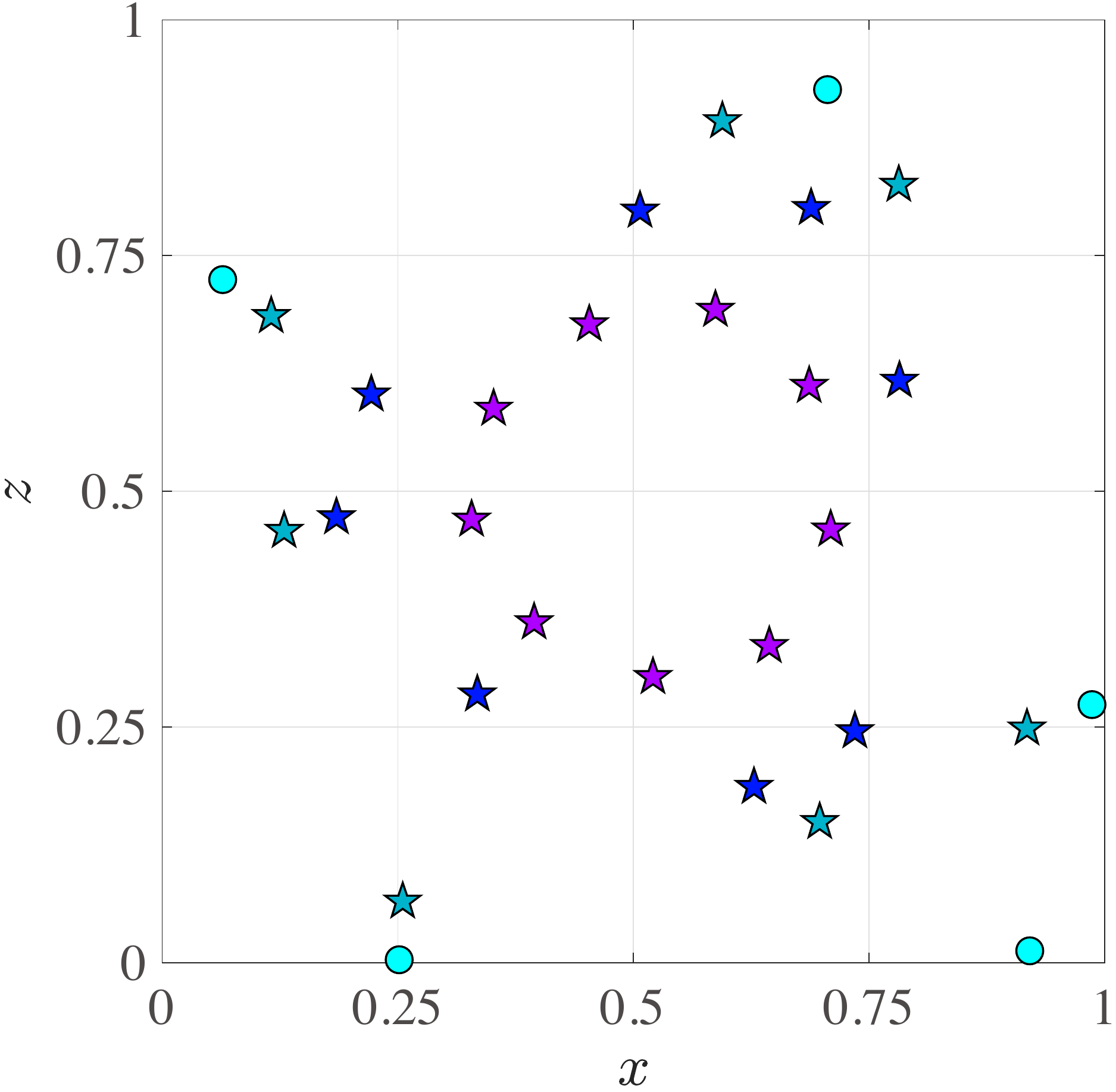}
\label{fig:prd_point_MP6_7_8_9}}
\hspace{5mm}
\vspace{-1mm}
\subfigure[The projection of periodic orbits onto $M$]
{\includegraphics[scale=0.25]{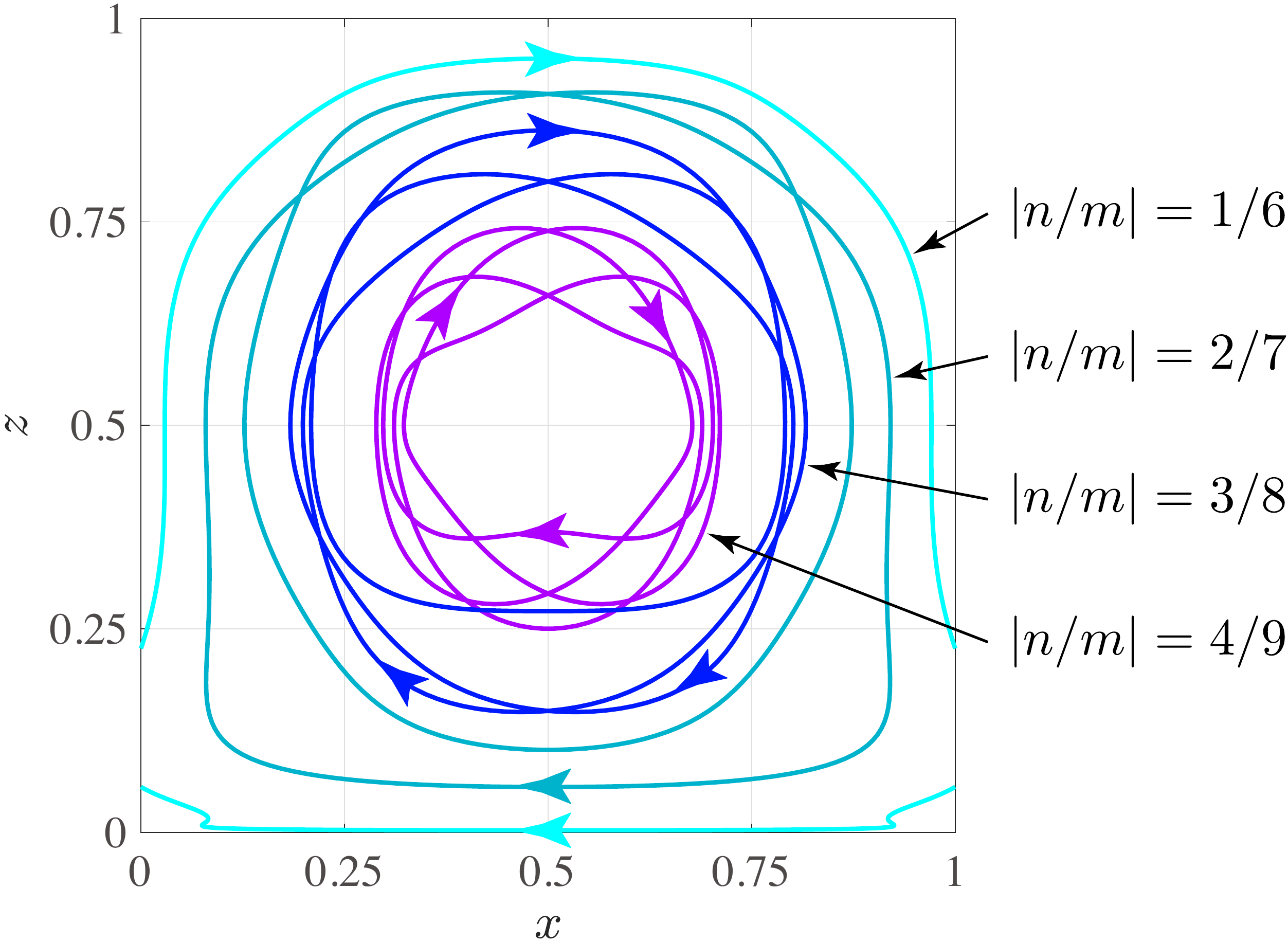}
\label{fig:prd_orbit_MP6_7_8_9}}
\subfigure{\includegraphics[scale=0.3]{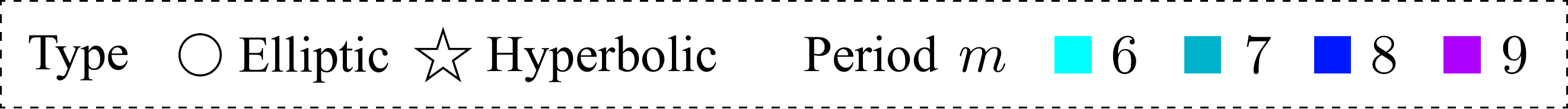}}
\vspace{-3mm}
\caption{Orbits symmetric only with respect to $x=\pi/(2k)$}
\label{fig:prd_point_orbit_MP6_7_8_9}
\end{center}
\end{figure}

\vspace{-7mm}

\begin{figure}[H]
\begin{center}
\subfigure[Peirodic points on $\Sigma^{\theta_0}$]
{\includegraphics[scale=0.25]{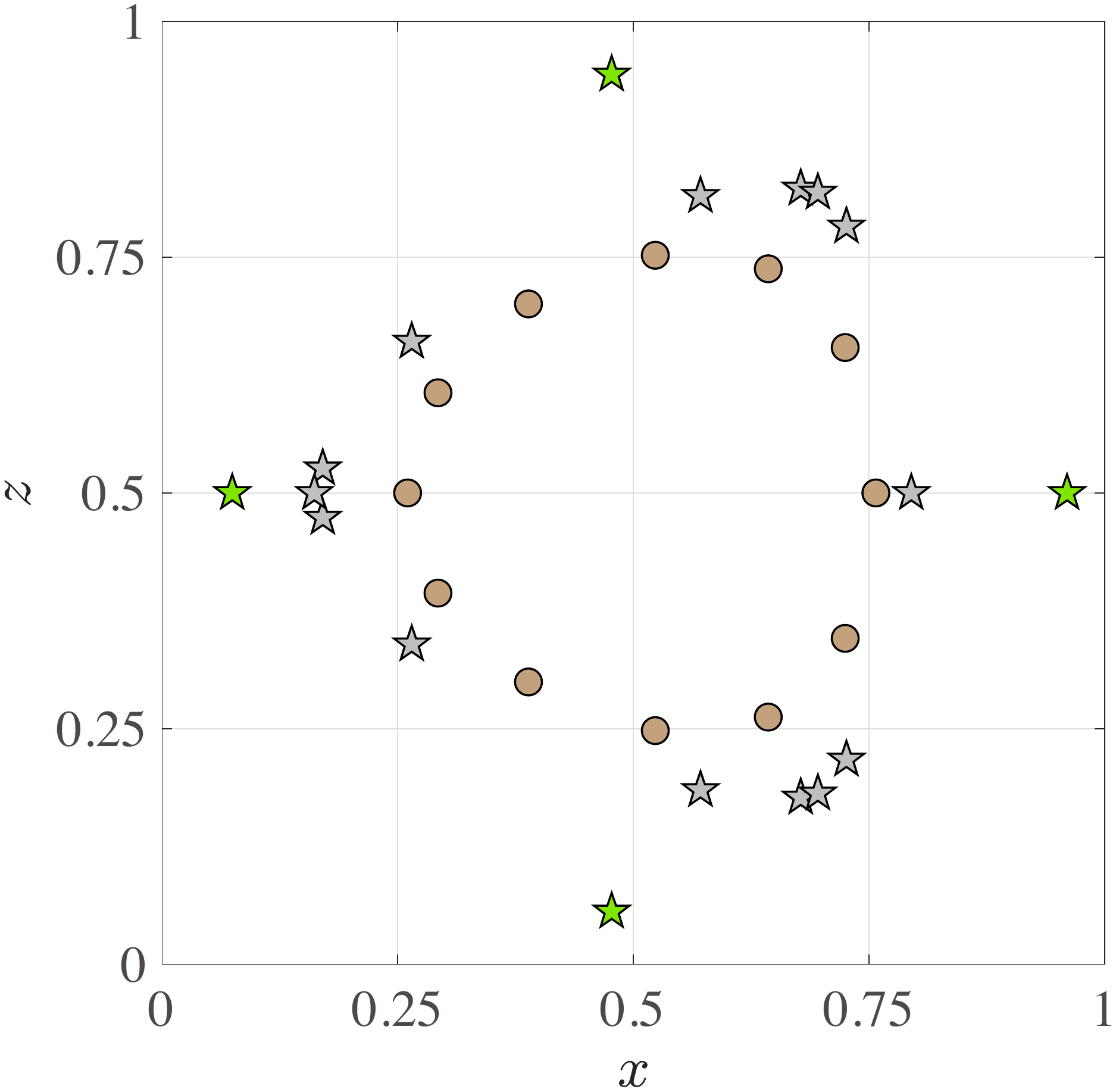}
\label{fig:prd_point_MP4_12_14}}
\hspace{5mm}
\vspace{-1mm}
\subfigure[The projection of periodic orbits onto $M$]
{\includegraphics[scale=0.25]{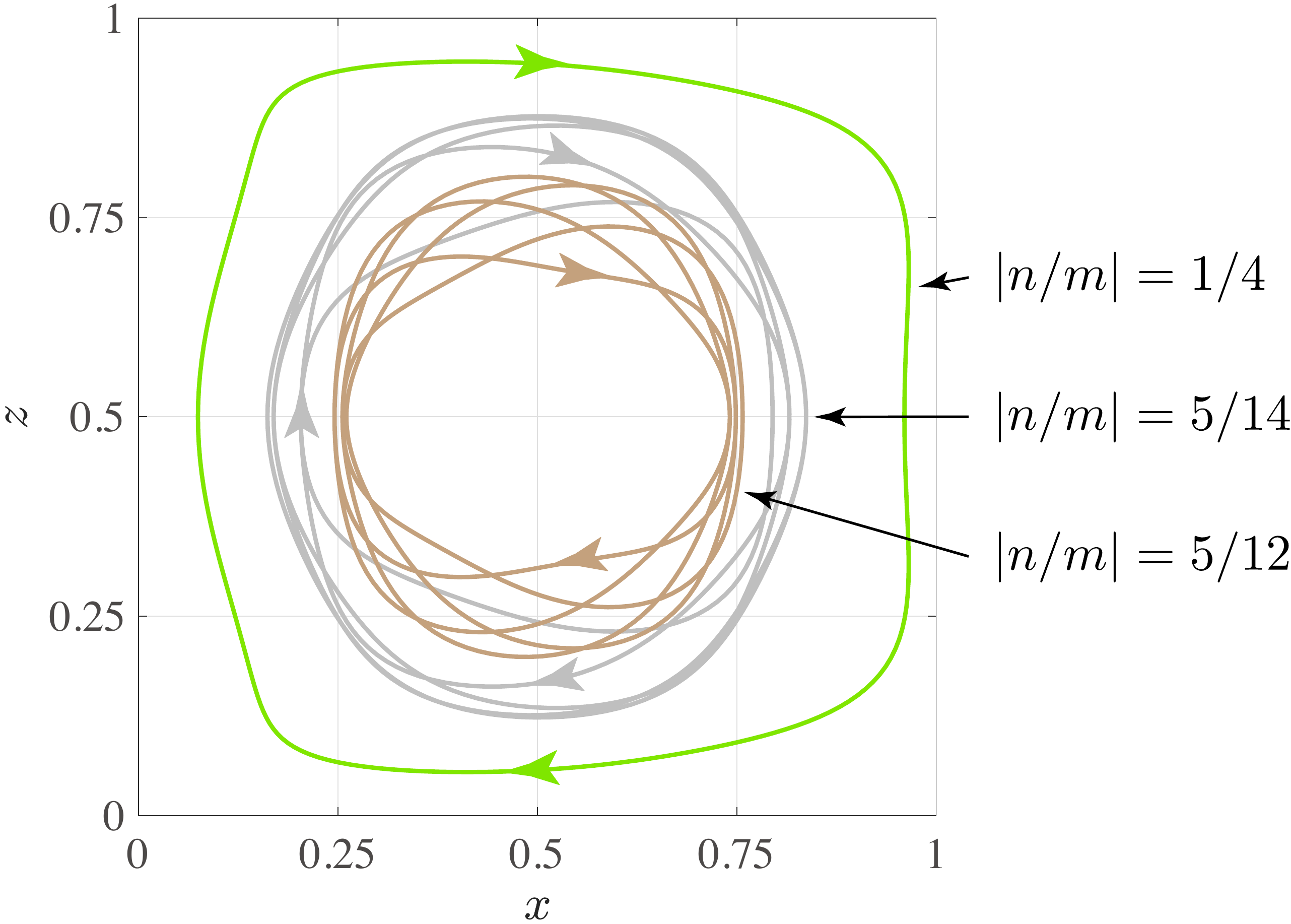}
\label{fig:prd_orbit_MP4_12_14}}
\subfigure{\includegraphics[scale=0.3]{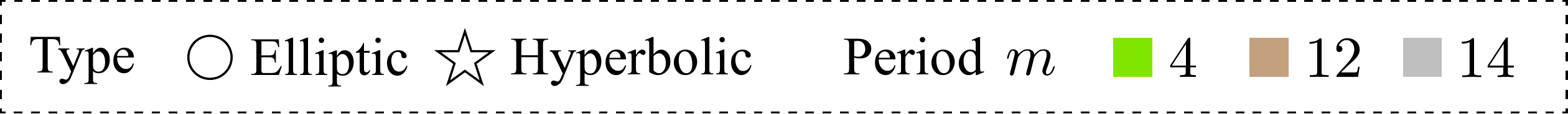}}
\vspace{-3mm}
\caption{Orbits symmetric only with respect to $z=1/2$}
\label{fig:prd_point_orbit_MP4_12_14}
\end{center}
\end{figure}

\begin{figure}[H]
\begin{center}
\subfigure[Peirodic points on $\Sigma^{\theta_0}$]
{\includegraphics[scale=0.25]{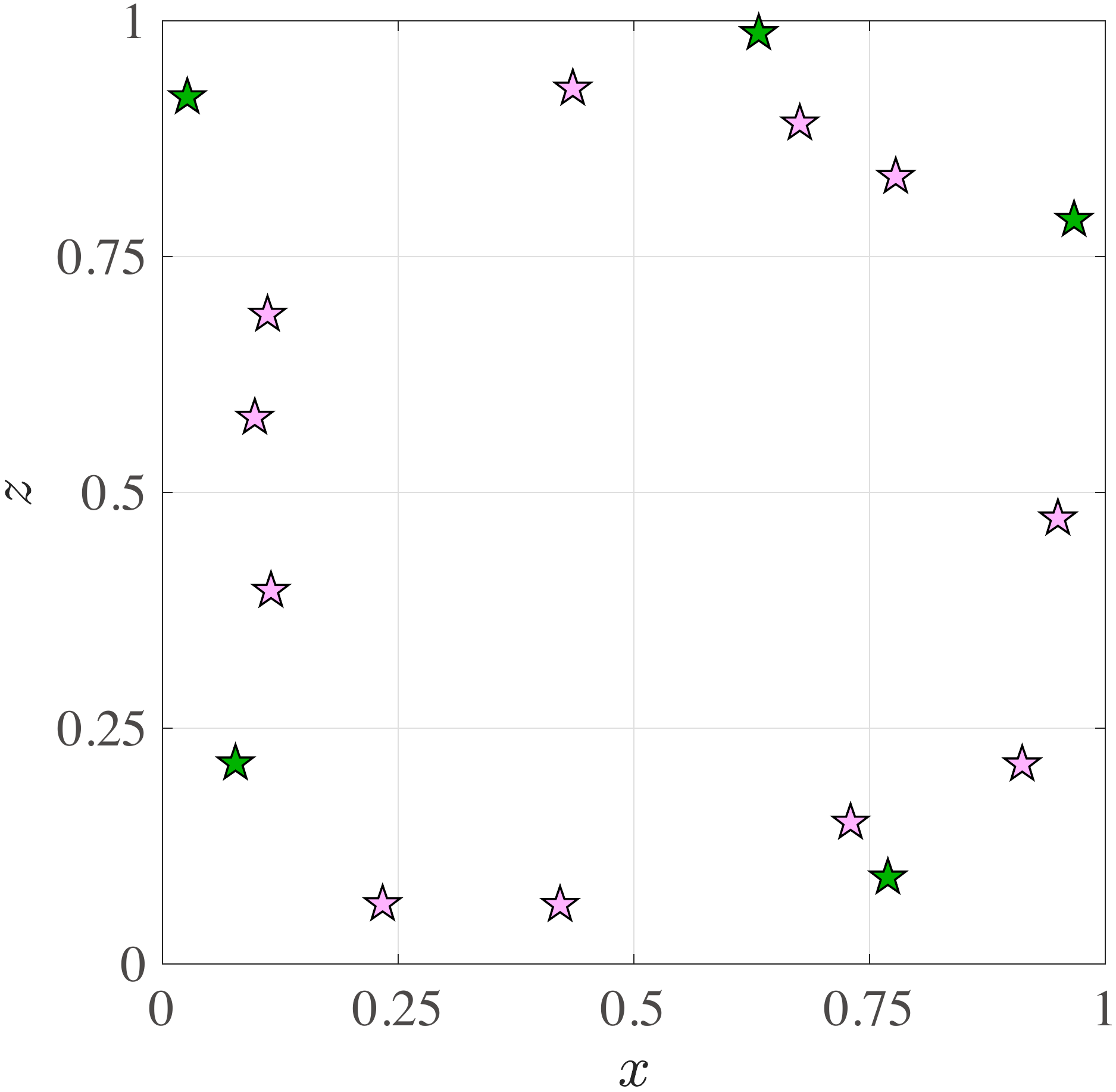}
\label{fig:prd_point_MP5_11}}
\hspace{5mm}
\subfigure[The projection of periodic orbits onto $M$]
{\includegraphics[scale=0.25]{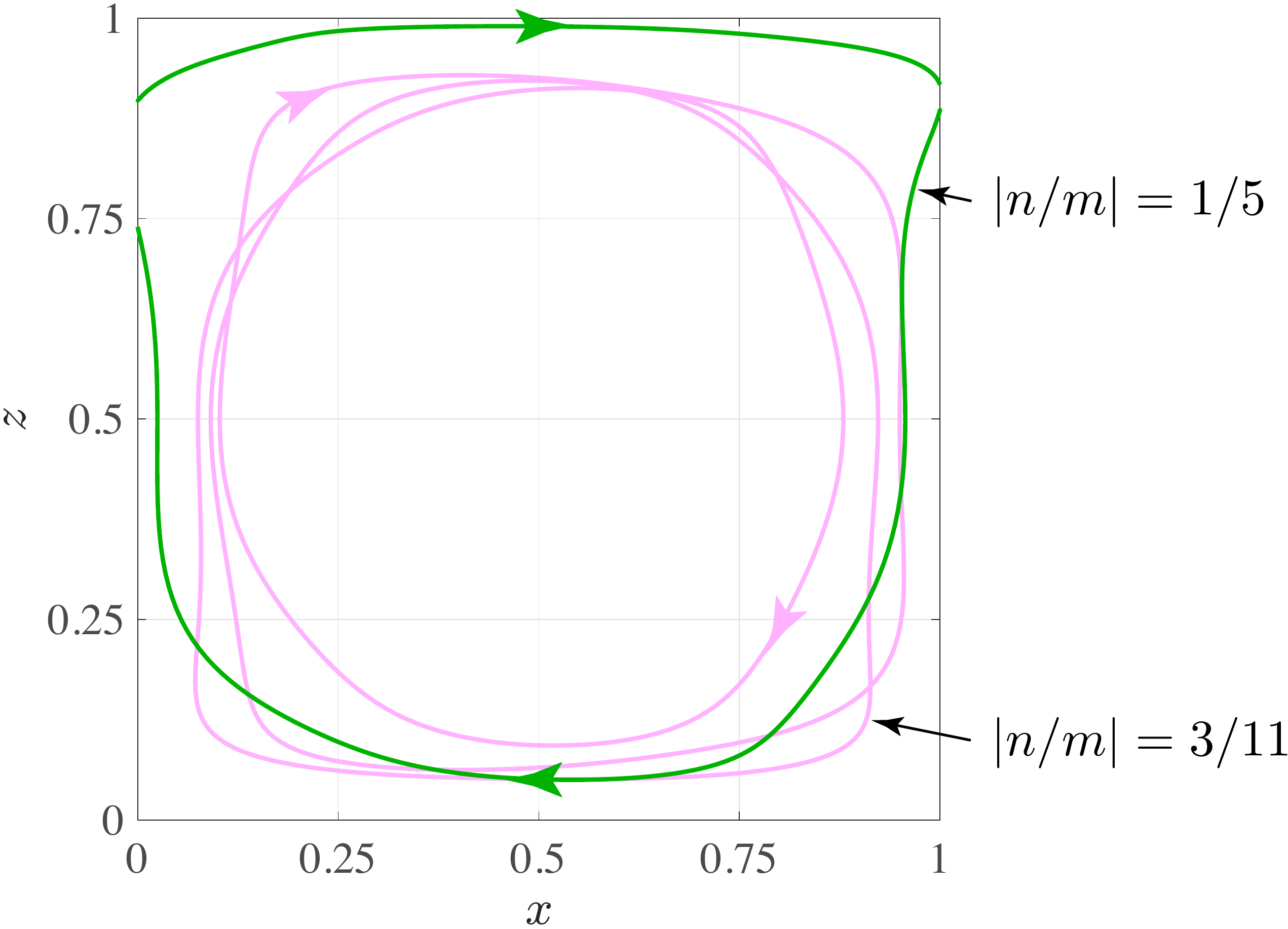}
\label{fig:prd_orbit_MP5_11}}
\subfigure{\includegraphics[scale=0.3]{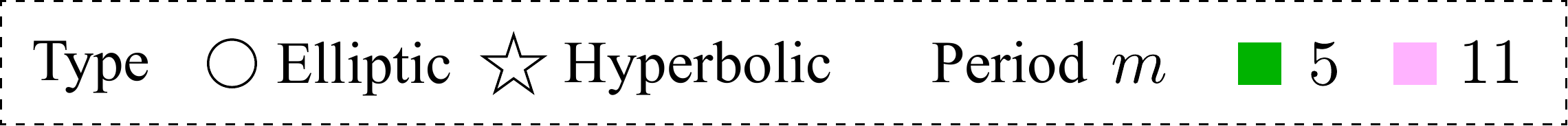}}
\vspace{-2mm}
\caption{Asymmetric orbits}
\label{fig:prd_point_orbit_MP5_11}
\end{center}
\end{figure}

\vspace{-5mm}

\begin{figure}[H]
\begin{center}
\subfigure[9-peirodic orbit ($n=-3$)]
{\includegraphics[scale=0.25]{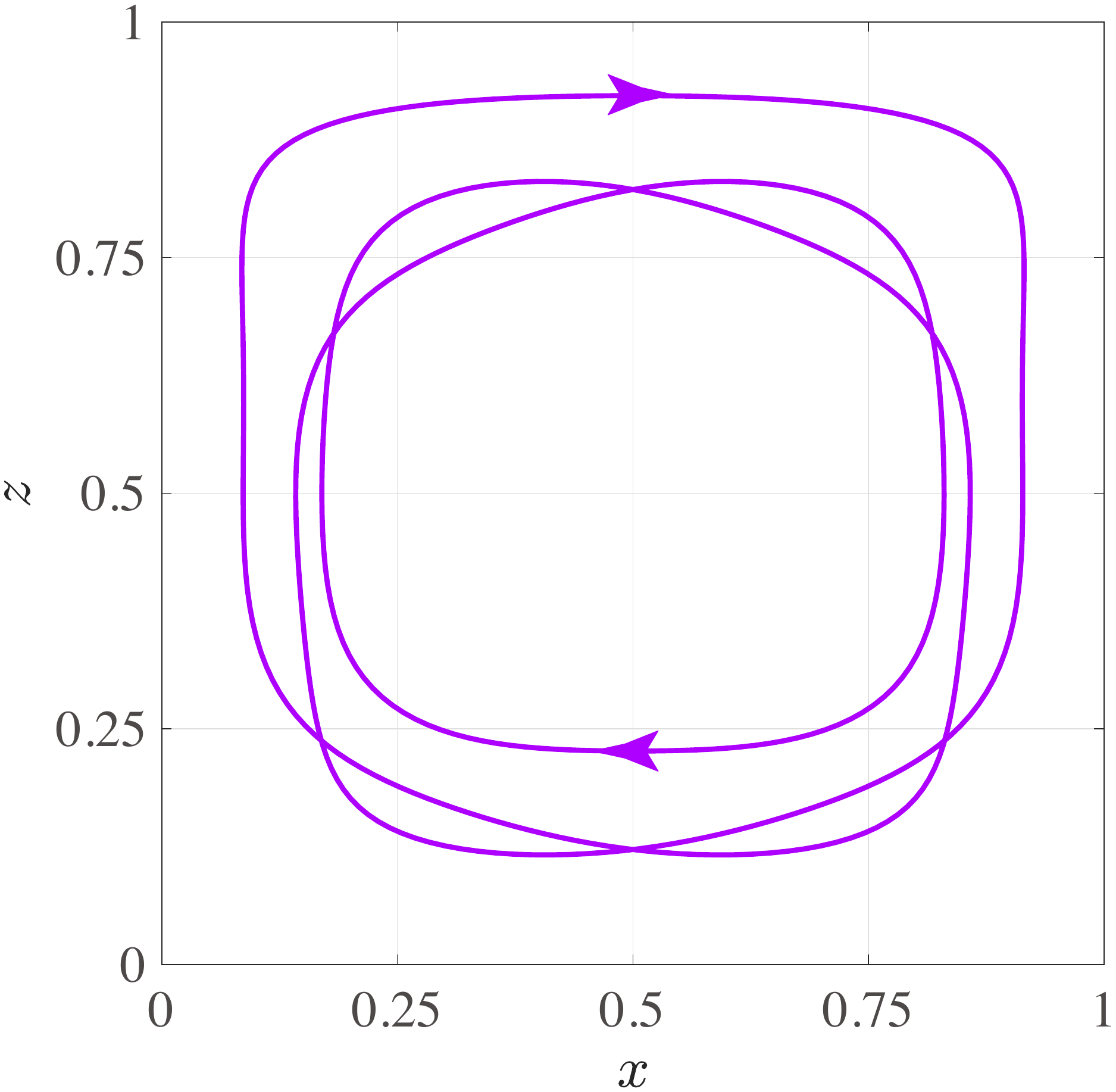}
\label{fig:prd_point_9a}}
\hspace{5mm}
\subfigure[12-peirodic orbit ($n=-4$)]
{\includegraphics[scale=0.25]{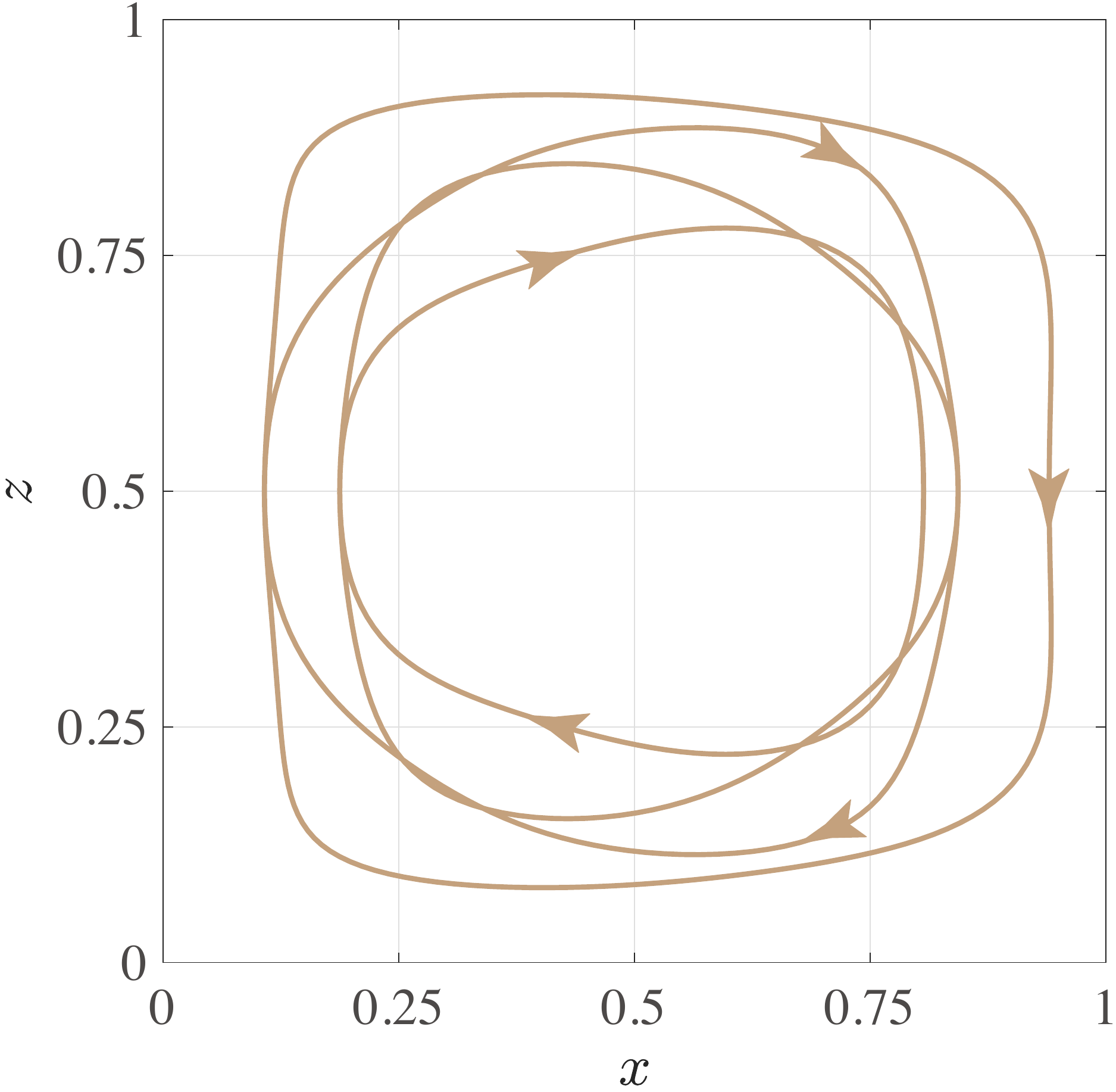}
\label{fig:prd_point_12a}}
\subfigure[15-peirodic orbit ($n=-5$)]
{\includegraphics[scale=0.25]{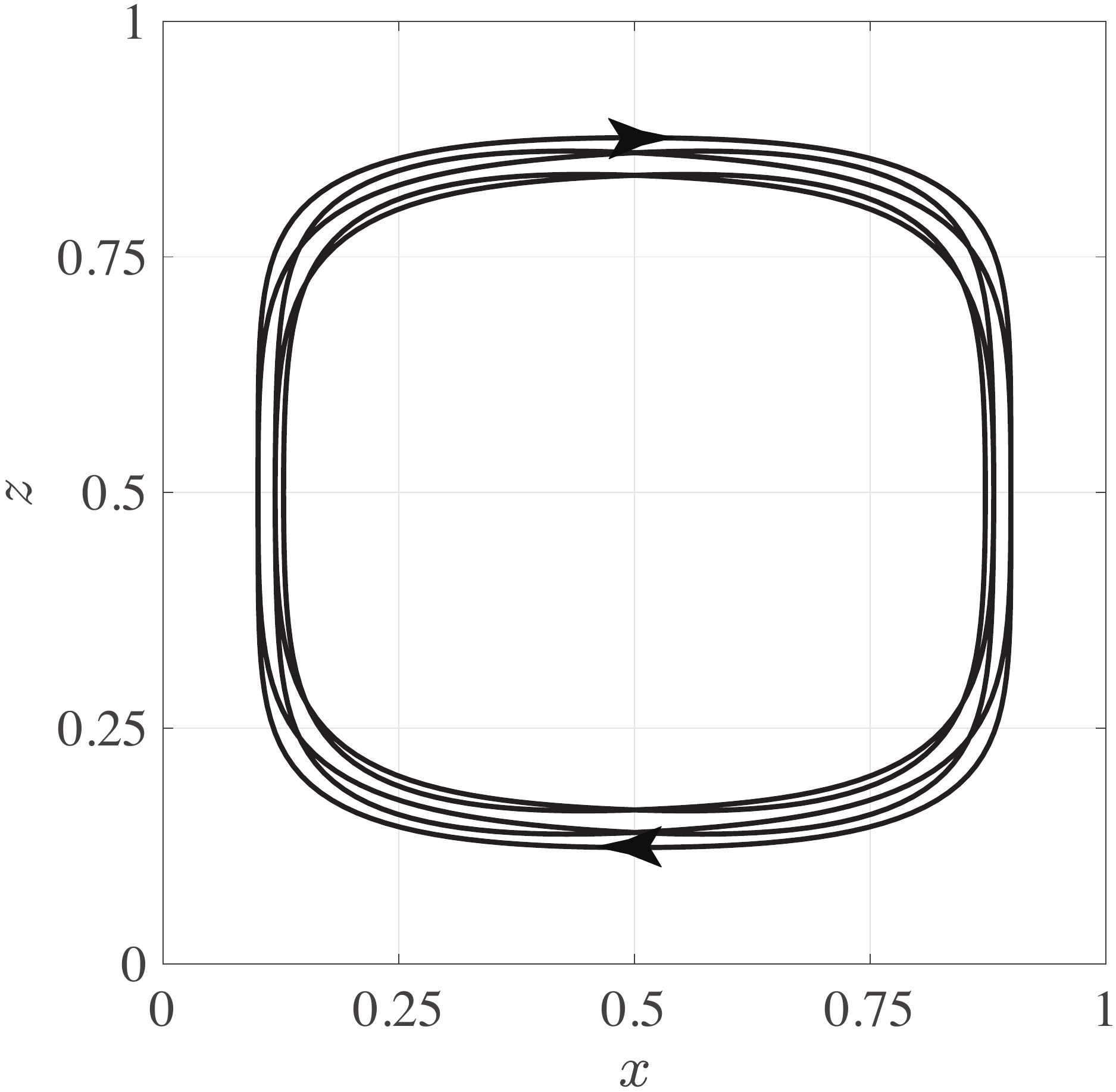}
\label{fig:prd_point_15a}}
\hspace{5mm}
\subfigure[15-peirodic orbit ($n=-5$)]
{\includegraphics[scale=0.25]{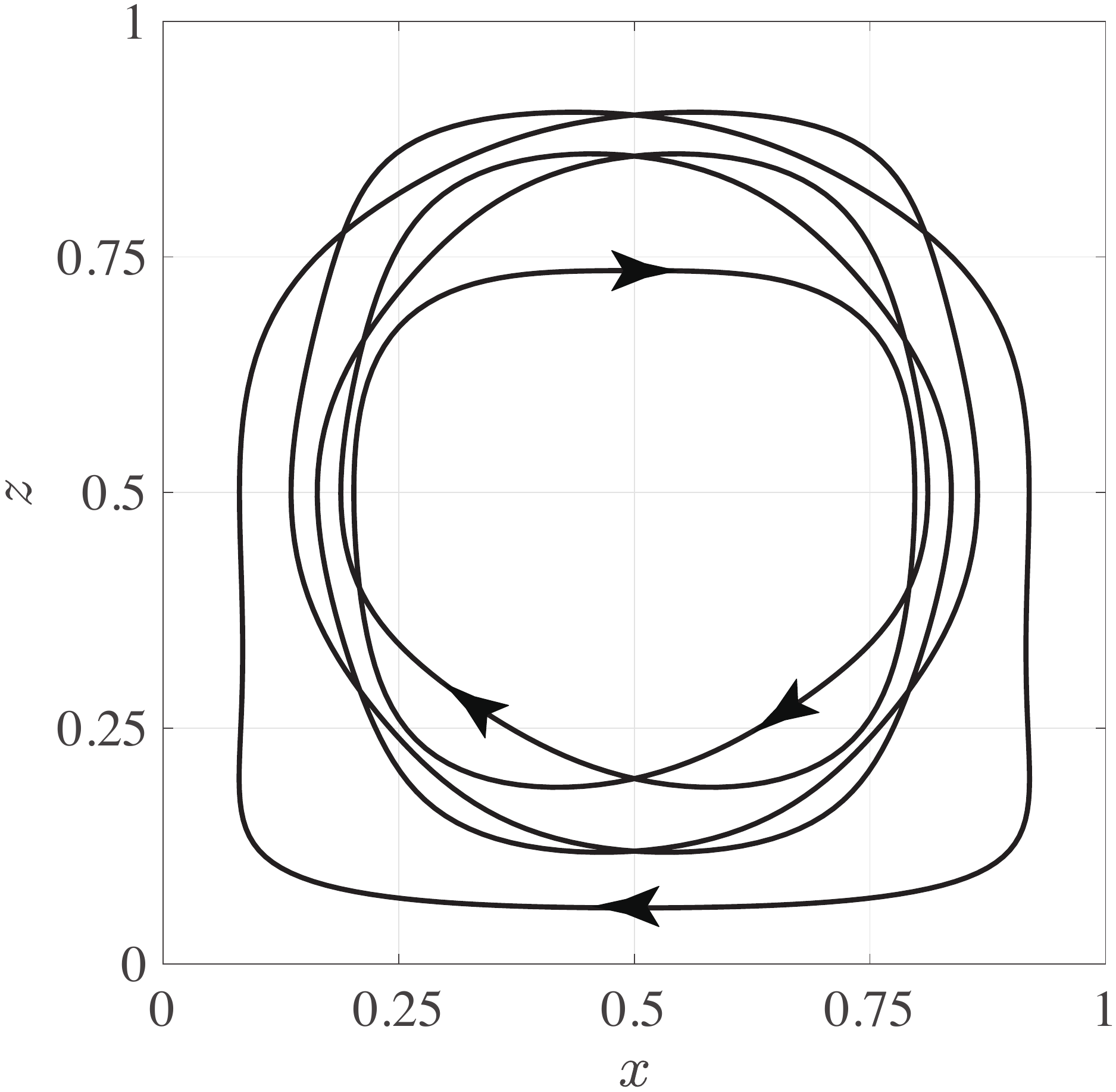}
\label{fig:prd_point_15a}}
\vspace{-2mm}
\caption{The projection of periodic orbits with $|n/m|=1/3$}
\label{fig:prd_point_orbit_1/3}
\end{center}
\end{figure}

Furthermore, it is found in our numerical computation that some of the orbits have the same resonance conditions 
even when their periods are different. For example, we illustrate some of the orbits of which 
resonance condition is $|n/m|=1/3$ in Fig.\ref{fig:prd_point_orbit_1/3}. 
As can be seen, the winding number of the 9, 12, and 15-periodic orbits are $n=-3, -4$ and -5 respectively, 
where we recall that the negative sign indicates that the periodic orbits have the clockwise direction. 
Such periodic orbits seem to be related to the fold bifurcations as we shall discuss this in \S\ref{Sec:bifurcation}.

\paragraph{Symmetries of the projected orbits.}
We next focus on the symmetric properties of the projected orbits with respect to the horizontal and vertical center lines of a cell. 
Let us recall the following symmetric properties i) and iv) of the non-autonomous system in \eqref{PerHamEq}: 
\begin{equation*}
\begin{split}
&\textrm{i)} \quad \displaystyle  x \mapsto x + \frac{2a\pi}{k},\quad z \mapsto -z + 1, \quad t \mapsto -t+bT,\\[3mm]
&\textrm{iv)} \quad \displaystyle x \mapsto -x + \frac{(2a+1)\pi}{k},\quad z \mapsto z, \quad t \mapsto -t + \biggl(b+\frac{1}{2} \biggr)T,
\end{split}
\end{equation*}
where we recall $T(=2\pi/\omega)$ is the period of the perturbation and $a,b \in \mathbb{Z}$. 
Since the projection of the periodic curve $c(t)=\pi(\tilde{c}(t))$ correspond to the solution curve of the non-autonomous system, it follows 
that if the projection $c$ of a periodic orbit is not symmetric with respect to the vertical center line $x=(2a+1)\pi/(2k)$, 
there exists another periodic orbit of which projection is symmetric with $c$ with respect to $x=(2a+1)\pi/(2k)$. 
This is the same with respect to the horizontal center line $z=1/2$ as well.  
Hence, if the projection of a periodic orbit is not 
symmetric with respect to $x=(2a+1)\pi/(2k)$ and $z=1/2$, there exist three more orbits of which each projection is symmetric 
with $c$ with respect to $x=(2a+1)\pi/(2k)$ or $z=1/2$.
Fig.\ref{fig:prd_orbit_sym} illustrate the orbits that are symmetric with those in Fig.\ref{fig:prd_point_orbit_MP6_7_8_9} 
and Fig.\ref{fig:prd_point_orbit_MP4_12_14}. Note that the evolution of the orbits are depicted in the positive direction of time $t$ in both figures and also that the orientation of orbits could be opposite when computing the evolution for the negative direction of $t$, while the resonance conditions of the orbits that are symmetric in spatial coordinates $(x,z)$ with each other are the same.

\begin{figure}[H]
\begin{center}
\subfigure[Orbits symmetric with those in Fig.\ref{fig:prd_point_orbit_MP6_7_8_9}]
{\includegraphics[scale=0.25]{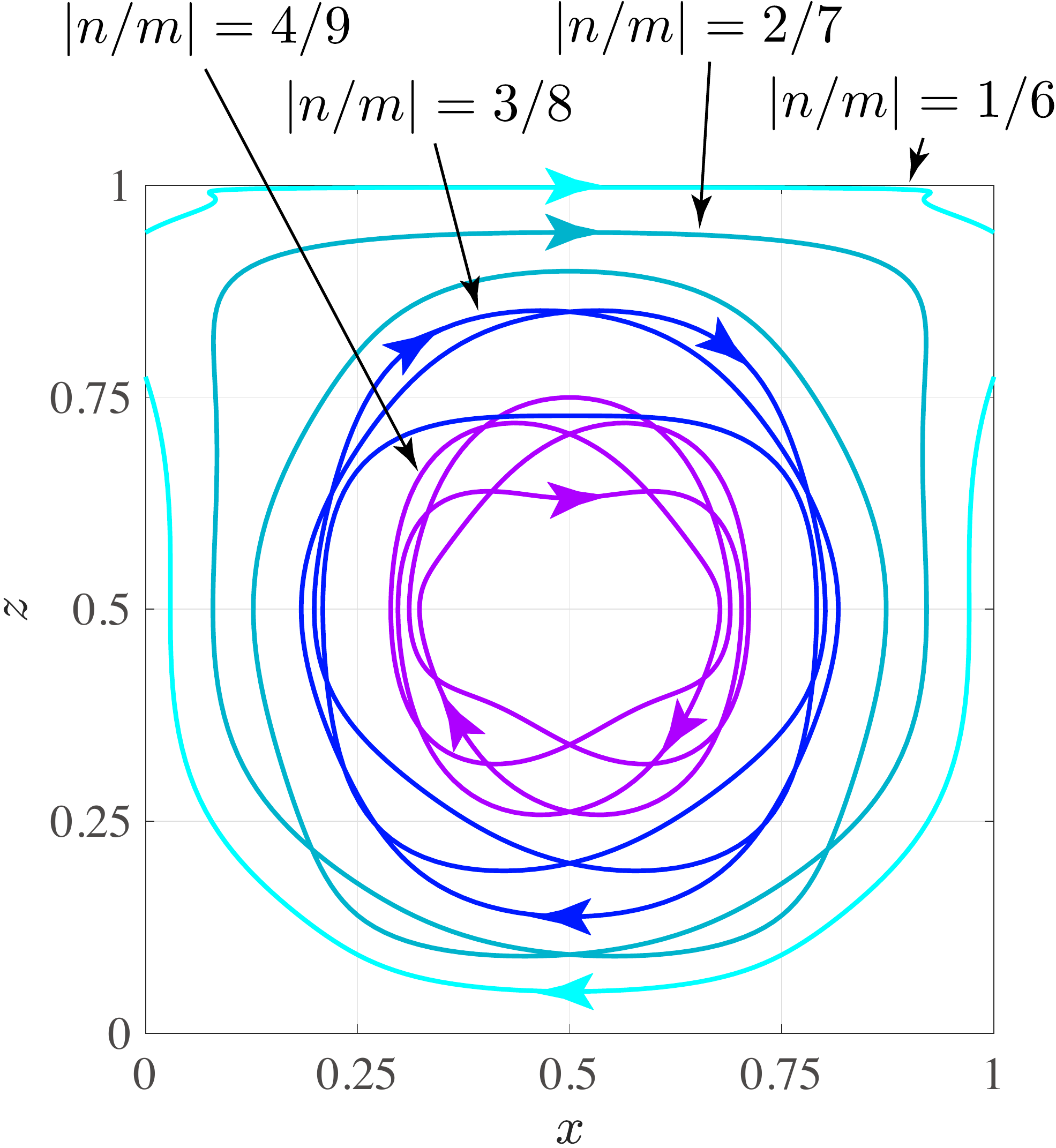}
\label{fig:prd_orbit_MP6_7_8_9b}}
\hspace{8mm}
\subfigure[Orbits symmetric with those in Fig.\ref{fig:prd_point_orbit_MP4_12_14}]
{\includegraphics[scale=0.25]{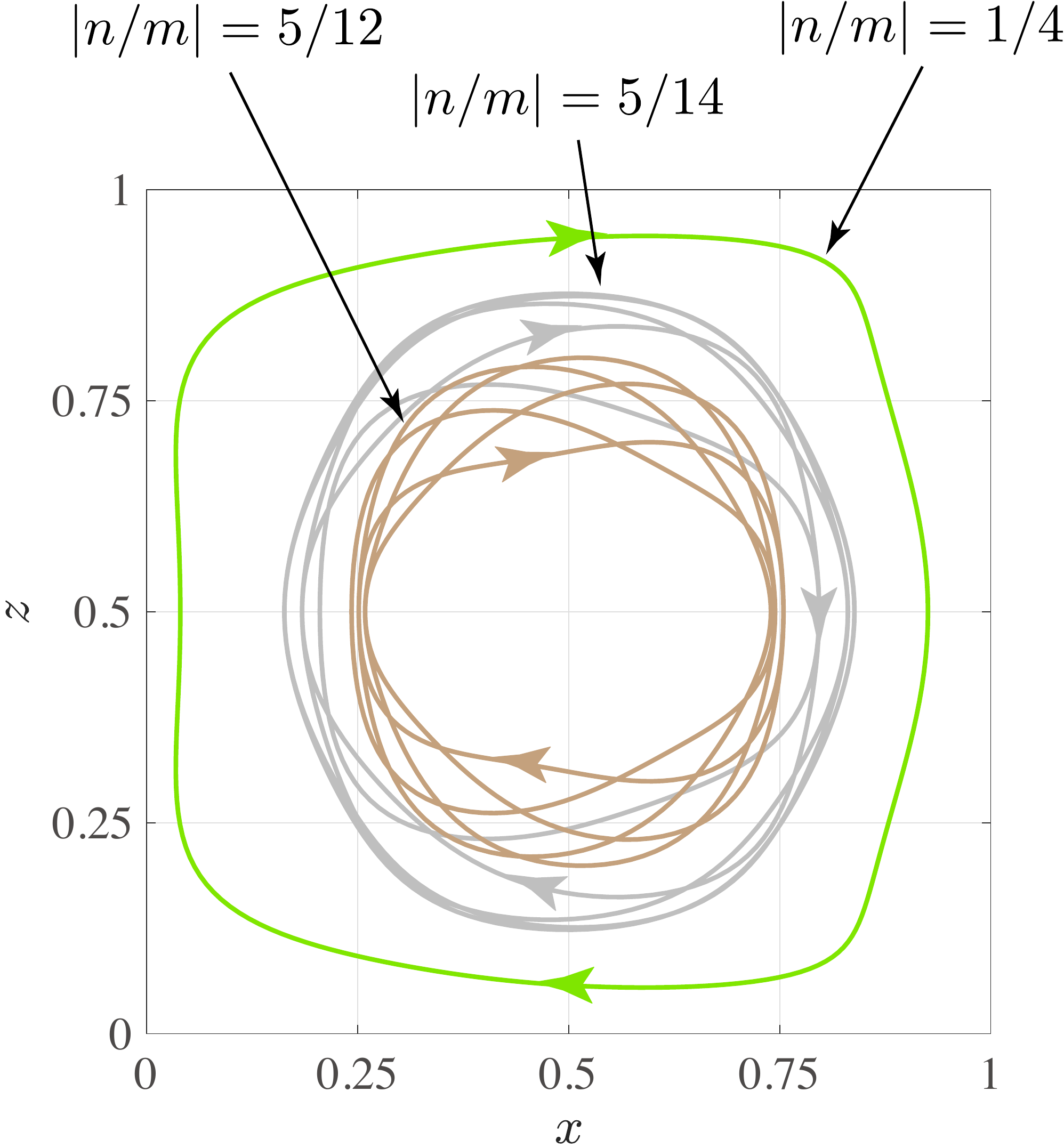}
\label{fig:prd_orbit_MP4_2_14b}}
\subfigure{\includegraphics[scale=0.3]{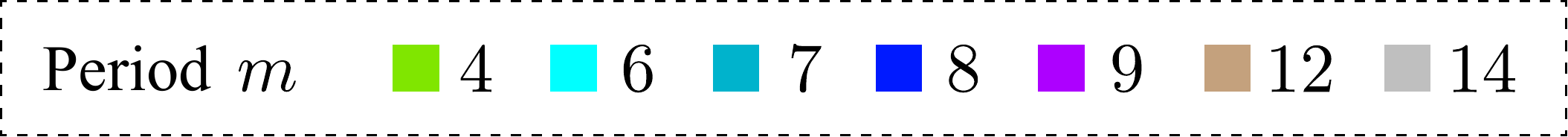}}
\caption{Orbits symmetric with those in Fig.\ref{fig:prd_point_orbit_MP6_7_8_9} and Fig.\ref{fig:prd_point_orbit_MP4_12_14}}
\label{fig:prd_orbit_sym}
\end{center}
\end{figure}

\begin{remark}[Action angle variables]\rm
We can introduce the action angle variables $(J,\phi)$ to transform the Hamiltonian system 
in terms of $(x,z)$ to that in terms of $(J,\phi)$. 
When the model is unperturbed, i.e., $\varepsilon=0$, $J$ and $\phi$ are obtained by
\begin{equation}\label{actionangle_unperturbed}
\begin{split}
J = \frac{1}{2\pi} \oint zdx,\qquad 
\phi = \frac{2\pi}{T}t,
\end{split}
\end{equation}
where the integral is taken over one cycle of the periodic curve of \eqref{RBCModel} which preserves
$
H(x,z)=H \;(\textrm{constant})
$
and $T$ is the period of the orbit. 
Then the unperturbed model \eqref{RBCModel} can be rewritten as
\begin{eqnarray*}
\frac{dJ}{dt}=0,\qquad
\frac{d\phi}{dt}=\Lambda(J),
\end{eqnarray*}
where
$
\Lambda(J)=2\pi/T.
$
Then, the perturbed system \eqref{PerHamEq} can be restated in terms of $(J,\phi)$ as
\begin{equation}\label{actionangle_perturbed}
\begin{split}
\frac{dJ}{dt}=\varepsilon f(J,\phi,t),\qquad
\frac{d\phi}{dt}=\Lambda(J)+\varepsilon g(J,\phi,t),
\end{split}
\end{equation}
where
$f(J,\phi,t)=\frac{\partial J}{\partial x}\frac{\partial H_1}{\partial z}+\frac{\partial J}{\partial z}\frac{\partial H_1}{\partial x}$ and $g(J,\phi,t)=\frac{\partial \phi}{\partial x}\frac{\partial H_1}{\partial z}+\frac{\partial \phi}{\partial z}\frac{\partial H_1}{\partial x}$; see, for instance, \cite{Wi1990}.


As was shown in \eqref{ExtHamEq}, the perturbed Hamiltonian system \eqref{actionangle_perturbed} can be written as an autonomous system in terms of $(J,\phi,\theta) \in \mathbb{R} \times S^1 \times S^1$ as
\begin{equation}\label{ActAngle_pertHamSys}
\begin{split}
\frac{dJ}{dt}&=\varepsilon f (J,\phi,\theta),\\[2mm]
\frac{d\phi}{dt}&=\Lambda(J)+\varepsilon g (J,\phi,\theta),\\[2mm]
\frac{d\theta}{dt}&=\omega,
\end{split}
\end{equation}
where $\theta=\omega t + \theta_0$.
Of course, the perturbed Hamiltonian system \eqref{ExtHamEq} with the variables $(x,z,\theta)$ is transformed into the system \eqref{ActAngle_pertHamSys} with the action-angle variables  $(J,\phi,\theta)$.

\end{remark}

\begin{remark}[Poincar\'e-Birkhoff theorem]\rm
Let us consider an invariant curve with action $J$ such that $\Lambda(J)=n/m$ in the unperturbed system, 
where $m$ and $n$ are integers. The Poincar\'e-Birkhoff theorem states that when the system is perturbed, 
$2lm$ of $m$-periodic points appear in the neighborhood of the original invariant curve, where $l$ is some unknown integer. 
In particular,  $l$ of them are to be elliptic and the others are to be hyperbolic; see \cite{Bi1927} and \cite{LiLi1991}.
\end{remark}

\subsection{Symmetries of $n/m$-resonant orbits}
In this subsection, we consider the symmetric properties concerning the $n/m$-resonant orbits, namely, the $m$-periodic orbits with winding number $n$.

We consider the special case of such $n/m$-resonant orbits $\tilde{c}(t)$ in the extended phase space $M \times S^1$ in which $c(t) =\pi(\tilde{c}(t) )$ is symmetric with respect to the horizontal and vertical center lines of a cell, namely $x=(2a+1)\pi/(2k)~ (a \in \mathbb{Z})$ and $z=1/2$, by the following theorem.

\begin{theorem}[Symmetries of $n/m$-resonant orbits]\label{SymTheorem}\rm
Let
\begin{equation*}
\tilde{c}(t):=(x(t),z(t), \theta(t)) \subset M \times S^1, ~ t \in \mathbb{I}=[0,mT]
\end{equation*}
be a $n/m$-resonant orbit such that $\tilde{c}(0) = \tilde{c}(mT) = p$, where $p$ is an $m$-periodic point on $\Sigma^{\theta_0}$.
Then, let 
\begin{equation*}
c(t):=\pi(\tilde{c}(t))
\end{equation*}
be a periodic curve on $M$. If $c(t)$ is symmetric with respect to the horizontal and vertical center lines of a cell, namely $x=(2a+1)\pi/(2k)~ (a \in \mathbb{Z})$ and $z=1/2$, the period $m$ and the winding number $n$ of $c(t)$ are both odd. 
\end{theorem}

\begin{proof}
For the sake of proving this theorem, recall the following symmetric properties i), iv), and v) of the non-autonomous system in \eqref{PerHamEq}, since $c(t) \subset M$ corresponds to the solution curve of the non-autonomous system. 
\begin{equation*}
\begin{split}
&\textrm{i)} \quad \displaystyle  x \mapsto x + \frac{2a\pi}{k},\quad z \mapsto -z + 1, \quad t \mapsto -t+bT,\\[3mm]
&\textrm{iv)} \quad \displaystyle x \mapsto -x + \frac{(2a+1)\pi}{k},\quad z \mapsto z, \quad t \mapsto -t + \biggl(b+\frac{1}{2} \biggr)T, \\[3mm]
&\textrm{v)}\qquad x \mapsto -x + \frac{(2a+1)\pi}{k},\quad z \mapsto -z + 1,\quad t \mapsto t + \biggl(b+\frac{1}{2} \biggr)T,
\end{split}
\end{equation*}

As is shown in Fig.\ref{fig:proofB}, consider an $n/m$-resonant periodic orbit $\tilde{c}(t) \in M \times S^1$ such that $\tilde{c}(0)=\tilde{c}(mT)=p$, 
where $p$ is an $m$-periodic point, and suppose that the periodic curve $c(t) =\pi(\tilde{c}(t) )$ on $M$ has the symmetric properties that $c(t)$ is symmetric with respect to $x=(2a+1)\pi/(2k)~ (a \in \mathbb{Z})$ and $z=1/2$. 
Note that $c(t)$ is partly illustrated in dashed lines to indicate a general curve in Fig.\ref{fig:proofB}, 
which denotes that the dashed lines can have a loop as long as $c(t)$ maintain the symmetric properties.
\medskip

First, we shall prove that $m$ is odd. To do this, let $p_1:=\pi(p)$ such that $p_1=c(0)=c(mT) \in M$ 
and let $p_2 \in M$ be the associated symmetric point with $p_1$ regarding the horizontal center line of a cell, 
namely $z=1/2$. Since $c(t)$ is symmetric with respect to $z=1/2$, 
it follows from property i) that $p_2$ can be expressed as $p_2=c(lT)$, where $l$ is some integer such that $0 \leq l \leq m-1$.

We denote the initial time for $p_1$ and $p_2$ by $t_1=0$ and $t_2=lT$ respectively. 
Then, one can define an intermediate point $p_3$ in a path from $p_1$ to $p_2$ such that $p_3:=c(t_3)$, where
\begin{equation*}
t_3=\frac{t_1 + t_2}{2}=\frac{l}{2}T
\end{equation*}
is the middle time between $t_1$ and $t_2$.
Further, we denote the first return time for $p_1$ as $t_1'=mT$. 
Then, one can define an intermediate point $p_4$ in a path from $p_2$ to $p_1$ such that $p_4:=c(t_4)$, 
where 
\begin{equation*}
t_4=\frac{t_2 + t_1'}{2}=\frac{l+m}{2}T
\end{equation*}
is the middle time between $t_2$ and $t_1'$.
Since $p_1$ and $p_2$ are the points of curve $c(t)$ at $t \equiv 0 ~({\rm mod}\; T)$ 
and that they are symmetric with respect to $z=1/2$, 
it follows from property i) that $p_3$ and $p_4$ lie on the horizontal center line $z=1/2$, as is shown in Fig.\ref{fig:proofB}. 

Next, let $p_5 \in M$ be the associated symmetric point with $p_1$ regarding the vertical center line of a cell, 
namely $x=(2a+1)\pi/(2k)$. Since $c(t)$ is symmetric with respect to the vertical center line, $p_5$ is a point of $c(t)$. 
Furthermore, it follows from property v) that the integration time from $p_2$ to $p_5$ is the half of the period of the orbit, 
namely $mT/2$, since $p_2$ and $p_5$ are symmetric with respect to point $(x,z)=((2a+1)\pi/(2k),1/2)$.
Thus, $p_5$ can be expressed as $p_5=c(t_5)$, where
\begin{equation*}
t_5=t_2+\frac{mT}{2}=\left( l + \frac{m}{2} \right) T.
\end{equation*}
Then, the integration times from $p_1$ to $p_3$ and $p_4$ to $p_5$ become the same as is shown below.
\begin{eqnarray*}
\begin{split}
t_3 - t_1 = t_5 - t_4 = \frac{l}{2}T
\end{split}
\end{eqnarray*}
Since $p_1$ and $p_5$ are symmetric with respect to the vertical center line, 
it follows from property iv) that $p_3$ and $p_4$ are also symmetric with respect to the vertical center line, 
as is shown in Fig.\ref{fig:proofB}. 

Now, we prove by contradiction that period $m$ is an odd number. To do this, let us assume that $m$ is an even number. 
Then, time $t_3$ and $t_4$ become $t_3 \equiv t_4 \equiv 0$ when $l$ is even, 
while they become $t_3 \equiv t_4 \equiv T/2$ when $l$ is odd.
However, since $p_3$ and $p_4$ are symmetric with respect to $x=(2a+1)\pi/(2k)$, 
it follows from property iv) that there are only two cases; One is the case when $t_3 \equiv 0$ and $t_4 \equiv T/2$, 
and the other is the case when $t_3 \equiv T/2$ and $t_4 \equiv 0$. 
Therefore, the assumption that $m$ is an even number is not correct.
Thus, it is proved that $m$ is an odd number. 

\begin{figure}[H]
\begin{center}
\includegraphics[scale=0.48]{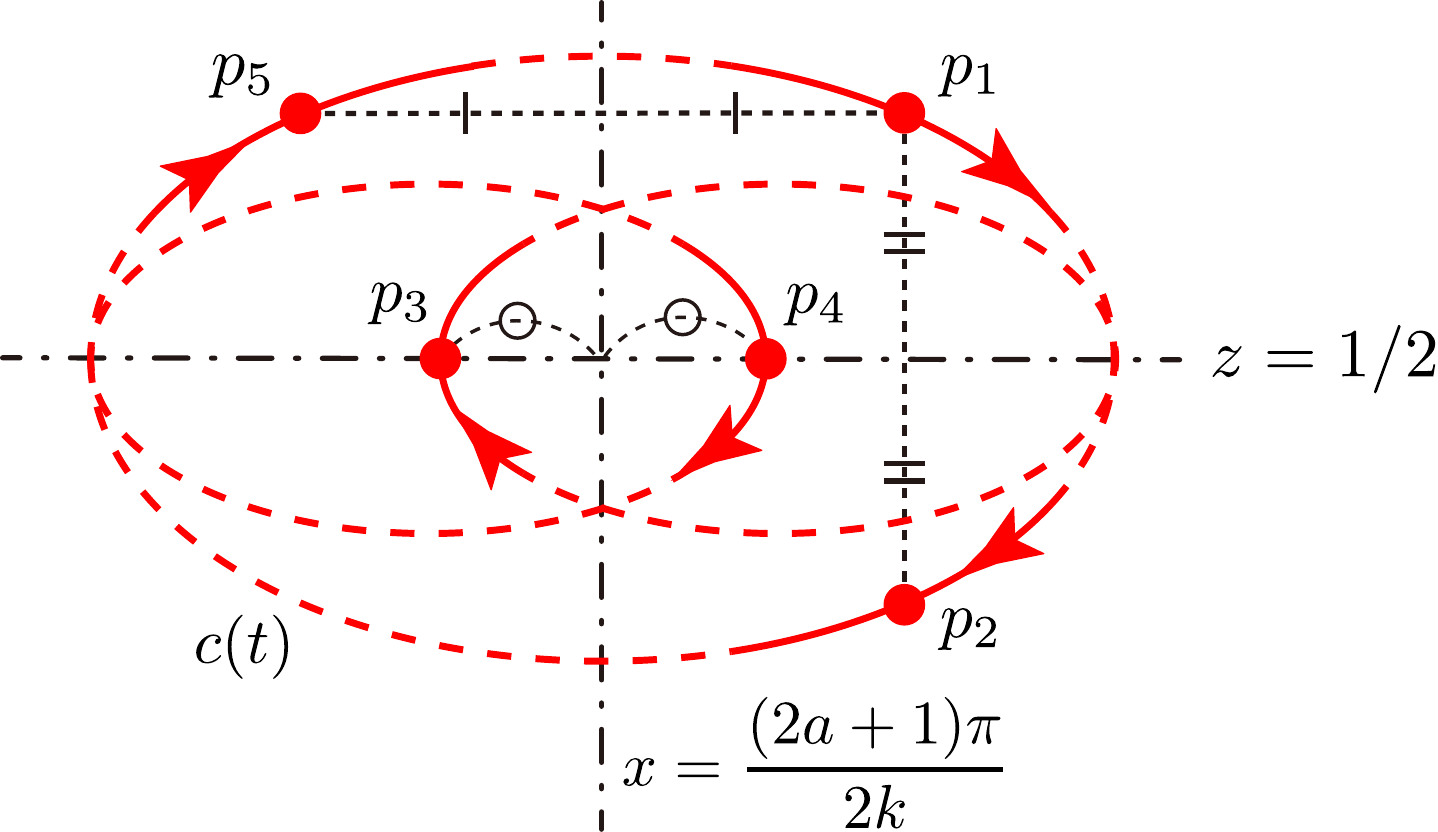}
\caption{Points $p_1,p_2,p_3$ and $p_4$ on curve $c(t)$}
\label{fig:proofB}
\end{center}
\end{figure}
\medskip

Next, we shall prove that $n$ is odd. Recall that the winding number $n$ of a periodic orbit $c(t)$ is given by \eqref{eq:definition_n}, where $c(t)$ is regarded as a closed curve $c(t)=x(t) + iz(t)$ in $\mathbb{C}$ and the interval of integration can be divided as
\begin{eqnarray*}
n &=& \frac{1}{2\pi i} \oint_{c(t)} \frac{dw}{w-w_c}\\
&=& \frac{1}{2\pi i} \left( \int_{c_{3,4}(t)} \frac{dw}{w-w_c} + \int_{c_{4,3}(t)} \frac{dw}{w-w_c} \right). 
\end{eqnarray*}
Here, $c_{3,4}(t)$ and $c_{4,3}(t)$ respectively denote the part of curve $c(t)$ from $p_3$ to $p_4$ and vice versa. 
From assumption, note that $c(t)$ is symmetric with respect to $z=1/2$ and also that $p_3$ and $p_4$ lie on $z=1/2$, and it follows
\begin{equation*}
 \int_{c_{3,4}(t)} \frac{dw}{w-w_c} = \int_{c_{4,3}(t)} \frac{dw}{w-w_c},
\end{equation*}
where $w=x+iz$ is an arbitrary point on $c(t)$ and $(x_c,z_c)=((2a+1)\pi/(2k),1/2)$ is set as a fixed point.
Therefore, 
\begin{eqnarray*}
n &=& \frac{1}{\pi i} \int_{c_{3,4}(t)} \frac{dw}{w-w_c}\\
&=& \frac{1}{\pi i} \{ \ln (w(t_4)-w_c) - \ln (w(t_3)-w_c) \}.
\end{eqnarray*}

Now, we rewrite a point $w=x+iz$ on $c(t)$ in the polar coordinates as 
\begin{equation*}
w(t)-w_c=r(t)e^{i \psi (t)},
\end{equation*}
where $r=|w|=\sqrt{x^2+ z^2}\geq 0$ and $\psi=\mathrm{arg}\, w=\arctan (z/x)$. 
Since $p_3$ and $p_4$ lie on $z=1/2$ and are symmetric with each other with respect to $x=(2a+1)\pi/(2k)$, 
\begin{eqnarray*}
r(t_3) = r(t_4),\\
\psi(t_4) - \psi(t_3) = \pi,
\end{eqnarray*}
Therefore, 
\begin{eqnarray*}
n &=& \frac{1}{\pi i} \{ \ln r(t_4) + i(\psi(t_4)+2l'\pi) - \ln r(t_3) - i(\psi(t_3)+2l''\pi) \}\\
&=& 2(l' - l'') +1,
\end{eqnarray*}
where $l', l'' \in \mathbb{Z}$. 
Hence, it is proved that $n$ is an odd number.
Thus, theorem is proved.
\end{proof}

As the theorem states, we can see in Fig.\ref{fig:prd_orbit_MP1_3_5_7_11} that the periodic orbits of which projection 
is symmetric with respect to the horizontal and vertical center lines of the cell have odd period $m$ and 
winding number $n$. Furthermore, the following corollary can be stated from Theorem \ref{SymTheorem}. 

\begin{corollary}\label{corollary}\rm
If the period $m$ or the winding number $n$ of a periodic orbit $\tilde{c}(t)$ is an even number, there appear one or three more 
$n/m$-resonant orbits of which projection is symmetric with $c(t)=\pi(\tilde{c}(t))$ 
with respect to either horizontal or vertical center lines of a cell.
\end{corollary}

\begin{proof}
Considering the contraposition of Theorem \ref{SymTheorem}, if $m$ or $n$ is an even number, 
$c(t)=\pi(\tilde{c}(t))$ is not symmetric with respect to either horizontal or vertical center lines of a cell.
If $c(t)$ is not symmetric with only one of the two lines, it follows from property i) or iv) that there appear one more $n/m$-resonant orbit
of which projection is symmetric with $c(t)$ with respect to either of the two lines.
If $c(t)$ is not symmetric with both of the two lines, it follows from property i) and iv) that there appear three more 
$n/m$-resonant orbit of which projection is symmetric with $c(t)$ with respect to either of the two lines. 
Thus, the corollary is proved.
\end{proof}

As the corollary states, it is observed in Fig.\ref{fig:prd_point_orbit_MP6_7_8_9}, Fig.\ref{fig:prd_point_orbit_MP4_12_14} 
and Fig.\ref{fig:prd_orbit_sym} that another symmetric $n/m$-resonant orbit appear when the period $m$ or 
the winding number $n$ of the periodic orbit is even. 


\section{Bifurcations of periodic orbits}
\label{Sec:bifurcation}
As already mentioned, the amplitude of the perturbation of the Rayleigh-B\'enard convection increases 
when the Rayleigh number $Ra$ is gradually raised from the critical number $Ra_t$ 
by increasing the temperature difference between the top and bottom planes. 
In this section, we study the bifurcations of periodic orbits in the perturbed Hamiltonian system 
by varying the parameter $\varepsilon$, i.e., the amplitute of the perturbation in order to clarify 
how the fluid transport changes with $\varepsilon$. 
We first describe the global structure of $\varepsilon$-bifurcation diagram 
and then clarify the structures of the bifurcations associated with the main KAM island $I_1$ and the surrounding islands
$I_2, I_3,$ and $I_4$, and furthermore those associated with other islands.

\subsection{Structure of $\varepsilon$-parameter bifurcation}
\paragraph{Computation of one-parameter bifurcation diagrams.}
In the numerical computations in \S\ref{Sec:Poincare} and \S\ref{Sec:symmetry}, we have analyzed 
the periodic points and the associated orbits when the amplitude $\varepsilon$ of the perturbation is set to $\varepsilon=0.1$. 
In order to obtain the $\varepsilon$-parameter bifurcation diagram of the periodic points in space $(x,z,\varepsilon)$, we shall compute to detect the elliptic and hyperbolic periodic points on the Poincar\'e section of one single cell for $\varepsilon=0.001, 0.002, \cdots ,0.5$ in the same way for $\varepsilon=0.1$. The other parameters of the convection and the initial condition of $\theta$ are set to
$A=\pi, k=\pi, T=1/\pi$, and $\theta_0=0$ as the same in Fig.\ref{fig:pm_prd}.  Fig.\ref{fig:bifurcation_all} shows the detected bifurcation diagram from diagonal and $z$ direction, 
where the periodic points with period $m \leq 15$ are depicted. 
Note that the computations are conducted independently for each $\varepsilon$. 
The color of the plots indicate the period $m$ of each point, however the types of the points, 
namely elliptic or hyperbolic, are not illustrated in Fig.\ref{fig:bifurcation_all}. 
We will depict them in the figures shown latter. 

\begin{figure}[H]
\begin{center}
{\includegraphics[scale=0.25]{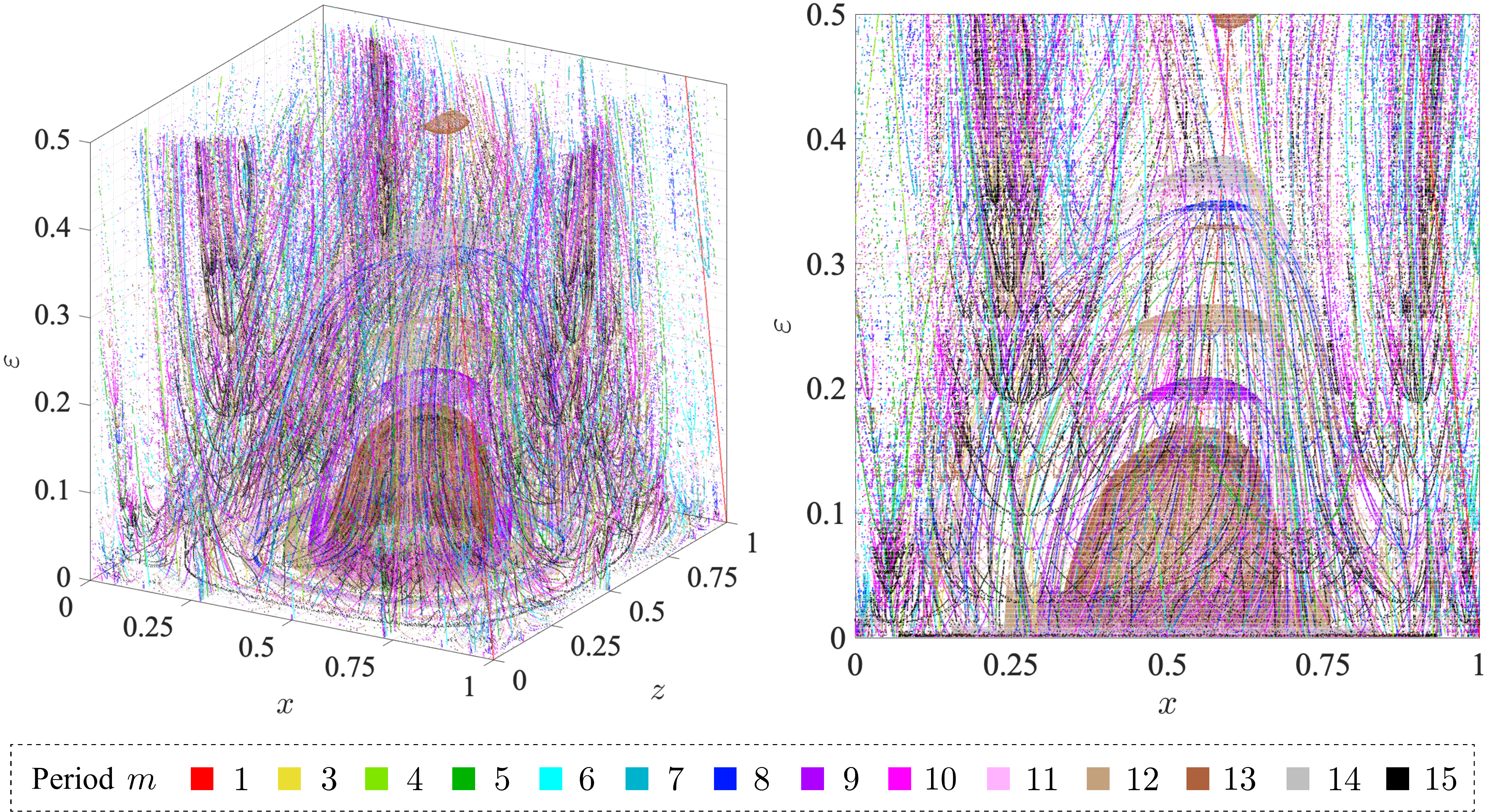}}
\caption{$\varepsilon$-parameter bifurcation diagram of periodic points}
\label{fig:bifurcation_all}
\end{center}
\end{figure}

\paragraph{Periodic points on the Poincar\'e section with some parameters $\varepsilon$.}
Before we take a look at the bifurcation diagram let us show how the Poincar\'e maps and 
the detected periodic points vary with the amplitude $\varepsilon$ of the perturbation. 
Fig.\ref{fig:pm_prd_epsilon} illustrates the image of the Poincar\'e section by Poincar\'e map $P^{\theta_0}_\varepsilon$ 
and the periodic points for $\varepsilon=0.2, 0.3, 0.4,$ and 0.5. 
As can be seen, the islands of KAM tori, which correspond to stable transport regions, 
exist for a while when $\varepsilon$ is increased.
However, when we increase it furthermore, the area of the islands and the number of elliptic periodic points 
gradually decrease. Especially, islands $I_2, I_3$ and $I_4$ seem to disappear by $\varepsilon=0.5$. 
In contrast, it is apparent that the area of chaotic regions increases. 
This denotes that the periodic orbits in the system of \eqref{ExtHamEq} bifurcate one after another 
and lead to chaotic orbits when $\varepsilon$ is increased. 

\paragraph{Bifurcations of 1 and 3-periodic points.}
Now we take a closer look at the bifurcation diagram detected in our numerical computation. 
Since it is too complicated to understand the structure of the diagram from Fig.\ref{fig:bifurcation_all}, 
let us first focus on the bifurcations of 1 and 3-periodic points, which are illustrated in Fig.\ref{fig:bifurcation_MP1_3} 
from diagonal and $z$ direction. 
Here, the branches of elliptic and hyperbolic periodic points are depicted in thick and thin lines respectively. 
In addition, we especially depict the 1 and 3-periodic points with the image of the Poincar\'e section 
at $\varepsilon=0.1, 0.4$ in Fig.\ref{fig:pm_prd_epsilon_MP1_3} so that we can clearly see the periodic points. 
As is shown in Fig.\ref{fig:pm_prd_epsilon_MP1_3}, an elliptic 1-periodic point and three hyperbolic 3-periodic points 
appear at the center and the corners of the main island $I_1$ respectively. 

\begin{figure}[H]
\begin{center}
\subfigure[$\varepsilon=0.2$]
{\includegraphics[scale=0.26]{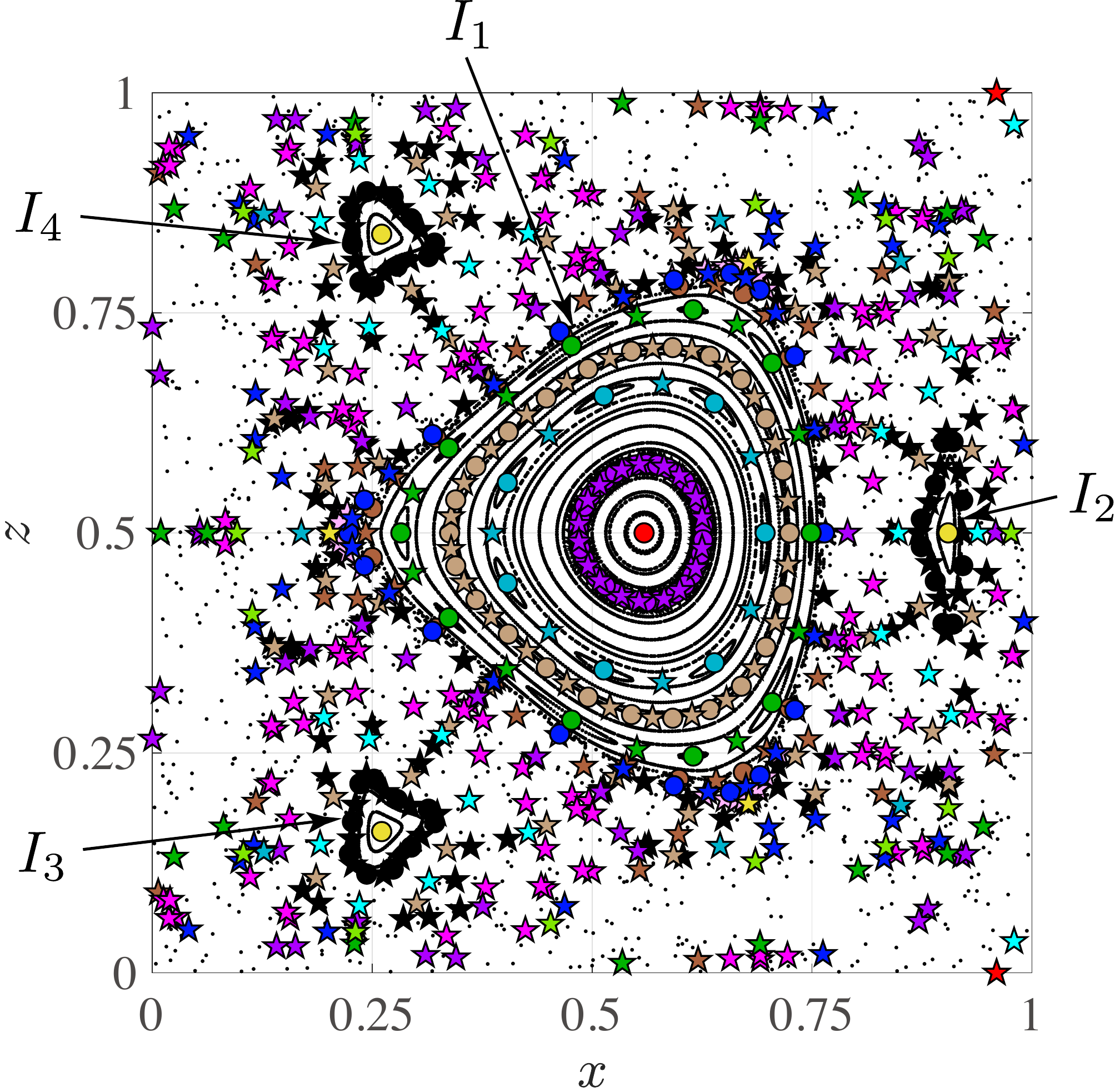}}
\hspace{3mm}
\vspace{-1mm}
\subfigure[$\varepsilon=0.3$]
{\includegraphics[scale=0.26]{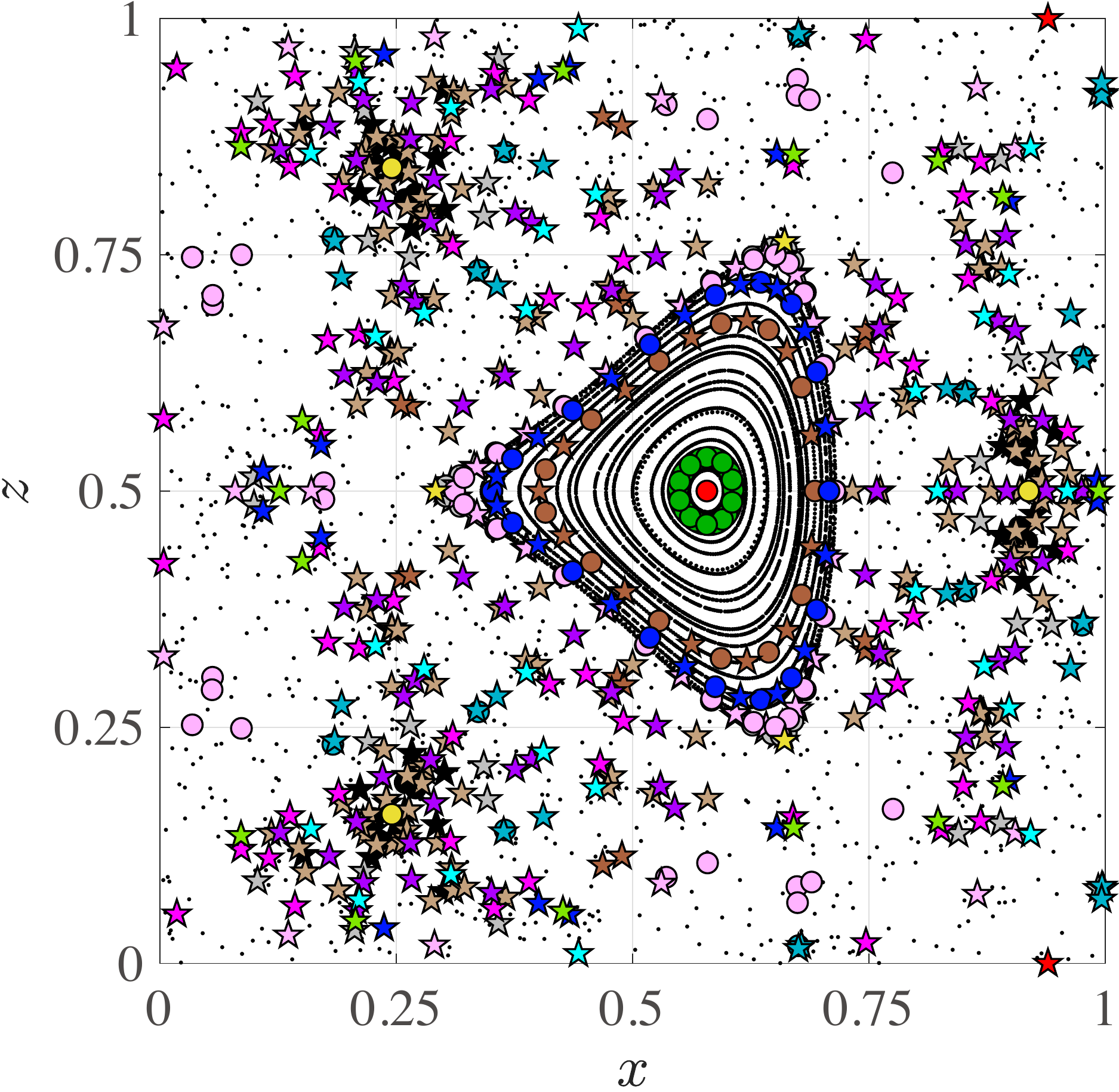}}
\vspace{-1mm}
\subfigure[$\varepsilon=0.4$]
{\includegraphics[scale=0.26]{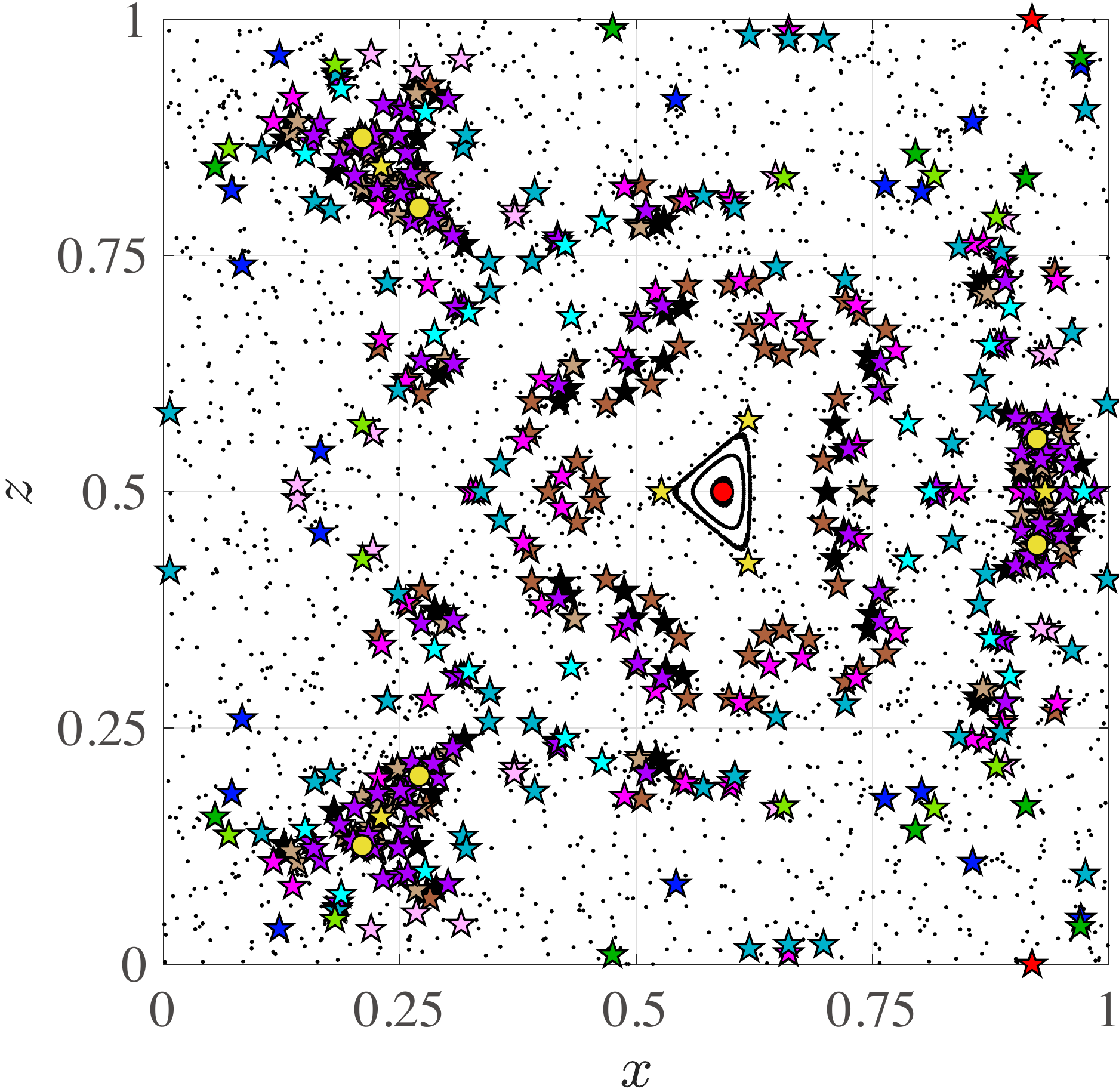}}
\hspace{3mm}
\subfigure[$\varepsilon=0.5$]
{\includegraphics[scale=0.26]{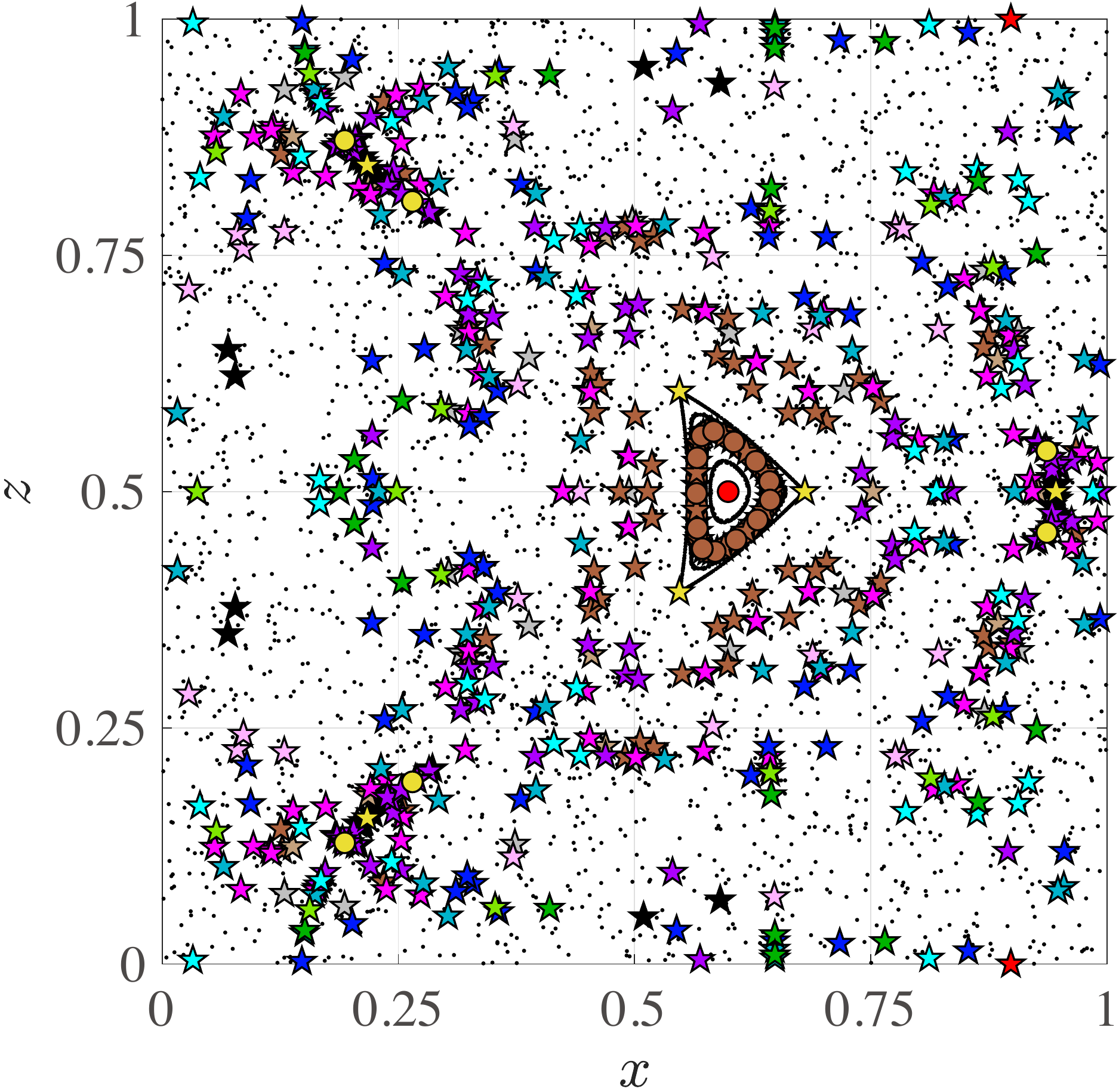}}
\subfigure
{\includegraphics[scale=0.3]{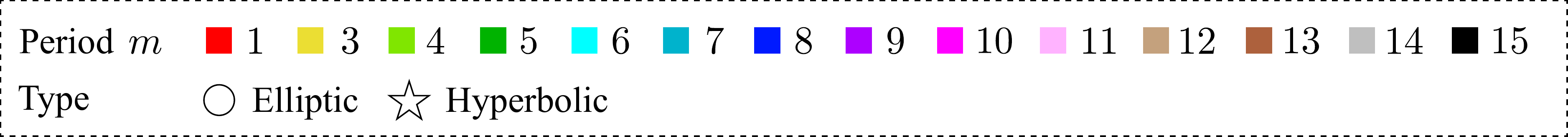}}
\vspace{-1mm}
\caption{Elliptic and hyperbolic periodic points ($\varepsilon=0.2,0.3,0.4,0.5$)}
\label{fig:pm_prd_epsilon}
\end{center}
\end{figure}

\vspace{-4mm}

\begin{figure}[H]
\begin{center}
{\includegraphics[scale=0.25]{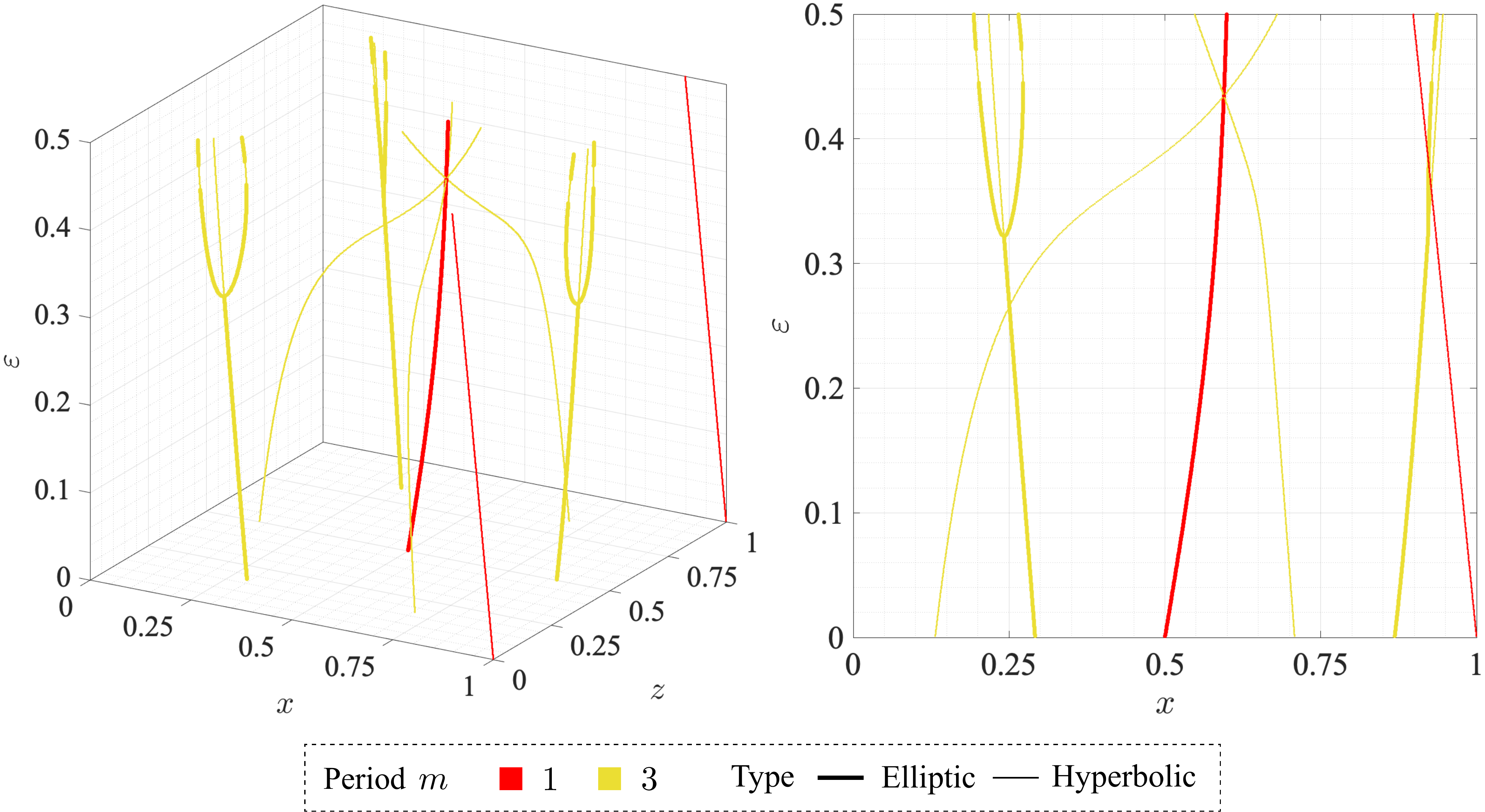}}
\caption{Bifurcations of 1 and 3-periodic points}
\label{fig:bifurcation_MP1_3}
\end{center}
\end{figure}

\begin{figure}[H]
\begin{center}
\subfigure[$\varepsilon=0.1$]
{\includegraphics[scale=0.26]{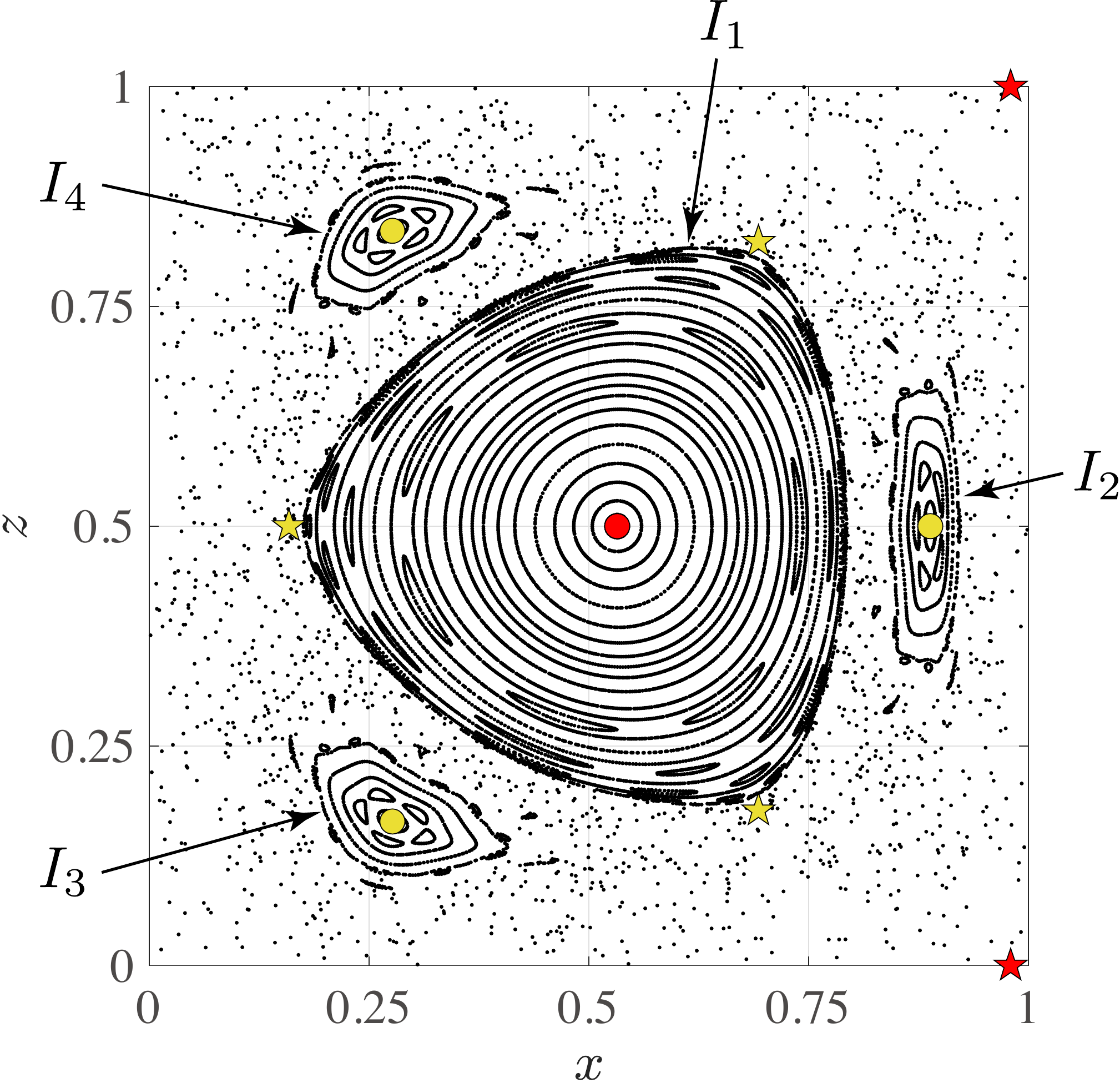}}
\subfigure[$\varepsilon=0.4$]
{\includegraphics[scale=0.26]{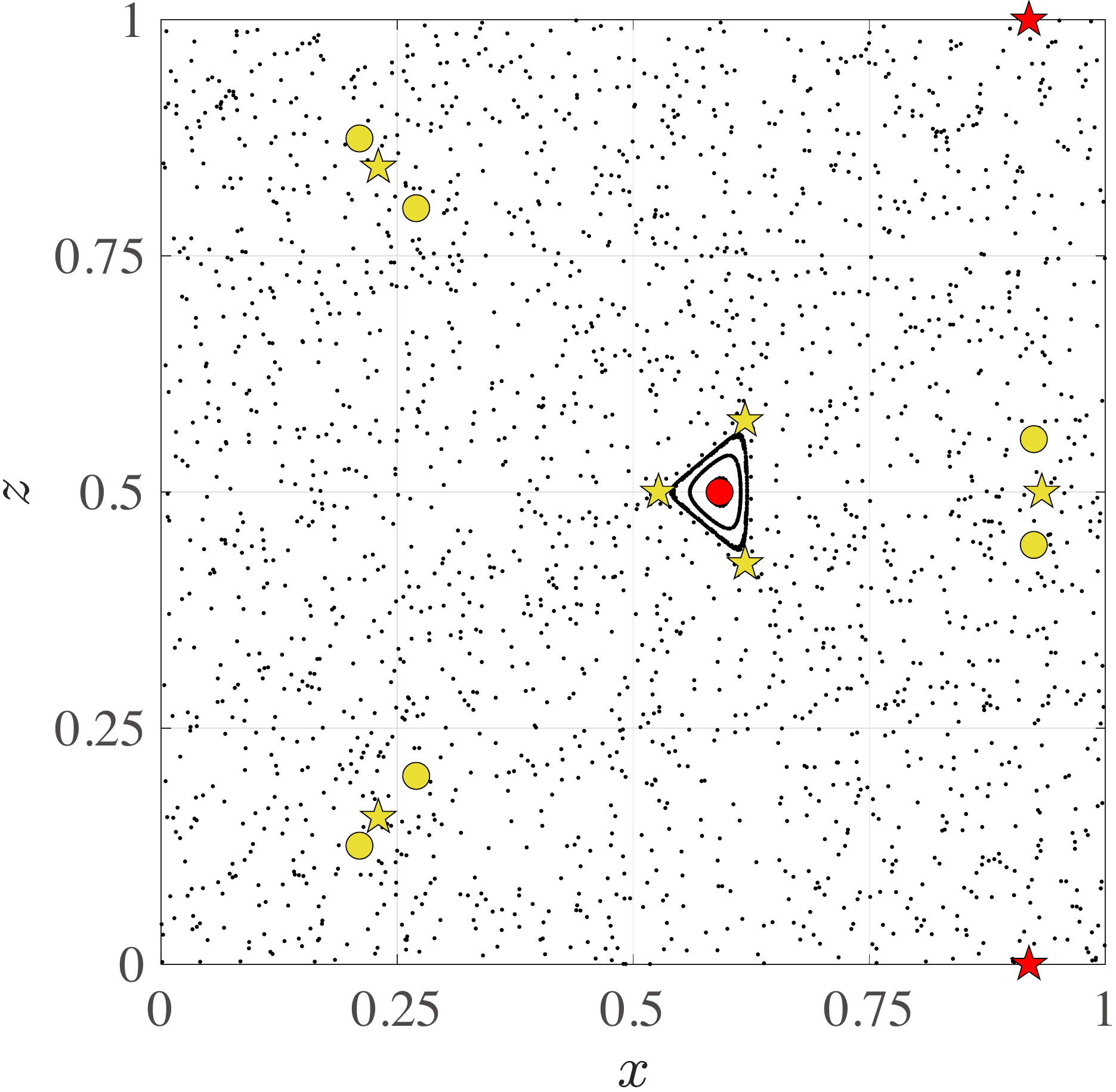}}
\subfigure
{\includegraphics[scale=0.31]{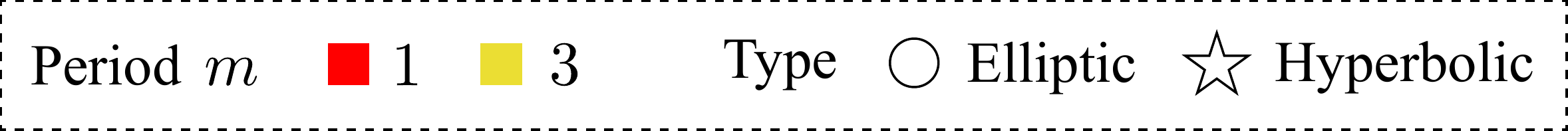}}
\caption{1 and 3-periodic points at $\varepsilon=0.1, 0.4$}
\label{fig:pm_prd_epsilon_MP1_3}
\end{center}
\end{figure}

Thus, the thick red branch of elliptic 1-periodic points in the middle of Fig.\ref{fig:bifurcation_MP1_3} and the three thin yellow branches of hyperbolic 3-periodic points, which cross with the red branch, correspond to those of the periodic points 
associated with $I_1$. Furthermore, Fig.\ref{fig:pm_prd_epsilon_MP1_3} indicates that an elliptic 3-periodic point appear
in the middle of each island $I_2, I_3,$ and $I_4$ when $\varepsilon$ is small but vary to two elliptic and one hyperbolic 
3-periodic points when $\varepsilon$ is increased. 
Thus, the three fork-shaped branches of elliptic 3-periodic points in Fig.\ref{fig:bifurcation_MP1_3}
correspond to those of the periodic points associated with islands $I_2, I_3,$ and $I_4$. 
The two straight branches of hyperbolic 1-periodic points on the wall of Fig.\ref{fig:bifurcation_MP1_3} 
are those of the 1-periodic points on the upper and lower boundaries of the convection. 

\paragraph{Bifurcations associated with KAM islands $I_1, I_2, I_3,$ and $I_4$.}
Next, we focus on the bifurcations associated with KAM islands $I_1, I_2, I_3,$ and $I_4$. 
First, we take a look at those of the main island $I_1$. As is shown in Fig.\ref{fig:pm_prd_epsilon}, 
the periodic points in $I_1$ appear along the KAM curves around an elliptic 1-periodic point. 
Thus, the mountainous structure depicted in Fig.\ref{fig:bifurcation_I1} may correspond to the bifurcations 
associated with $I_1$. Though the type of the periodic points are not illustrated here, 
it follows that many branches of various periods gather to the branch of the elliptic 1-periodic points. 
Especially, it is observed that the three branches of hyperbolic 3-periodic points at the corners of island $I_1$
appear around the outer side of the mountainous structure. 
Furthermore, since they cross with the branch of 1-periodic points at around $\varepsilon=0.432$, 
it seems that $I_1$ once disappear when the amplitude $\varepsilon$ is increased. 
We will analyze the bifurcations associated with $I_1$ more in detail in \S\ref{Sec:bifurcation_I1}.
Then, let us take a look at the bifurcations of islands $I_2, I_3,$ and $I_4$. 
It is found in our numerical computation that the bifurcations shown in Fig.\ref{fig:bifurcation_I234} 
may correspond to those associated with $I_2, I_3,$ and $I_4$. 
As can be seen, many branches of $3l$-periodic points $(l=2,3,4,5)$ grow from the fork-shaped 
branch of 3-periodic points to form the shapes of three broom tips standing upside down as in Fig.\ref{fig:bifurcation_MP1_3} and create three tree-like structures. 
We will clarify the structure of the bifurcations more in detail in \S\ref{Sec:bifurcation_I234}. 


\begin{figure}[H]
\begin{center}
{\includegraphics[scale=0.25]{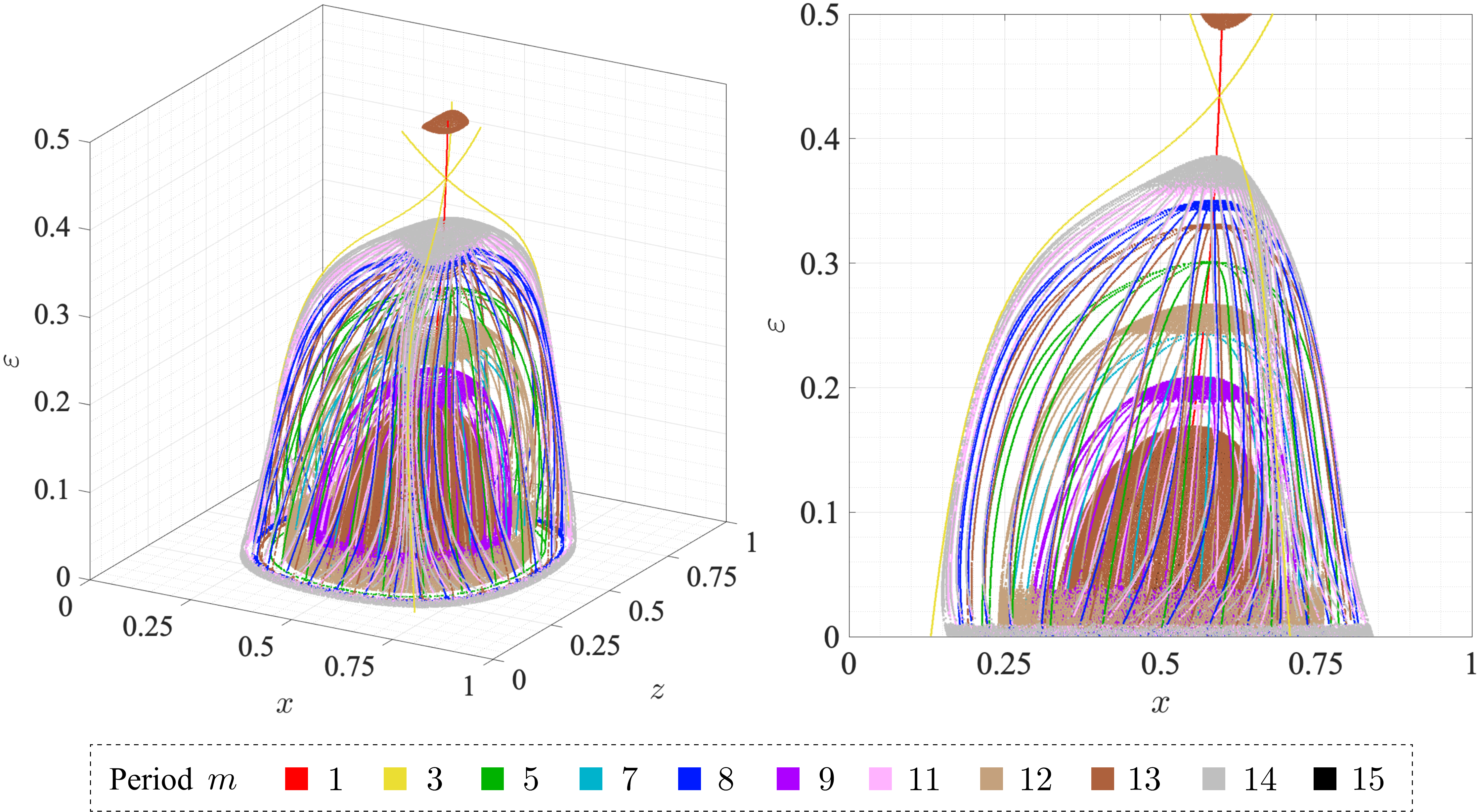}}
\caption{Mountainous structure in bifurcations of island $I_1$}
\label{fig:bifurcation_I1}
\end{center}
\end{figure}

\vspace{-2mm}

\begin{figure}[H]
\begin{center}
{\includegraphics[scale=0.25]{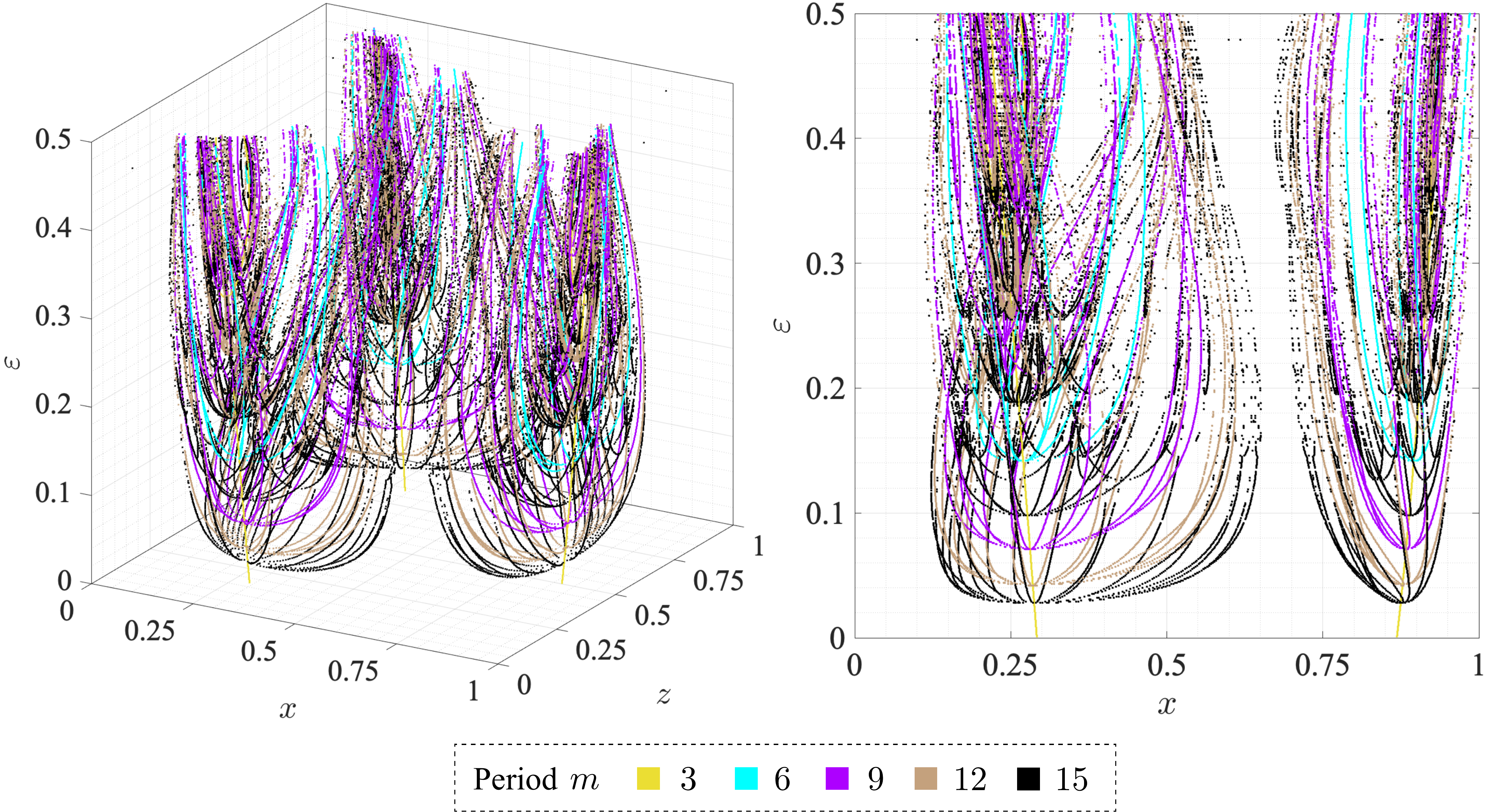}}
\caption{Bifurcations associated with island $I_2, I_3,$ and $I_4$}
\label{fig:bifurcation_I234}
\end{center}
\end{figure}


\subsection{Numerical algorithm for detecting bifurcation points}\label{Sec:comp}
Before we clarify the global structures of the $\varepsilon$-bifurcation diagram more in detail, let us briefly review 
the classification of bifurcations of periodic points and describe how each bifurcation point is detected in numerical computations. 

\paragraph{Classification of bifurcation points.}
Recall that multipliers $\mu$ of an $m$-periodic point are eigenvalues of the Jacobian matrix 
of the Poincar\'e return map
\begin{equation*}
J_\varepsilon(\mathbf{x})=\frac{\partial (P^{\theta_0}_\varepsilon)^m (\mathbf{x})}{\partial \mathbf{x}} \Biggr|_{\mathbf{x}=\mathbf{x}_0}, 
\end{equation*}
where $\mathbf{x}_0$ indicates the $m$-periodic point. 
According to the multipliers  $\mu$ of the $m$-periodic point at the bifurcation point (see, for instance, \cite{Ku2004}), 
the bifurcations of $m$-periodic points are classified into the following types:
\begin{itemize}
\item Fold bifurcation (also called, tangent or saddle-node bifurcation): $\mu=1$
\item Flip bifurcation (also called, period-doubling bifurcation): $\mu=-1$
\item Neimark-Sacker bifurcation (also called, Hopf bifurcation for maps): $|\mu|\!=\!1$ but $\mu \neq \pm 1$
\end{itemize}
In this paper, we mainly focus on the fold and flip bifurcations.

\paragraph{Computation of fold and flip bifurcation points.}
We shall show the numerical method for detecting the fold and flip bifurcation point of $m$-periodic points. To do this, we shall employ the numerical computation method that was developed by \cite{TsUeYoKa2012}; see also \cite{Ku2004}. Using the Poincare map $P_{\varepsilon}^{\theta_0}: \Sigma^{\theta_{0}} \to \Sigma^{\theta_{0}}$, the following two conditions have to be satisfied at the bifurcation point for some $m$-periodic point $\mathbf{x}_0$:
\medskip

\begin{itemize}
\item[(i)]\textbf{Condition for $m$-periodic points.}
Recall that associated with the vector field of the autonomous Hamiltonian system in \eqref{ExtHamEq}, we can uniquely define the flow $\phi^{\varepsilon}: \mathbb{R} \times \mathcal{M} \to \mathcal {M};~  (t,x_0,z_0,\theta_0) \to (x_t,z_t,\theta_t)=\phi^\varepsilon (t,x_0,z_0,\theta_0)$ for some given parameter $\varepsilon \in \mathbb{R}$. Then, a diffeomorphism $\phi^\varepsilon_t: \mathcal {M} \to \mathcal {M}; ~(x_0,z_0,\theta_0) \mapsto (x_t,z_t,\theta_t)=\phi^\varepsilon_t(x_0,z_0,\theta_0)$ can be given for each fixed $t$.
\medskip

Recall also that we can define the Poincar\'e $m$-return map by
$$
(P_{\varepsilon}^{\theta_0})^m:=\phi^\varepsilon_{mT}\Bigr\rvert_{\Sigma^{\theta_{0}}} : \Sigma^{\theta_{0}} \to \Sigma^{\theta_{0}},
$$ 
which is locally given by
%
\begin{eqnarray*}
(x(0)=x_0, z(0)=z_0, \theta(0)=\theta_{0}) \hspace{40mm}\\
 \hspace{30mm}\mapsto (x (mT), z (mT), \theta(mT)=\theta_{0}+2\pi m\equiv \theta_{0}).
\end{eqnarray*}
Therefore, the condition that some point $\mathbf{x}_0=(x_0,z_0) \in \Sigma^{\theta_{0}} $ becomes the $m$-periodic point is given by
\begin{equation}\label{mPerCon}
(P_{\varepsilon}^{\theta_0})^m (\mathbf{x}_0)=\mathbf{x}_0.
\end{equation}
\item[(ii)]\textbf{Condition for bifurcation points.}
Suppose that $\mathbf{x}_0$ is an $m$-periodic point on $\Sigma^{\theta_{0}}$ and consider to find a bifurcation point for $\mathbf{x}_0$ associated with the parameter $\varepsilon$, where we need to vary $\varepsilon$ to detect the bifurcation point. Recall that the Poincar\'e $m$-return map is given by, for some $\mathbf{x}^{(l)} \in \Sigma^{\theta_{0}}$ and with fixed $\varepsilon$,
$$
\mathbf{x}^{(l+1)}=(P^{\theta_0}_{\varepsilon})^m (\mathbf{x}^{(l)}),\;\;l=0,1,2,\cdots.
$$
Let $\mathbf{x}_0=(x_0,z_0) \in \Sigma^{\theta_{0}} $ be an $m$-periodic solution and we define the variation of $\mathbf{x}^{(l)}$ associated with $\mathbf{x}_0$, i.e., a small deviation from $\mathbf{x}_0$ by 
$$
\mathbf{w}^{(l)}:=\mathbf{x}^{(l)}-\mathbf{x}_0.
$$
Then, by definition $\mathbf{x}^{(l+1)}=\mathbf{x}_0+\mathbf{w}^{(l+1)}$, and it follows by Tayler expansion and by neglecting the higher-order terms that the variational equations may be given as
\begin{equation}\label{VarEqn}
\mathbf{w}^{(l+1)}= J_\varepsilon (\mathbf{x}_0) ~ \mathbf{w}^{(l)},
\end{equation}
where 
\begin{equation*}
J_{\varepsilon}(\mathbf{x}_0)=\frac{\partial (P_{\varepsilon}^{\theta_0})^m (\mathbf{x})} {\partial \mathbf{x}}\Biggr|_{\mathbf{x} = \mathbf{x}_0}.
\end{equation*}
The characteristic equation of \eqref{VarEqn} is
\begin{equation}\label{ChaEqn}
\mathrm{det}\,\left( J_{\varepsilon}(\mathbf{x}_0)-\mu \mathbf{I}\right)=0,
\end{equation}
where $\mathbf{I}$ denotes the unit matrix and $\mu$ a multiplier that corresponds to an eigenvalue. 
\end{itemize}
Notice that the parameter $\varepsilon$ is fixed in equations \eqref{mPerCon}, \eqref{VarEqn} and \eqref{ChaEqn}. On the other hand, the $m$-periodic point may be bifurcated at some $\varepsilon_0$ when 
 $\mu$ satisfy $|\mu|=1$; for instance, the fold and flip bifurcations can be occurred when $\mu=1$ and $\mu= -1$ respectively.

Thus, when a bifurcation associated with some specific $\mu_0 \in \mathbb{R}$ for an $m$-periodic point $\mathbf{x}_0=(x_0,z_0) \in \Sigma^{\theta_{0}}$ occurs at some $\varepsilon_0$,  the following set of $\mu_0$-dependent nonlinear algebraic equations \eqref{mPerCon} and \eqref{ChaEqn} holds:
\begin{equation}\label{BifPointCond}
G_{\mu_{0}}(\mathbf{x}_0, \varepsilon_0)=
\begin{bmatrix}\;
F(\mathbf{x}_0, \varepsilon_0)\\[3mm]
g_{\mu_{0}}(\mathbf{x}_0, \varepsilon_0)
\;
\end{bmatrix}
=\mathbf{0},
\end{equation}
where we define the map $F: \Sigma^{\theta_{0}} \times \mathbb{R}\to \mathbb{R}^2$ by, for each $(\mathbf{x}_0, \varepsilon_0) \in  \Sigma^{\theta_{0}} \times \mathbb{R}$,
\begin{equation*}
F(\mathbf{x}_0, \varepsilon_0):=\mathbf{x}_0 - (P_{{\varepsilon}_0}^{\theta_0})^m (\mathbf{x}_0),
\end{equation*}
and also the map $g_{\mu_{0}}: \Sigma^{\theta_{0}} \times \mathbb{R}\to \mathbb{R}$ by
\begin{equation*}
g_{\mu_{0}}(\mathbf{x}_0, \varepsilon_0):=\mathrm{det}\;\left( J_{\varepsilon_0}(\mathbf{x}_0)-\mu_0 \mathbf{I}\right).
\end{equation*}
In the above, notice that $\varepsilon_0$ is treated as a variable together with $\mathbf{x}_0$.
In other words,  in order to detect a bifurcation point associated with some $\mu_0$ that satisfies $|\mu_0|=1$ for the $m$-periodic point $\mathbf{x}_0$ together with the specific parameter $\varepsilon_0$, we have to find a solution $(\mathbf{x}_0, \varepsilon_0)$ that satisfies the nonlinear algebraic equations \eqref{BifPointCond}. 
\medskip
 
For numerical computations, we shall employ Newton's method again as follows.

\begin{framed}\paragraph{\textsf{Numerical algorithm for detecting the fold or flip bifurcation point:}\vspace{2mm}}
\begin{itemize}
\item[(1)]
Set $\mu_0=1$ for the fold bifurcation or $\mu_0=-1$ for the flip bifurcation.  
\item[(2)]Set $k=0$ with an initial approximation ${\bf y}^{(0)}_0=({\bf x}^{(0)}_0, \varepsilon^{(0)}_0)$ for some required bifurcation point ${\bf y}_0=({\bf x}_0, \varepsilon_0)$.
\item[(3)] Set $k:=k+1$ and compute the $k$-th approximation by
\begin{equation*}
\begin{split}
{\bf y_0}^{(k)} &:= {\bf y_0}^{(k-1)} - \left( \frac{\partial G_{\mu_0}({\bf y_0})}{\partial {\bf y_0}}\Biggr|_{{\bf y_0} = {\bf y_0}^{(k-1)}} \right)^{-1}G_{\mu_0}({ \bf y_0}^{(k-1)}),
\end{split}
\end{equation*}
where the Jacobian matrix is numerically approximated by the central difference scheme. 
\item[(4)]
If $|G_{\mu_0}({ \bf y_0}^{(k)})|< \delta,$ where the convergence radius is set to $\delta=10^{-10}$, then 
the computation ends up and the bifurcation point for the $m$-periodic point is to be detected as ${\bf y}_0={\bf y}_0^{(k)}$. 
\item[(5)] Otherwise, return to (3) in order to iterate the computation until convergence.
\end{itemize}
\end{framed}

\begin{remark}\rm
The initial approximation ${\bf y}^{(0)}_0=({\bf x}^{(0)}_0, \varepsilon^{(0)}_0)$ in the Newton's method is obtained from the $\varepsilon$-parameter bifurcation diagram.
\end{remark}

\subsection{Bifurcations associated with KAM island $I_1$}\label{Sec:bifurcation_I1}
In this subsection, we investigate the bifurcations of periodic points associated with the main KAM island $I_1$. 
As we have seen in Fig.\ref{fig:bifurcation_I1}, many branches of periodic points with various periods gather to 
the branch of 1-periodic points at the center of island $I_1$. Let us first show the bifurcation points numerically detected in our computation, 
and then illustrate how the periodic orbits vary with $\varepsilon$ by taking a look at the 7-periodic orbits for example.

\paragraph{Fold bifurcations associated with $I_1$.}
Fig.\ref{fig:I1_all} shows from $z$ direction the bifurcation points numerically detected in the $\varepsilon$-bifurcation diagram 
associated with $I_1$, where the branches of elliptic and hyperbolic periodic points are depicted in the same way. 
Each bifurcation point of $m$-periodic points is indicated with a circle in magenta. 
The amplitude $\varepsilon$ for each point is also shown beside them with the period $m$ and the type of the bifurcation. 
As can be seen, it was numerically clarified that the $m$-periodic points bifurcate in a fold bifurcation 
when they coalesce with the 1-periodic point at the center of $I_1$. 
Note that the 1-periodic points themselves do not seem to bifurcate when the $m$-periodic points bifurcate in a fold bifurcation.

\begin{figure}[H]
\begin{center}
\includegraphics[scale=0.25]{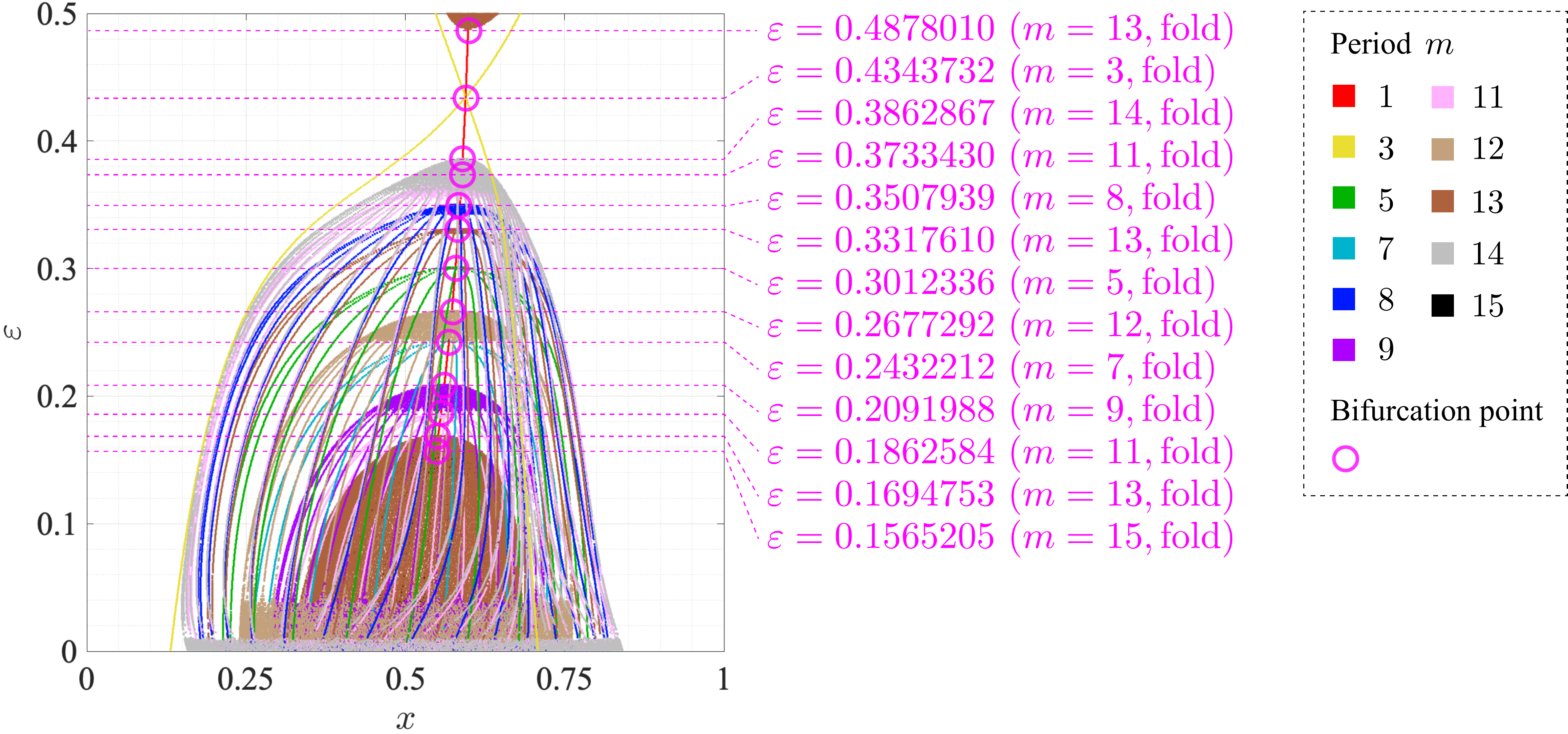}
\caption{$\varepsilon$-bifurcation diagram of periodic points associated with $I_1$}
\label{fig:I1_all}
\end{center}
\end{figure}

\vspace{-10mm}

\paragraph{Fold bifurcations of 7-periodic points.}
Next, let us investigate how the periodic orbits vary with $\varepsilon$ near the fold bifurcation point. 
Here, we take a look at the 7-periodic orbits for example. 
Fig.\ref{fig:I1_MP1_7} illustrates the $\varepsilon$-bifurcation diagram of 1 and 7-periodic points in $I_1$, 
where the branches of elliptic and hyperbolic periodic points are depicted in thick and thin lines respectively. 
Fig.\ref{fig:I1_MP1_7_e0.2} also shows the 1 and 7-periodic points on the Poincar\'e section $\Sigma^{\theta_0}$ 
at $\varepsilon=0.2$ and the projection of the associated periodic orbits onto the phase space $M$. 
As can be seen in Fig.\ref{fig:I1_MP1_7_e0.2}, elliptic and hyperbolic 7-periodic points appear seven each in addition to the 1-periodic point. 
It follows that stable and unstable 7-periodic orbits appear one each in addition to a stable 1-periodic orbit. 
The blue points in circles and stars in Fig.\ref{fig:I1_MP1_7_e0.2_point} correspond to the points of the stable and unstable 
7-periodic orbits respectively, while the red circle point corresponds to the point of the stable 1-periodic orbit. 
It is observed that the resonance condition of the 1 and 7-periodic orbits are $|n/m|=1$ and $|n/m|=3/7$ respectively.

However, when the amplitude is increased from $\varepsilon=0.2$, the resonance condition of the unstable 7-periodic 
orbit varies to $|n/m|=5/7$ at around $\varepsilon=0.213$; the case $\varepsilon=0.23$ is illustrated in Fig.\ref{fig:I1_MP1_7_e0.23}. 
Furthermore, right before the bifurcation point at around around $\varepsilon=0.237$, 
the resonance condition of both the stable and unstable 7-periodic orbits varies to $|n/m|=7/7=1$, 
which corresponds to that of the 1-periodic orbit; the case $\varepsilon=0.242$ is depicted in Fig.\ref{fig:I1_MP1_7_e0.242}.
Therefore, it seems that the 7-periodic orbits disappear at the bifurcation point and vary to a 1-periodic orbit. 
Further, it is observed in our numerical computation that the projection of the 7-periodic orbits associated with $I_1$ is symmetric

\begin{figure}[H]
\begin{center}
\includegraphics[scale=0.25]{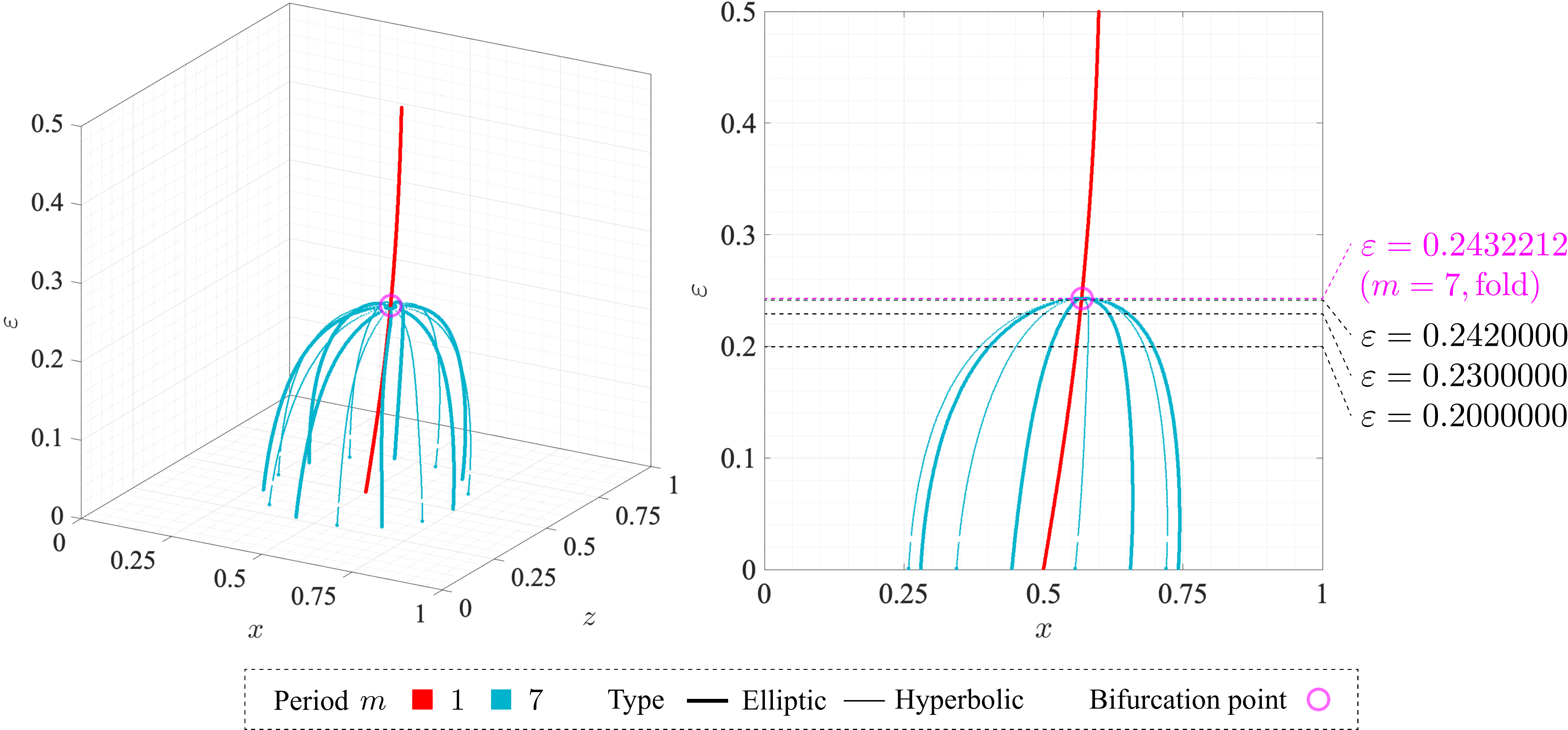}
\caption{$\varepsilon$-bifurcation diagram of 1 and 7-periodic points associated with $I_1$}
\label{fig:I1_MP1_7}
\end{center}
\end{figure}

\vspace{-6mm}

\begin{figure}[H]
\begin{center}
\subfigure[1 and 7-periodic points]
{\includegraphics[scale=0.25]{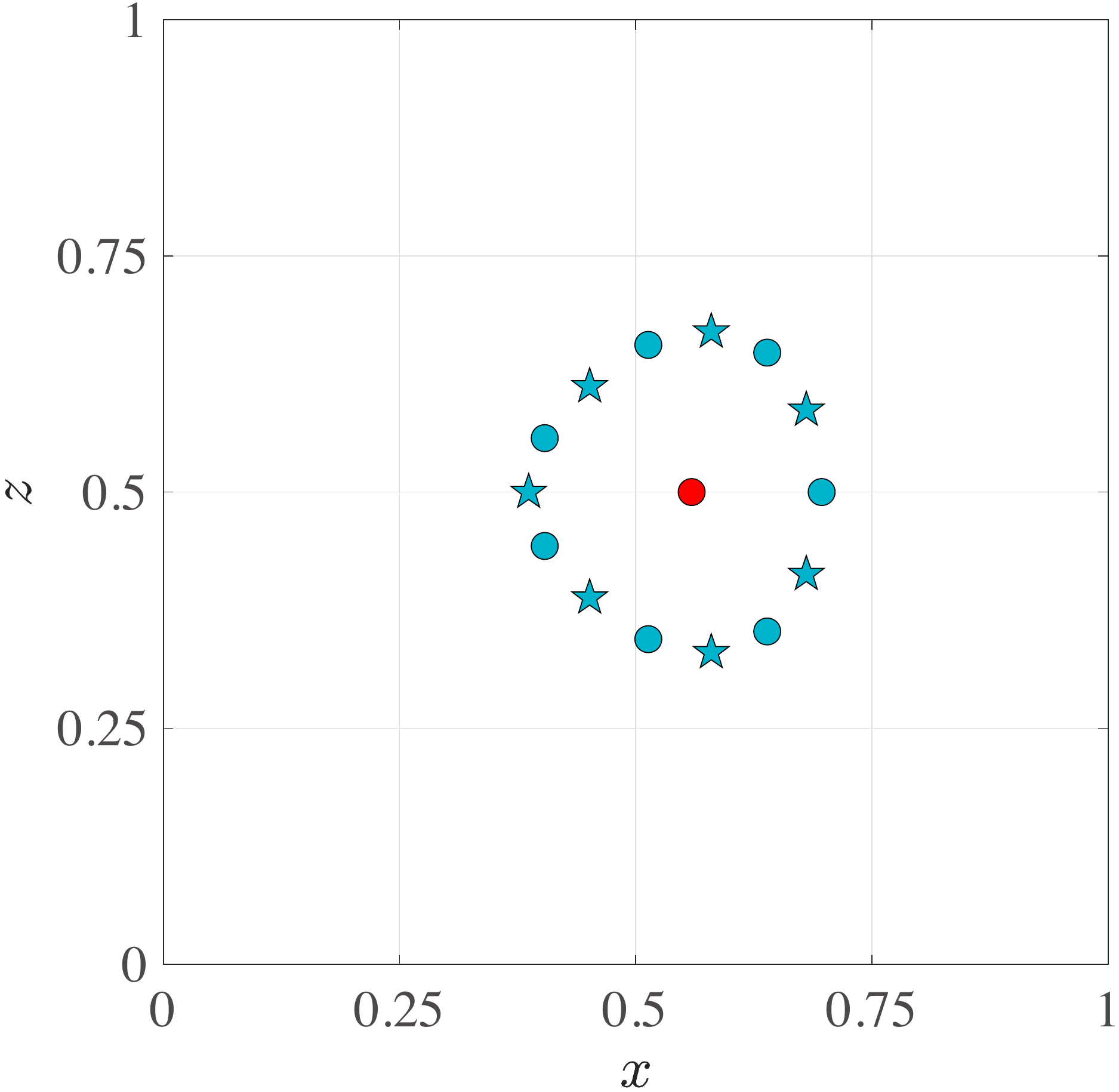}
\label{fig:I1_MP1_7_e0.2_point}}
\hspace{5mm}
\subfigure[Stable 1-periodic orbit]
{\includegraphics[scale=0.25]{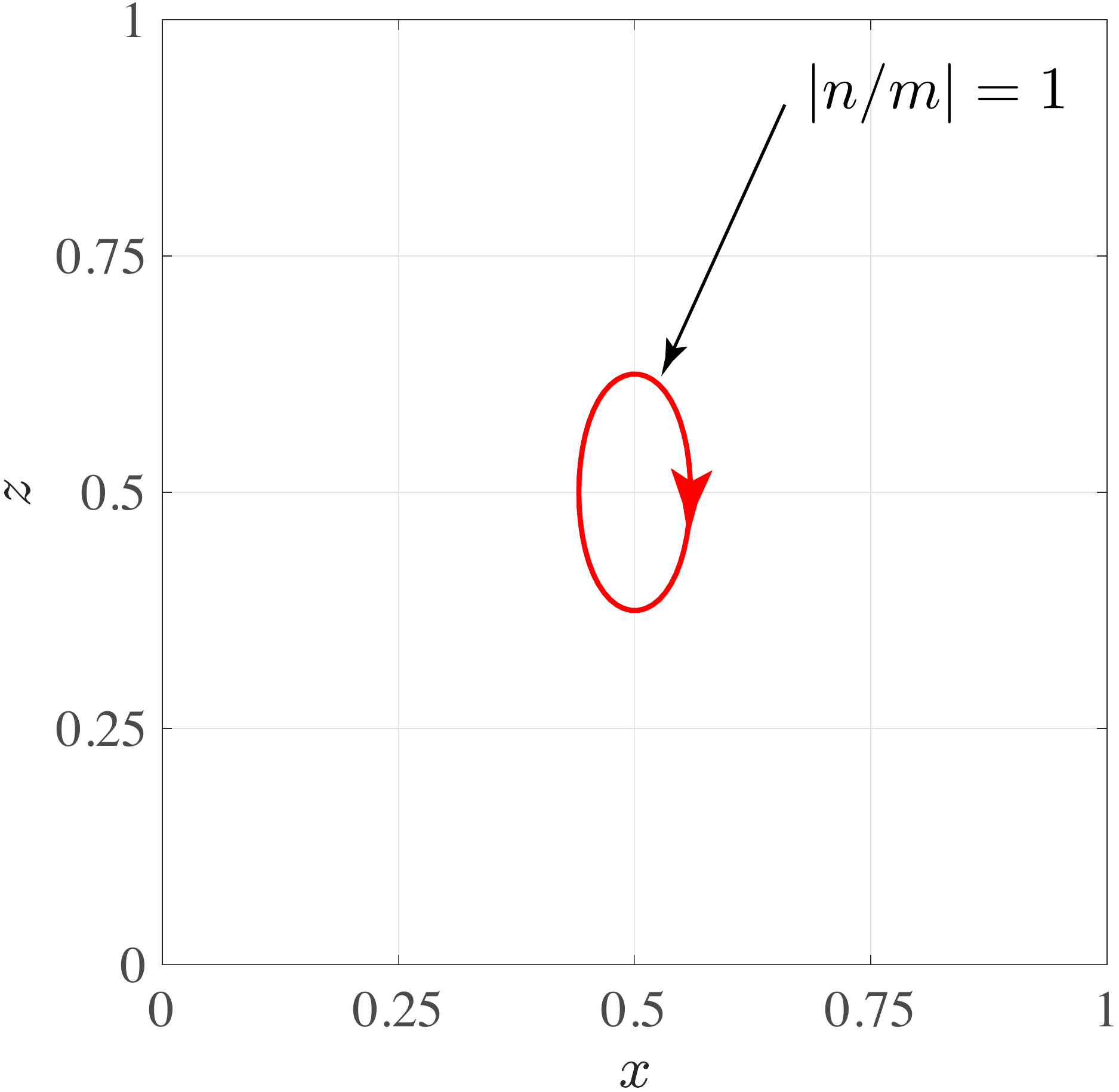}}
\subfigure[Stable 7-periodic orbit]
{\includegraphics[scale=0.25]{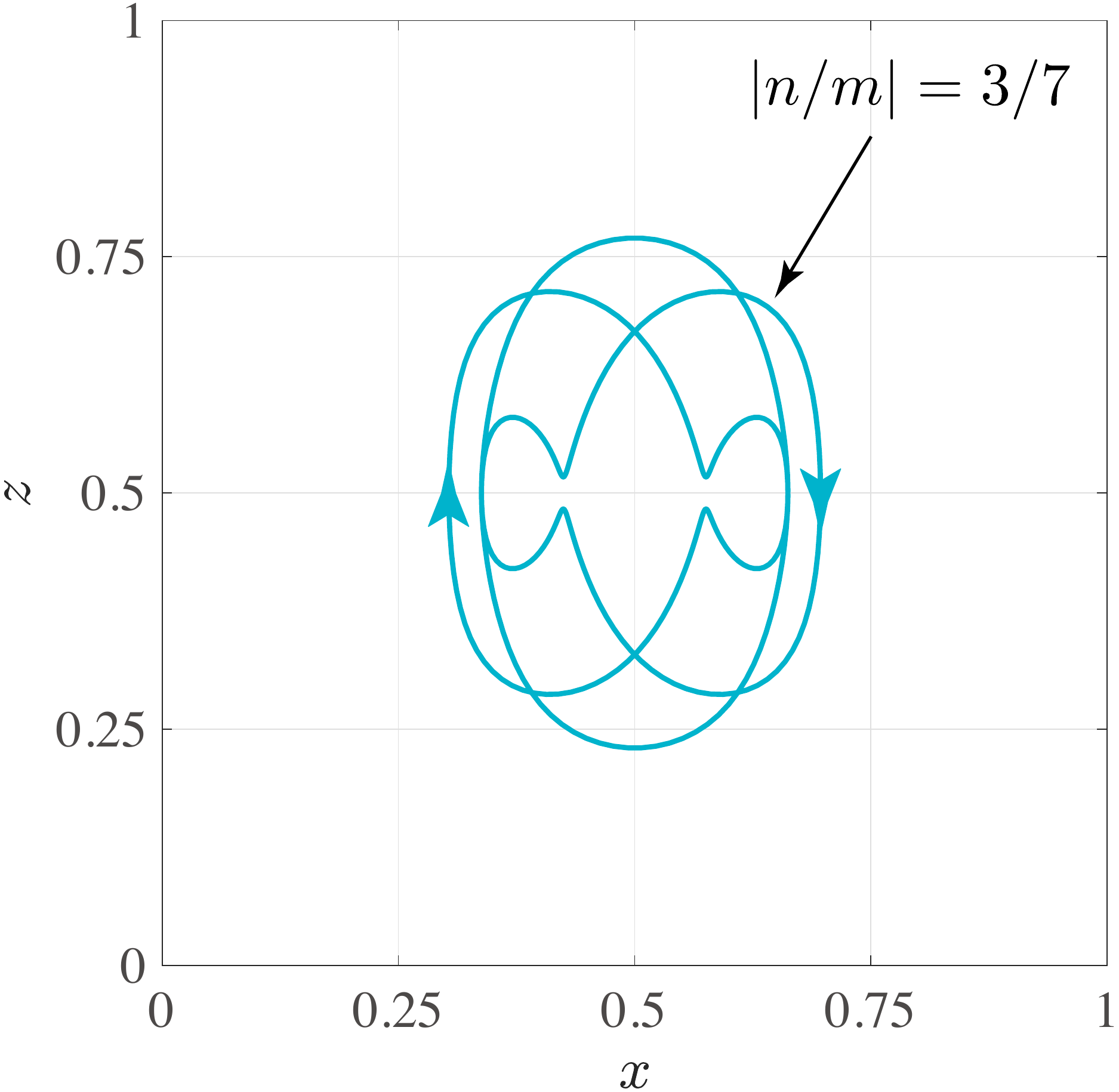}}
\hspace{5mm}
\subfigure[Unstable 7-periodic orbit]
{\includegraphics[scale=0.25]{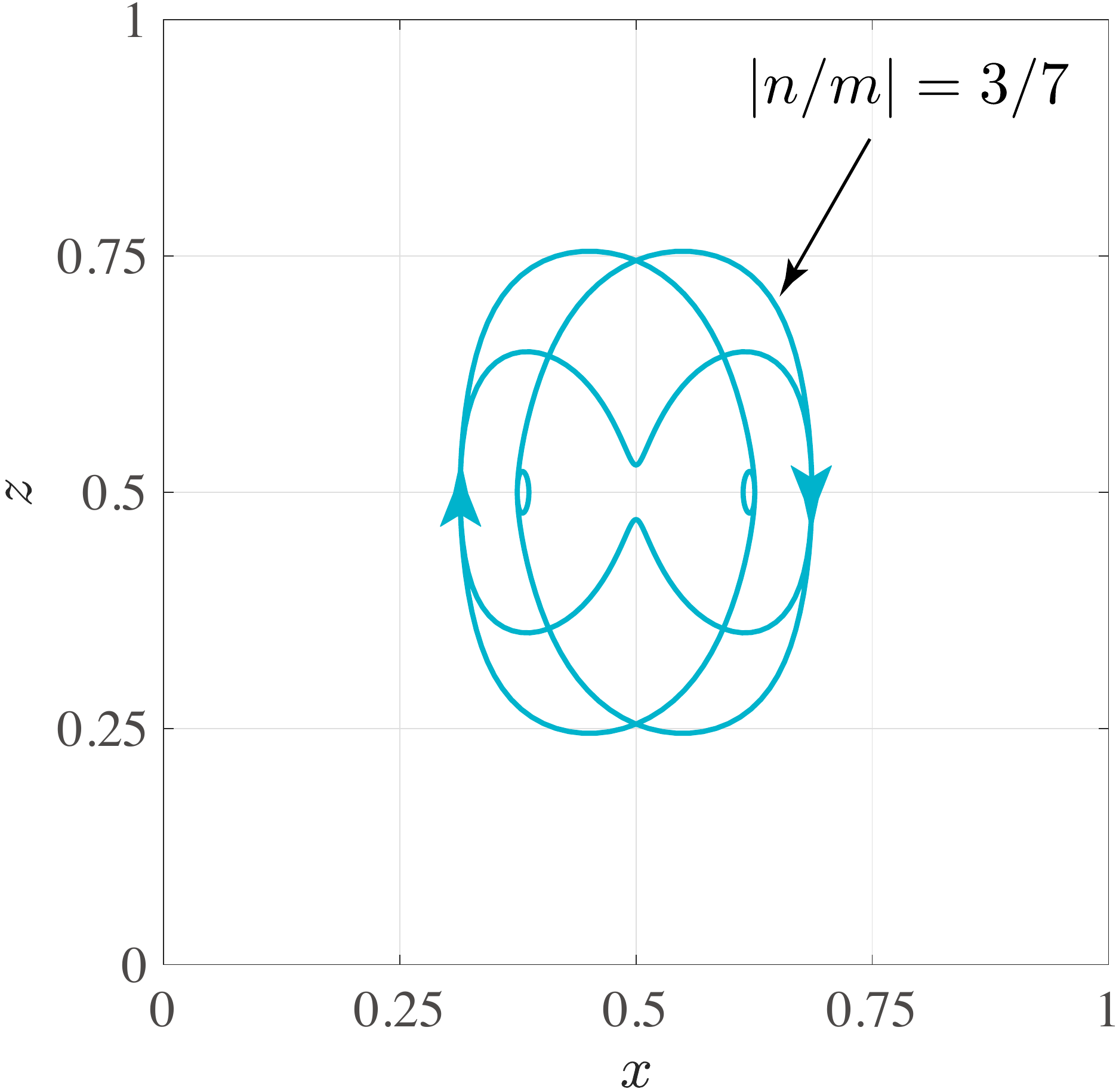}}
\caption{1 and 7-periodic points and the projection of their orbits at $\varepsilon=0.2$}
\label{fig:I1_MP1_7_e0.2}
\end{center}
\end{figure}

\begin{figure}[H]
\begin{center}
\subfigure[Stable 7-periodic orbit]
{\includegraphics[scale=0.25]{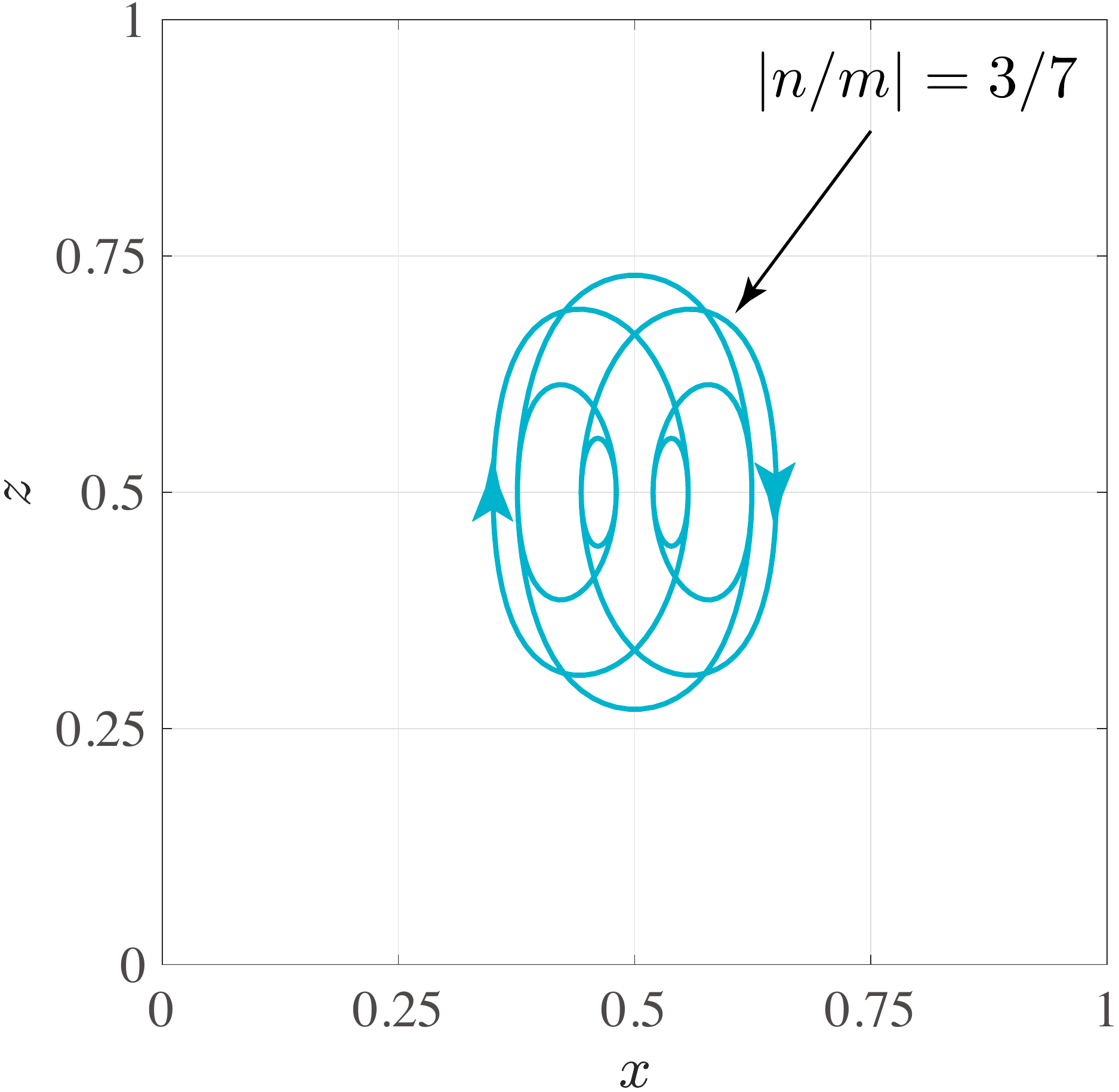}}
\hspace{5mm}
\subfigure[Unstable 7-periodic orbit]
{\includegraphics[scale=0.25]{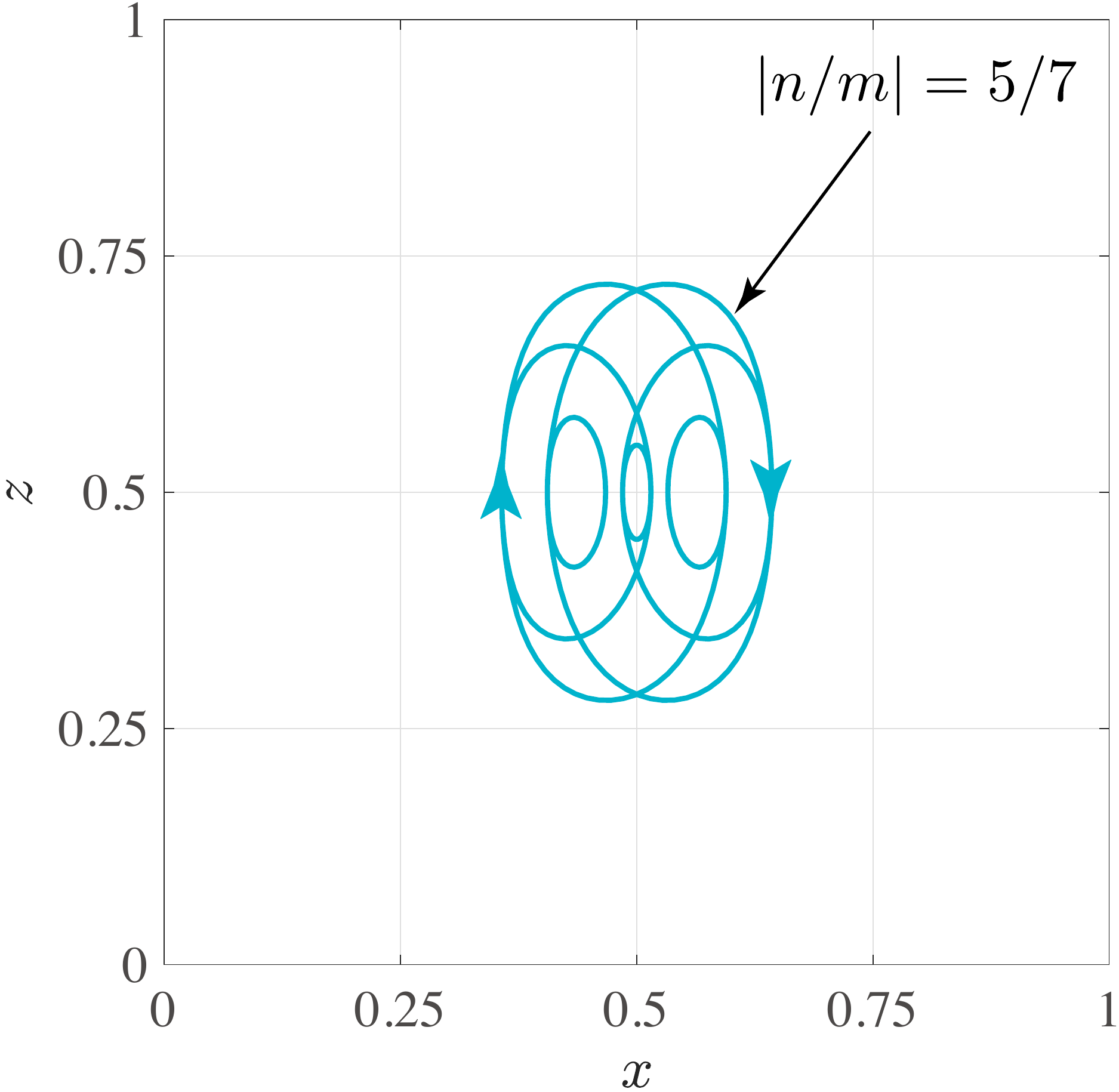}}
\caption{The projection of 7-periodic orbits at $\varepsilon=0.23$}
\label{fig:I1_MP1_7_e0.23}
\end{center}
\end{figure}

\vspace{-8mm}

\begin{figure}[H]
\begin{center}
\subfigure[Stable 7-periodic orbit]
{\includegraphics[scale=0.25]{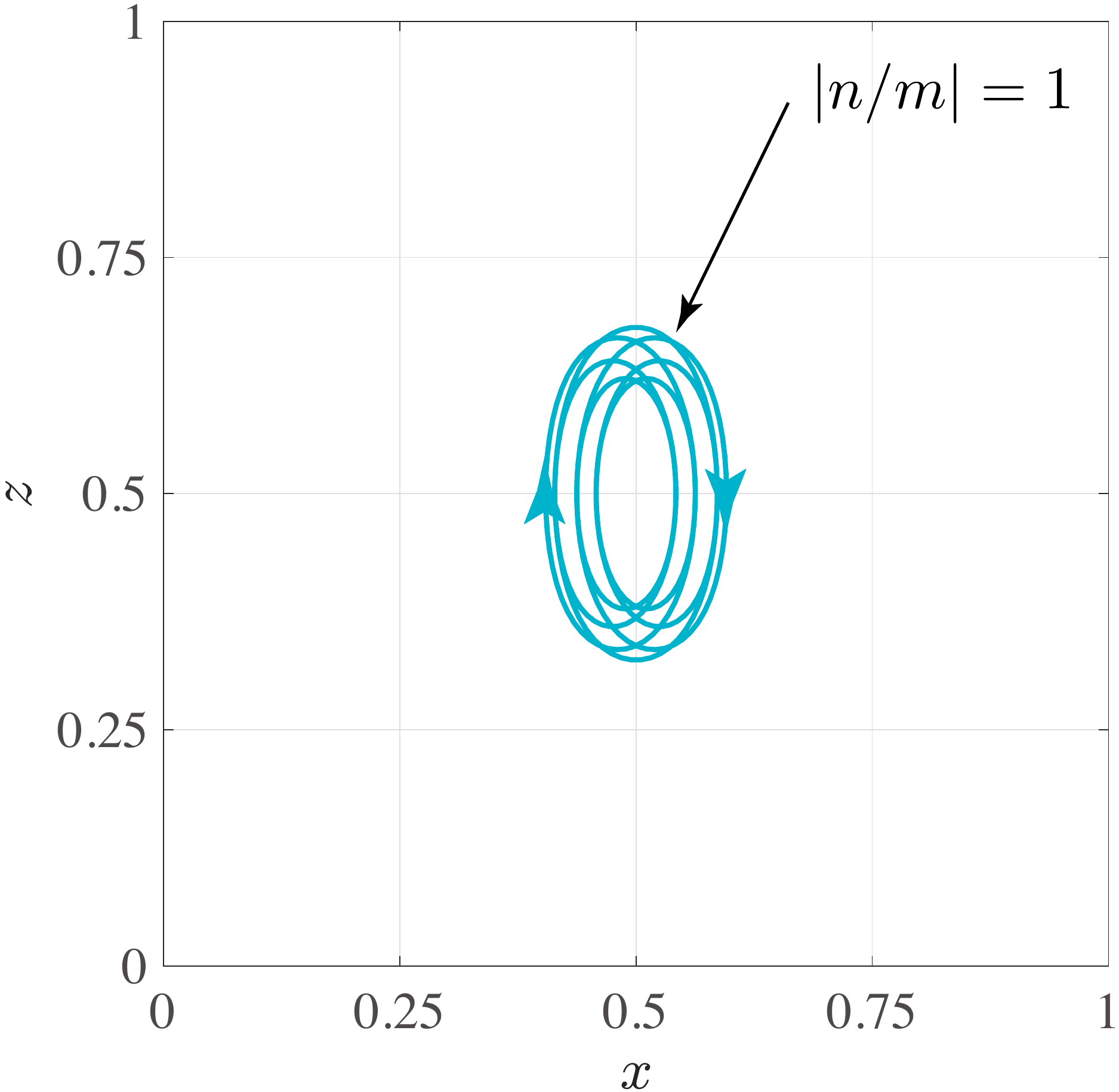}}
\hspace{5mm}
\subfigure[Unstable 7-periodic orbit]
{\includegraphics[scale=0.25]{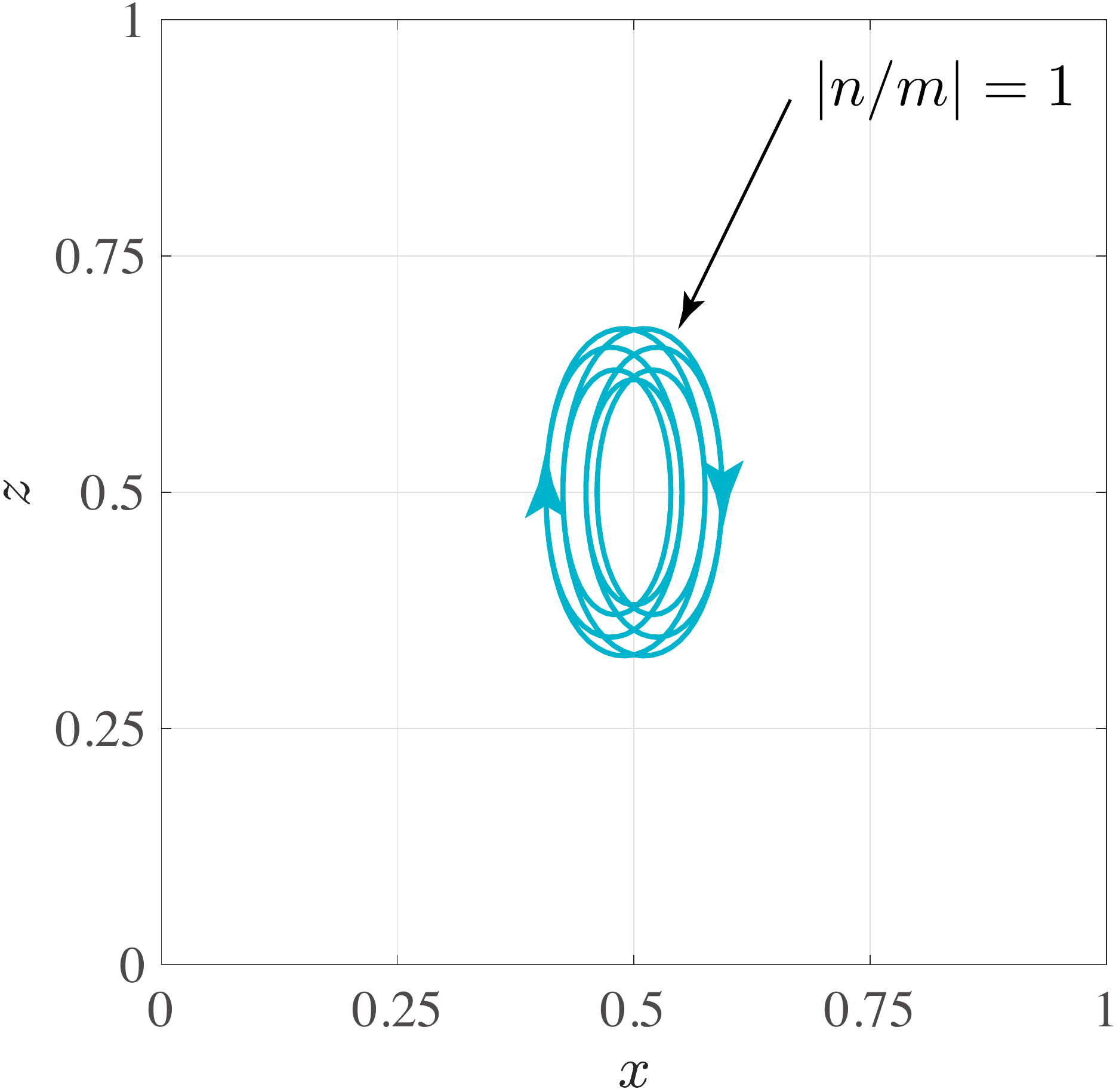}}
\caption{The projection of 7-periodic orbits at $\varepsilon=0.242$}
\label{fig:I1_MP1_7_e0.242}
\end{center}
\end{figure}

\vspace{-2mm}

\noindent
with respect to the horizontal and vertical center lines of the cell regardless of the amplitude $\varepsilon$. 
It is consistent with Theorem \ref{SymTheorem} that the period $m=7$ and the winding number $n=-3, -5,$ and -7 are odd.
The other $m$-periodic orbits associated with $I_1$ vary similarly to the 7-periodic orbits, 
which indicates that they disappear one by one when $\varepsilon$ is increased.

\subsection{Bifurcations associated with KAM islands $I_2, I_3,$ and $I_4$}\label{Sec:bifurcation_I234}
In this subsection, we analyze the bifurcations associated with the three KAM islands $I_2, I_3,$ and $I_4$ 
around the main island $I_1$. Let us first take a look at the bifurcations of 3-periodic points, 
and then investigate those of 6, 9, 12, and 15-periodic points.

\paragraph{Fold and flip bifurcations of 3-periodic points.}
Fig.\ref{fig:I234_MP3_6} illustrates the $\varepsilon$-bifurcation diagram of 3 and 6-periodic points 
associated with islands $I_2, I_3,$ and $I_4$. Let us first focus on the bifurcation of 3-periodic points which is depicted in yellow. 
As can be seen, the branches of 3-periodic points bifurcate similarly to a fork at around $\varepsilon=0.321$. 
It is found in our computation that a fold bifurcation occurs at $\varepsilon=0.3217135$. 
Let us take a look at how the 3-periodic points and the projection of the associated periodic orbits vary by the bifurcation. 
Fig.\ref{fig:I234_MP3_before} and Fig.\ref{fig:I234_MP3_after} show those at $\varepsilon=0.3$ and $\varepsilon=0.34$. 
We can see that three elliptic 3-periodic points appear at $\varepsilon=0.3$.
It follows that one stable 3-periodic orbit appears at $\varepsilon=0.3$.
However, as we increase $\varepsilon$, each elliptic 3-periodic point varies to a hyperbolic one 
at the bifurcation point and two more elliptic 3-periodic points appear in the neighborhood. 
This denotes that the stable 3-periodic orbit varies to an unstable 3-periodic orbit at the bifurcation point 
and two more stable 3-periodic orbits appear at the same time. 
We denote the two new stable 3-periodic orbits by $c_a$ and $c_b$, and label the associated elliptic 
3-periodic points in Fig.\ref{fig:I234_MP3_after} as $a$ and $b$ respectively in order to 
show the correspondence between the periodic points and their orbits. 
It is observed that the original stable 3-periodic orbit is symmetric with respect to the horizontal and vertical center 
lines of the cell, but those of the two new stable ones are only symmetric with respect to the vertical one. 
It seems that the stable orbits loose one of the symmetric properties by the bifurcation. 

Furthermore, it is clarified in our computation that the 3-periodic points bifurcate 
in a flip bifurcation at $\varepsilon=0.4453380$ and $\varepsilon=0.4713999$. 
Since the periodic orbits vary around these bifurcation points similarly to those around the flip bifurcations
of 4-periodic points discussed in \S\ref{Sec:others}, we do not discuss this here.

\begin{figure}[H]
\begin{center}
\includegraphics[scale=0.25]{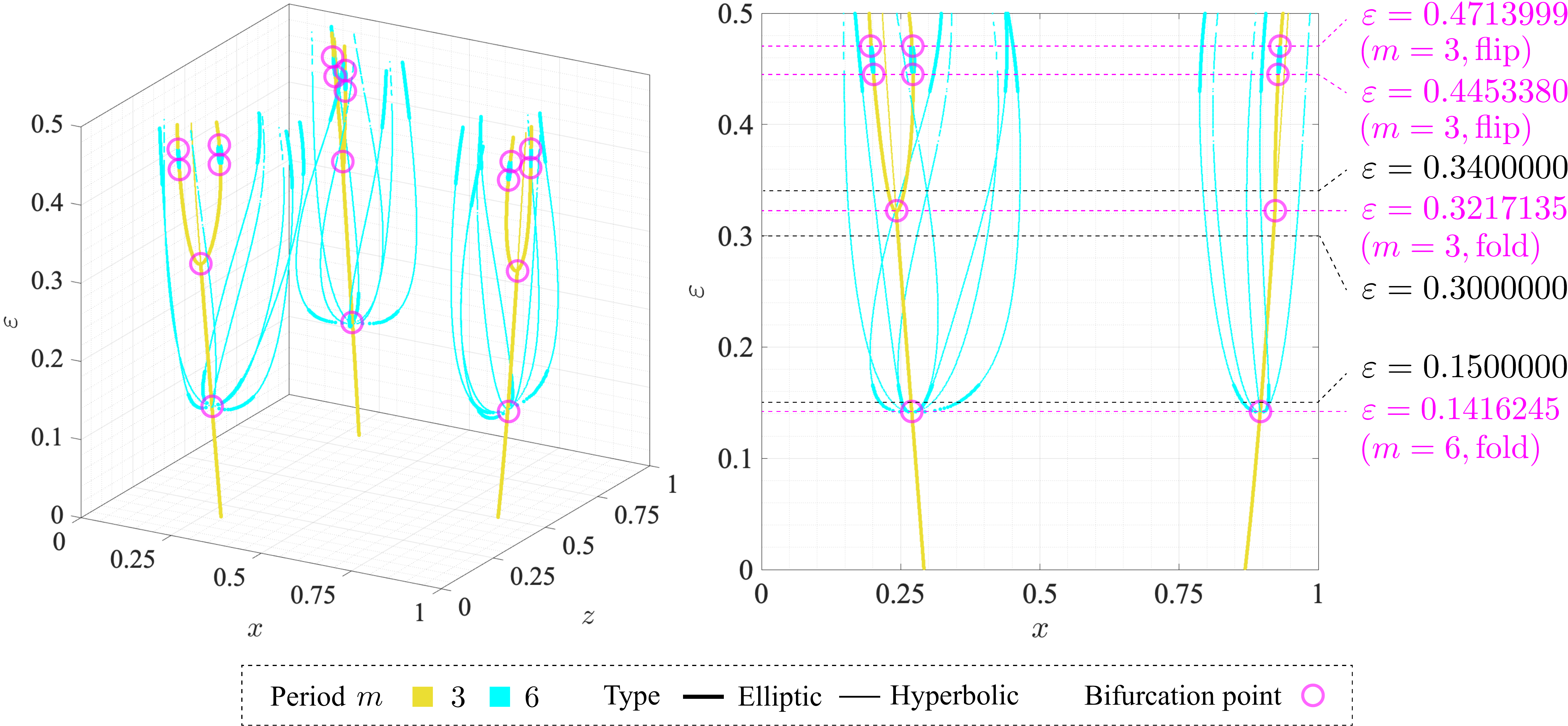}
\caption{$\varepsilon$-bifurcation diagram of 3 and 6-periodic points}
\label{fig:I234_MP3_6}
\end{center}
\end{figure}

\begin{figure}[H]
\begin{center}
\subfigure[Elliptic 3-periodic points]
{\includegraphics[scale=0.25]{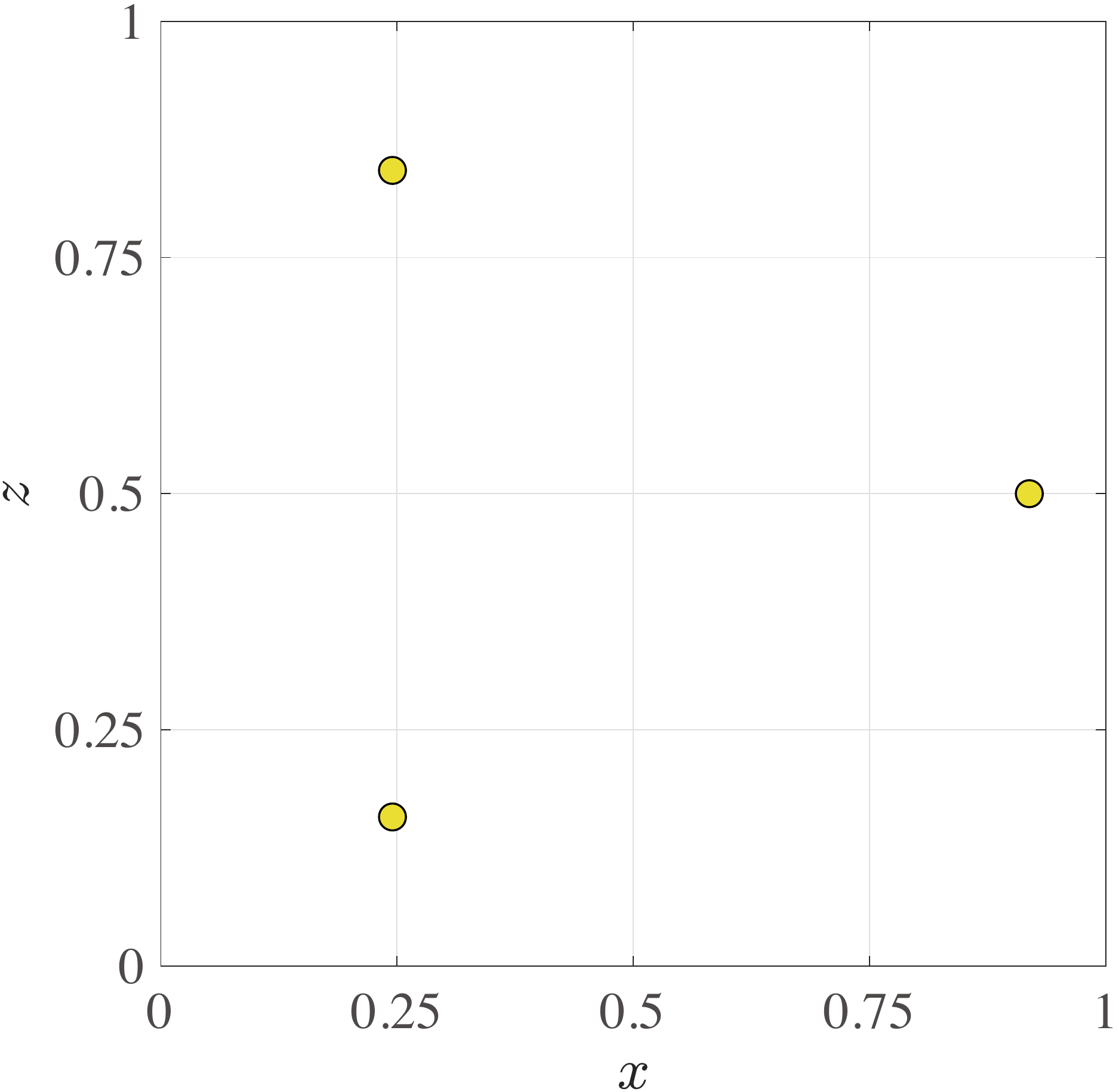}}
\hspace{5mm}
\subfigure[Stable 3-periodic orbit]
{\includegraphics[scale=0.25]{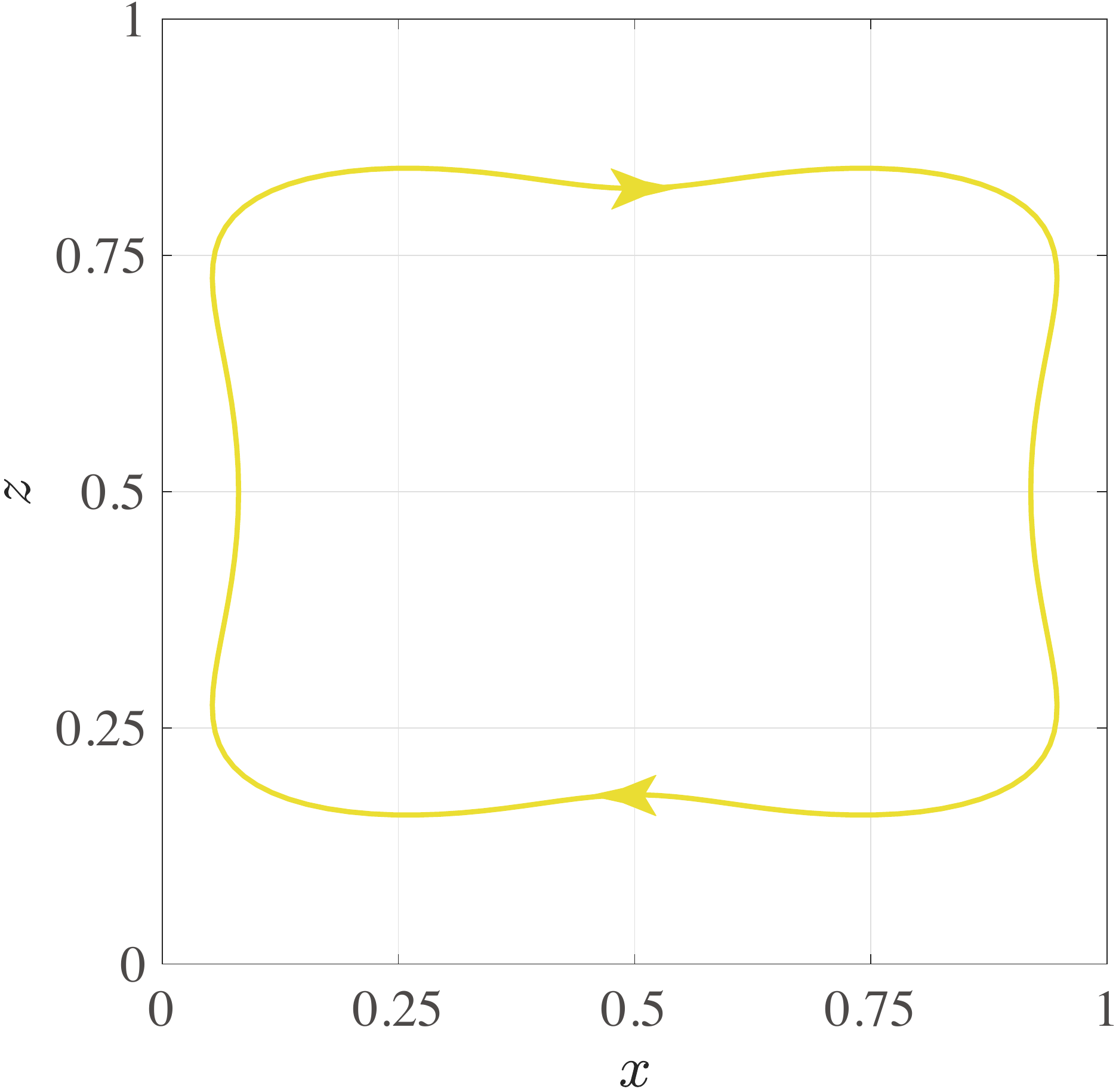}}
\caption{3-periodic points and the projection of their orbits at $\varepsilon=0.3$}
\label{fig:I234_MP3_before}
\end{center}
\end{figure}

\vspace{-6mm}
\begin{figure}[H]
\begin{center}
\subfigure[3-periodic points]
{\includegraphics[scale=0.25]{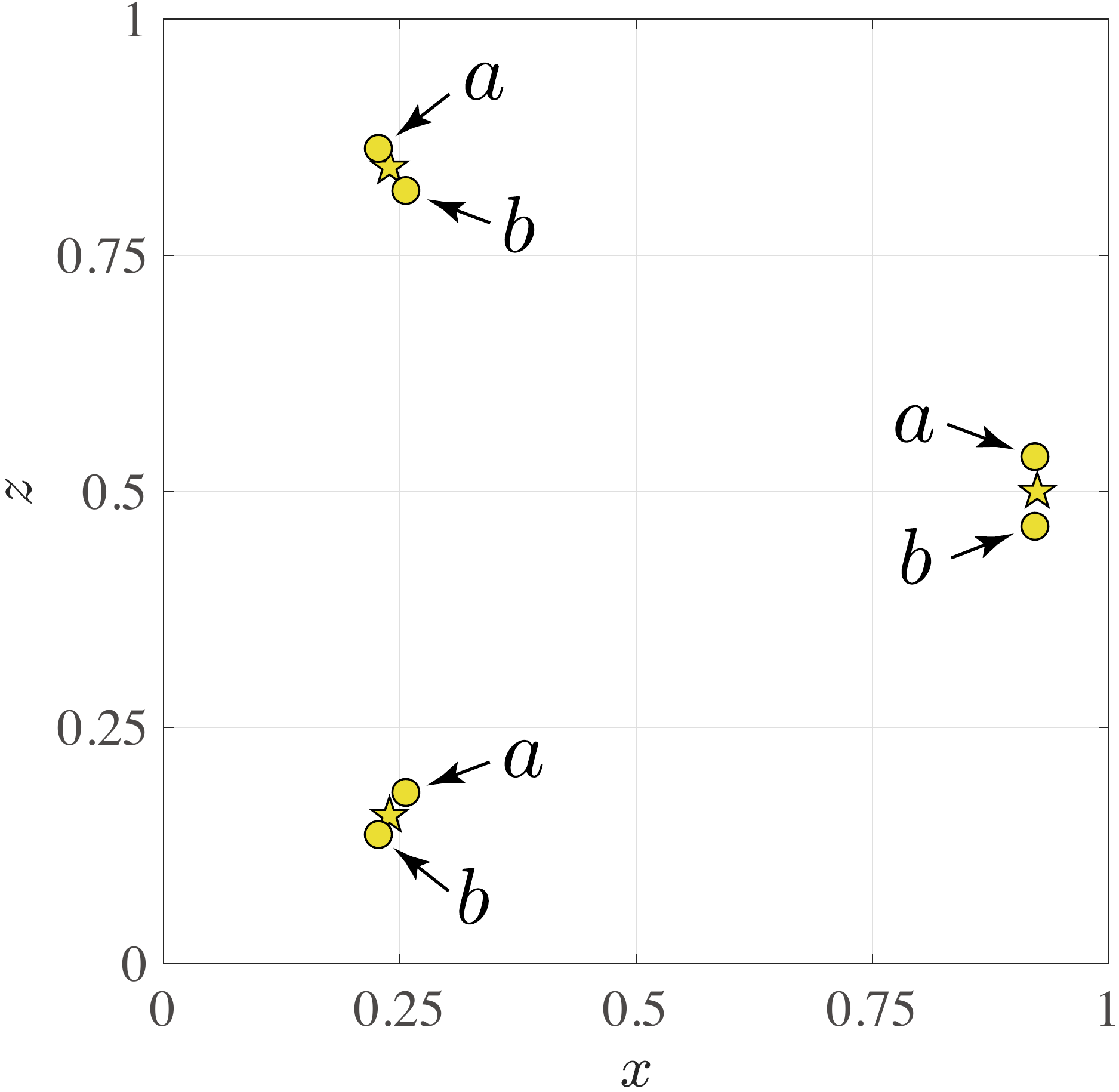}}
\hspace{5mm}
\vspace{-1mm}
\subfigure[Unstable 3-periodic orbit]
{\includegraphics[scale=0.25]{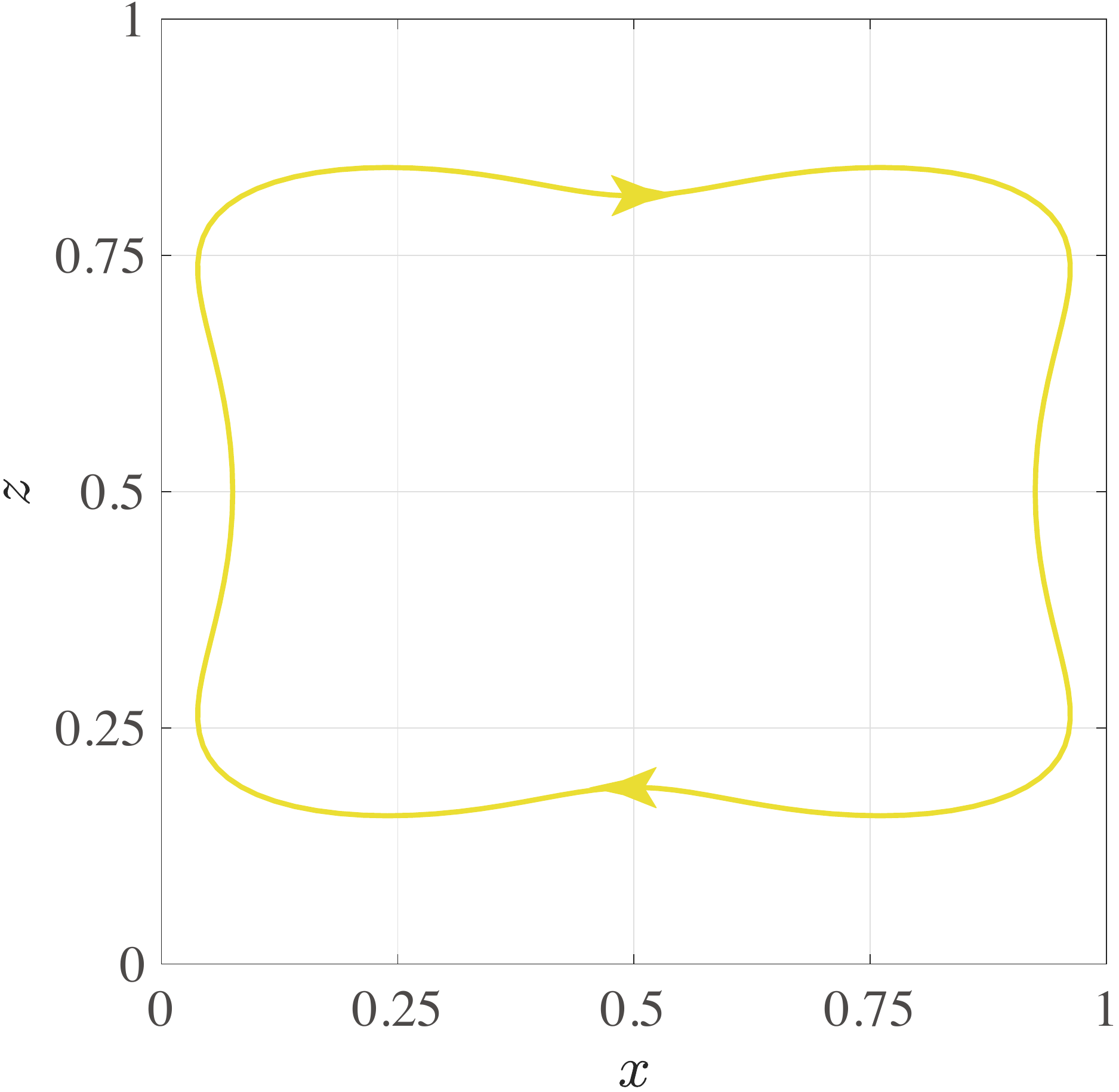}}
\subfigure[Stable 3-periodic orbit $c_a$]
{\includegraphics[scale=0.25]{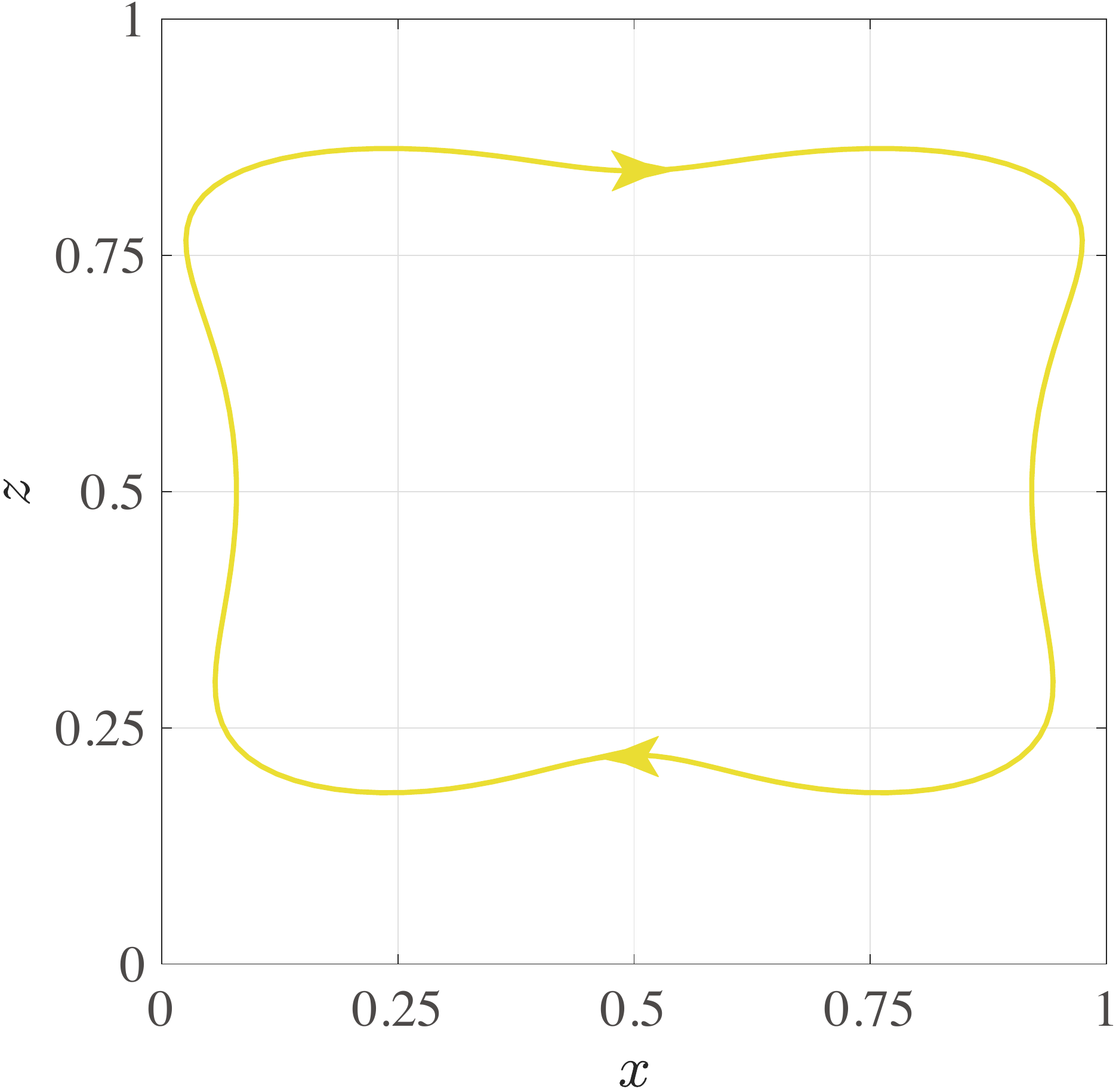}}
\hspace{5mm}
\subfigure[Stable 3-periodic orbit $c_b$]
{\includegraphics[scale=0.25]{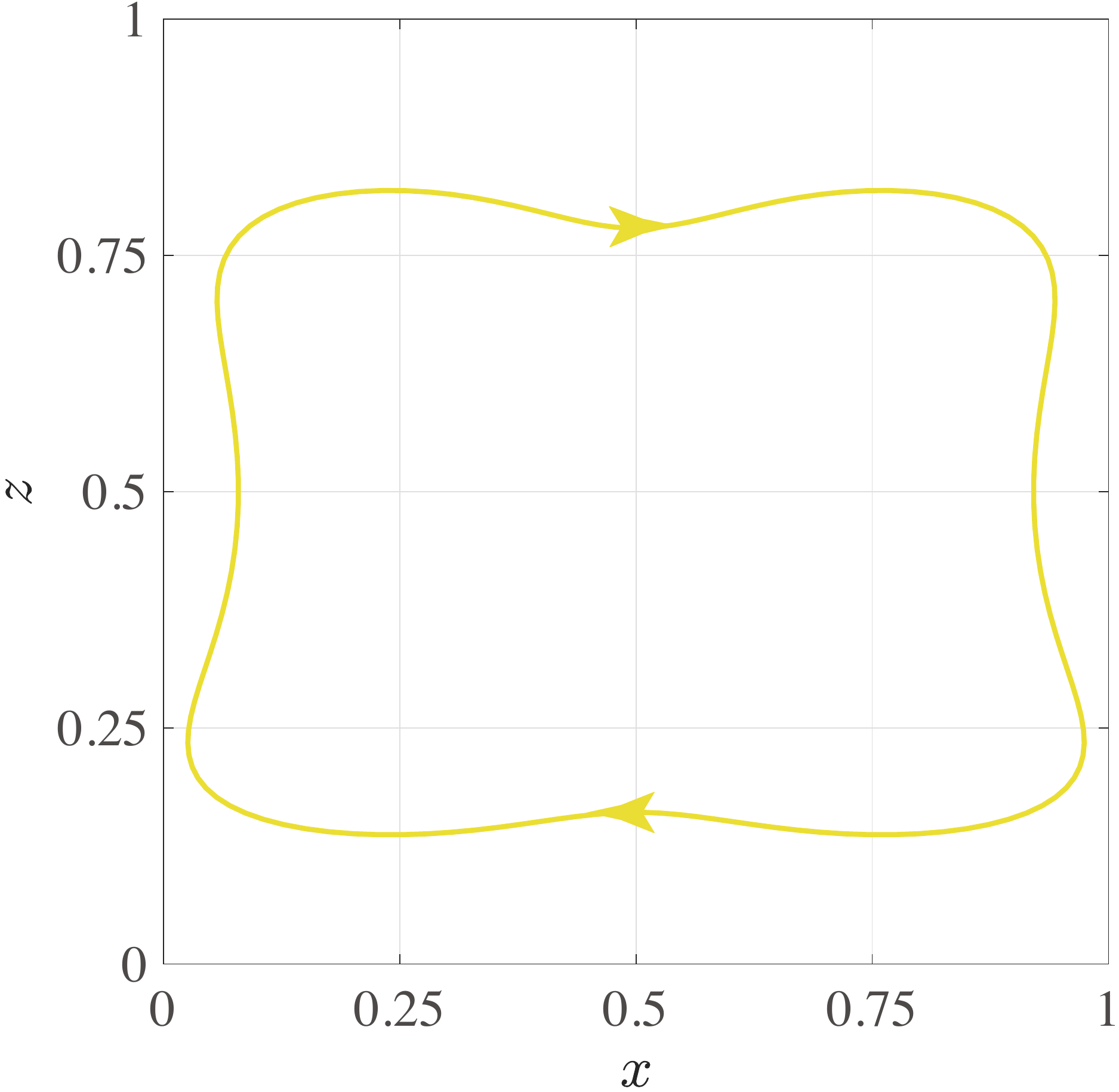}}
\vspace{-1mm}
\caption{3-periodic points and the projection of their orbits at $\varepsilon=0.34$}
\label{fig:I234_MP3_after}
\end{center}
\end{figure}

\vspace{-8mm}


\paragraph{Fold bifurcations of 6-periodic points.}
As can be seen in Fig.\ref{fig:I234_MP3_6}, the blue branches of 6-periodic points grow from the yellow branches of 3-periodic points 
at around $\varepsilon=0.142$. 
It is clarified in our computation that the 6-periodic points bifurcate in a fold bifurcation at $\varepsilon=0.1416245$. 
Note that the 3-periodic points themselves do not seem to bifurcate when the 6-periodic points bifurcate in a fold bifurcation. 
Let us take a look at how the periodic orbits vary with $\varepsilon$ near such fold bifurcation point. 
Fig.\ref{fig:I234_MP3_6_after} shows the 3 and 6-periodic points and the projection of the associated periodic orbits at $\varepsilon=0.15$. 
We can see that elliptic and hyperbolic 6-periodic points appear four each around each elliptic 3-periodic points. 
It follows that stable and unstable 6-periodic orbits appear two each in addition to the stable 3-periodic orbit. 
We name the projection of the four 6-periodic orbits as $c_a, c_b, c_c,$ and $c_d$, and label the associated 6-periodic points 
in Fig.\ref{fig:I234_MP3_6_after_point} as $a, b, c,$ and $d$ respectively in order to show the correspondence, 
where $c_a$ and $c_b$ are symmetric with each other with respect to the horizontal center line of the cell, 
and $c_c$ and $c_d$ are symmetric with respect to the vertical one. 
The projection of the 3-periodic orbit and that of the 6-periodic orbits $c_a$ and $c_c$
are illustrated in Fig.\ref{fig:I234_MP3_6_after_orbitMP3}, Fig.\ref{fig:I234_MP3_6_after_orbitMP6a}, 
and Fig.\ref{fig:I234_MP3_6_after_orbitMP6c}. 
It is observed in Fig.\ref{fig:I234_MP3_6_after} that the resonance condition of the 6-periodic orbits are $|n/m|=1/3$, 
which is the same of that of the 3-periodic orbit. 
This implies that the 6-periodic orbits at the bifurcation point correspond to the 3-periodic orbit. 

\vspace{-1mm}

\paragraph{Fold bifurcations of 9, 12, and 15-periodic points.}
We have seen that the 6-periodic points bifurcate in a fold bifurcation, 
which makes some branches of 6-periodic points grow from those of 3-periodic points. 
Such fold bifurcations are also observed in 9, 12, and 15-periodic points associated with islands $I_2, I_3$, and $I_4$, 
where the $\varepsilon$-bifurcation diagrams of those periodic points are illustrated in Fig.\ref{fig:I234_MP3_9} - Fig.\ref{fig:I234_MP3_15}. 
The 9, 12, and 15-periodic orbits vary similarly to the 6-periodic orbits near these bifurcation points. 
It follows that $3l$-periodic orbits ($l=2,3,4,$ and 5) are generated one after another from the 3-periodic orbits 
by increasing the amplitude $\varepsilon$ of the perturbation. 


\vspace{-1mm}

\begin{figure}[H]
\begin{center}
\subfigure[3 and 6-periodic points]
{\includegraphics[scale=0.25]{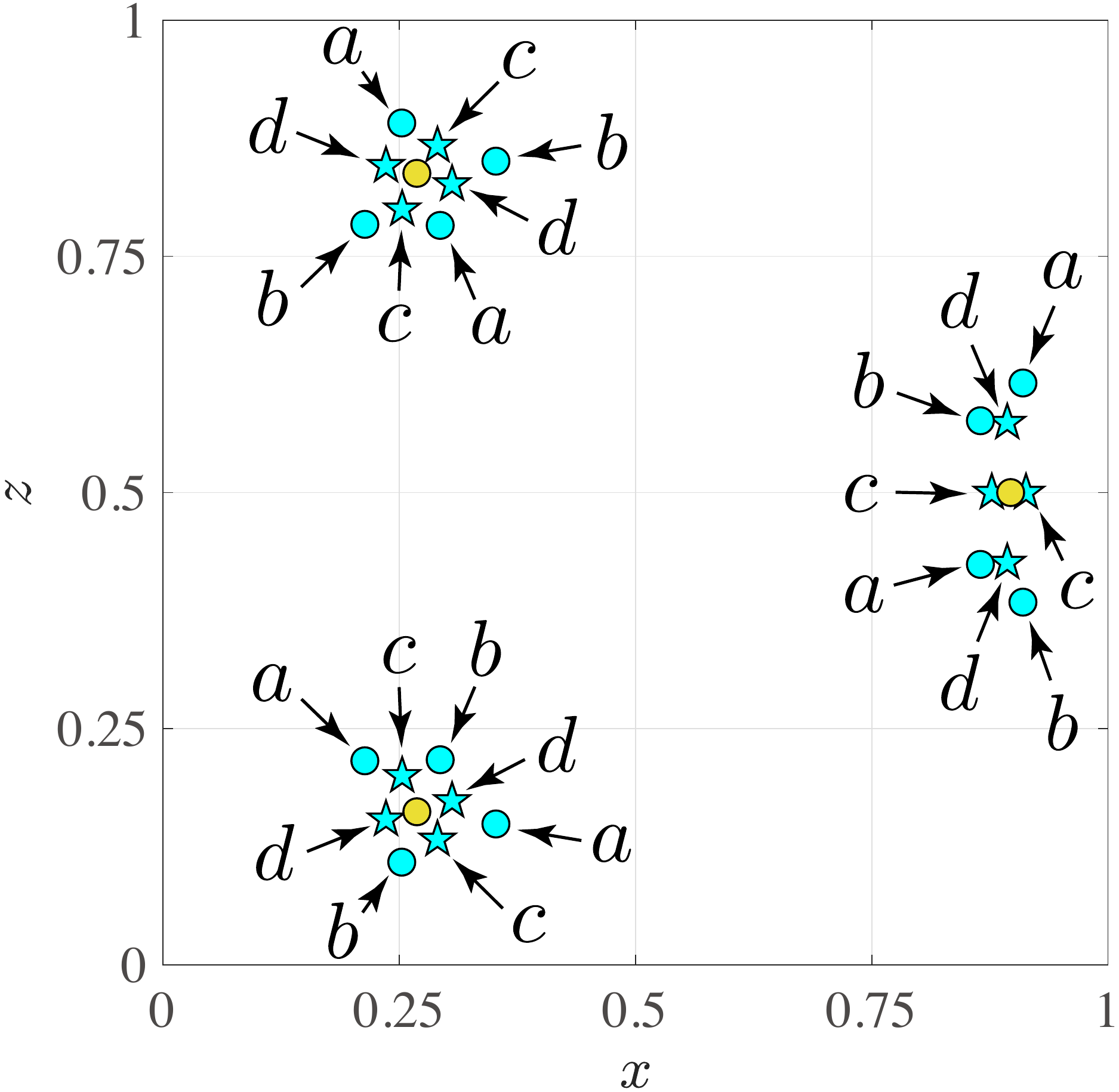}
\label{fig:I234_MP3_6_after_point}}
\hspace{5mm}
\vspace{-1mm}
\subfigure[Stable 3-periodic orbit]
{\includegraphics[scale=0.25]{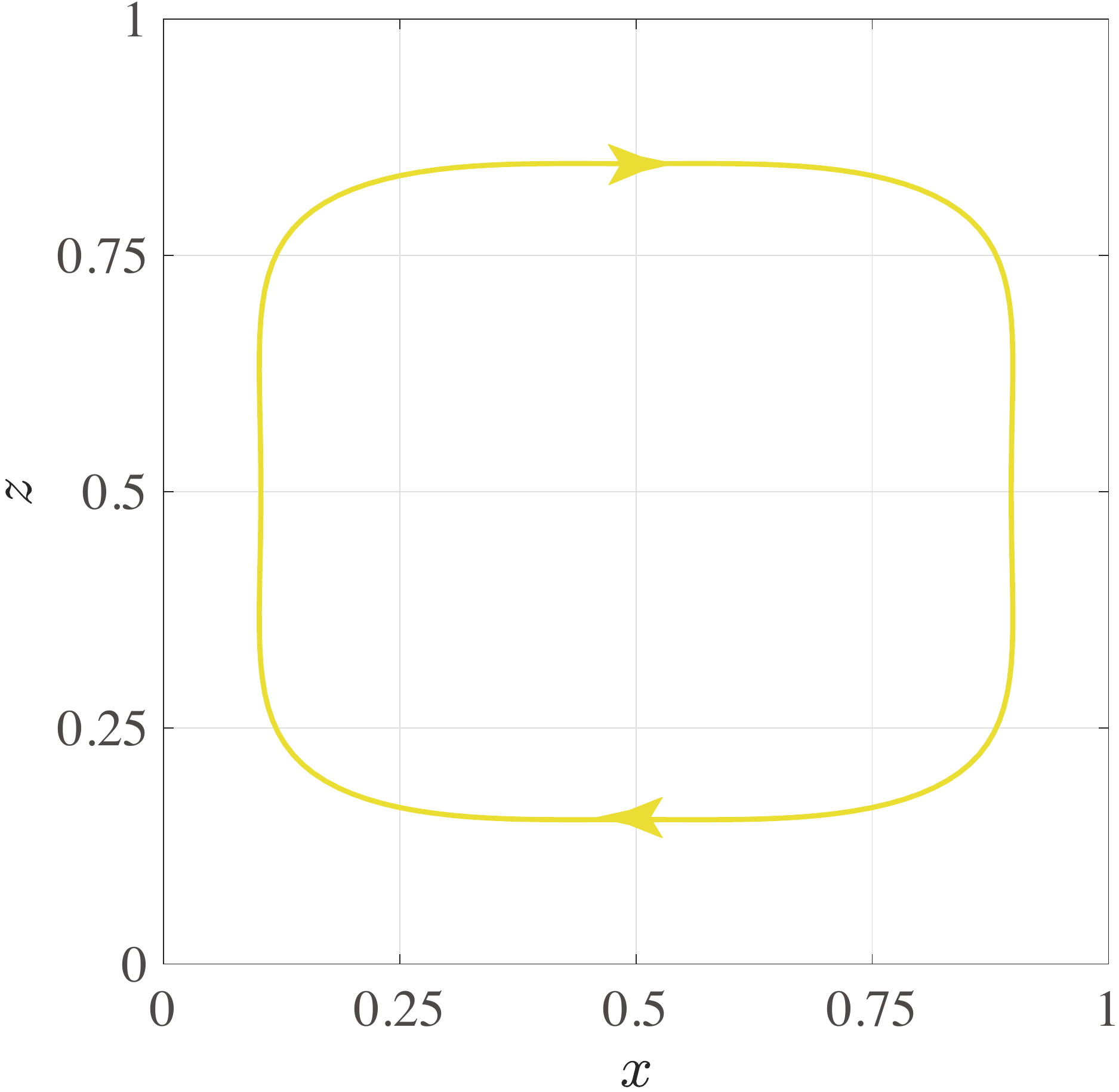}
\label{fig:I234_MP3_6_after_orbitMP3}}
\subfigure[Stable 6-periodic orbit $c_a$]
{\includegraphics[scale=0.25]{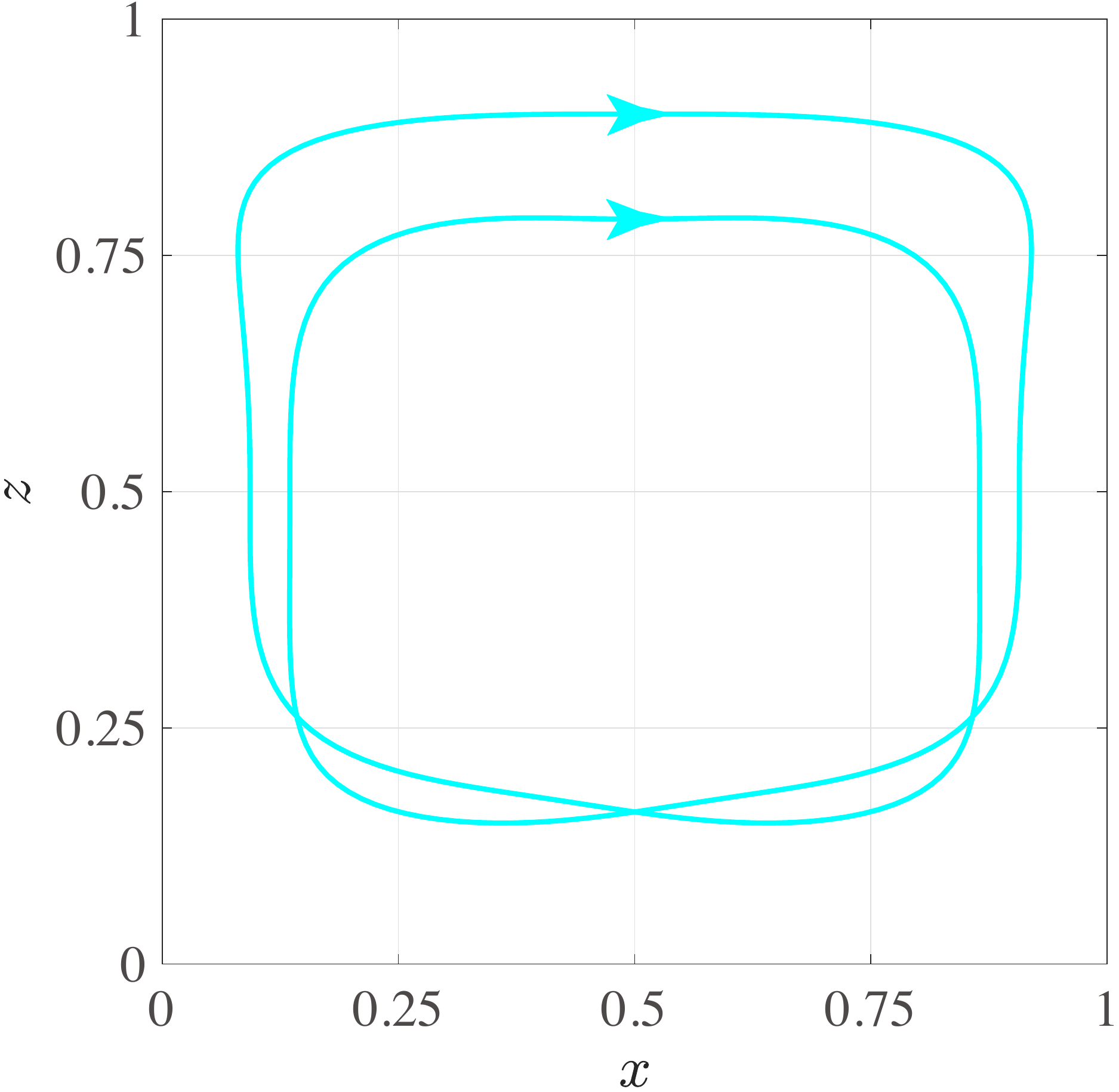}
\label{fig:I234_MP3_6_after_orbitMP6a}}
\hspace{5mm}
\subfigure[Unstable 6-periodic orbit $c_c$]
{\includegraphics[scale=0.25]{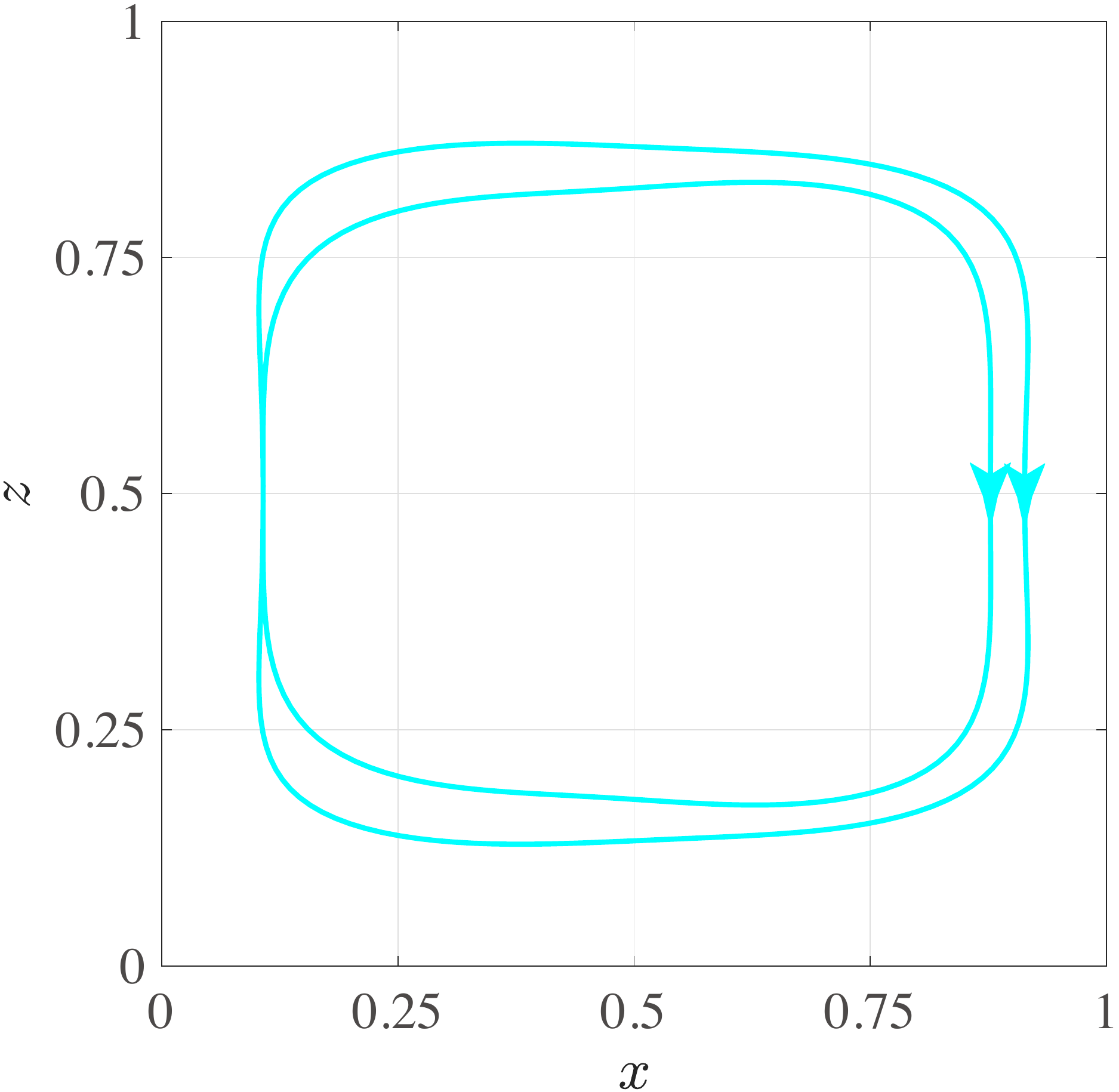}
\label{fig:I234_MP3_6_after_orbitMP6c}}
\vspace{-1mm}
\caption{3 and 6-periodic points and the projection of their orbits at $\varepsilon=0.15$}
\label{fig:I234_MP3_6_after}
\end{center}
\end{figure}

\vspace{-5mm}

\begin{figure}[H]
\begin{center}
\includegraphics[scale=0.25]{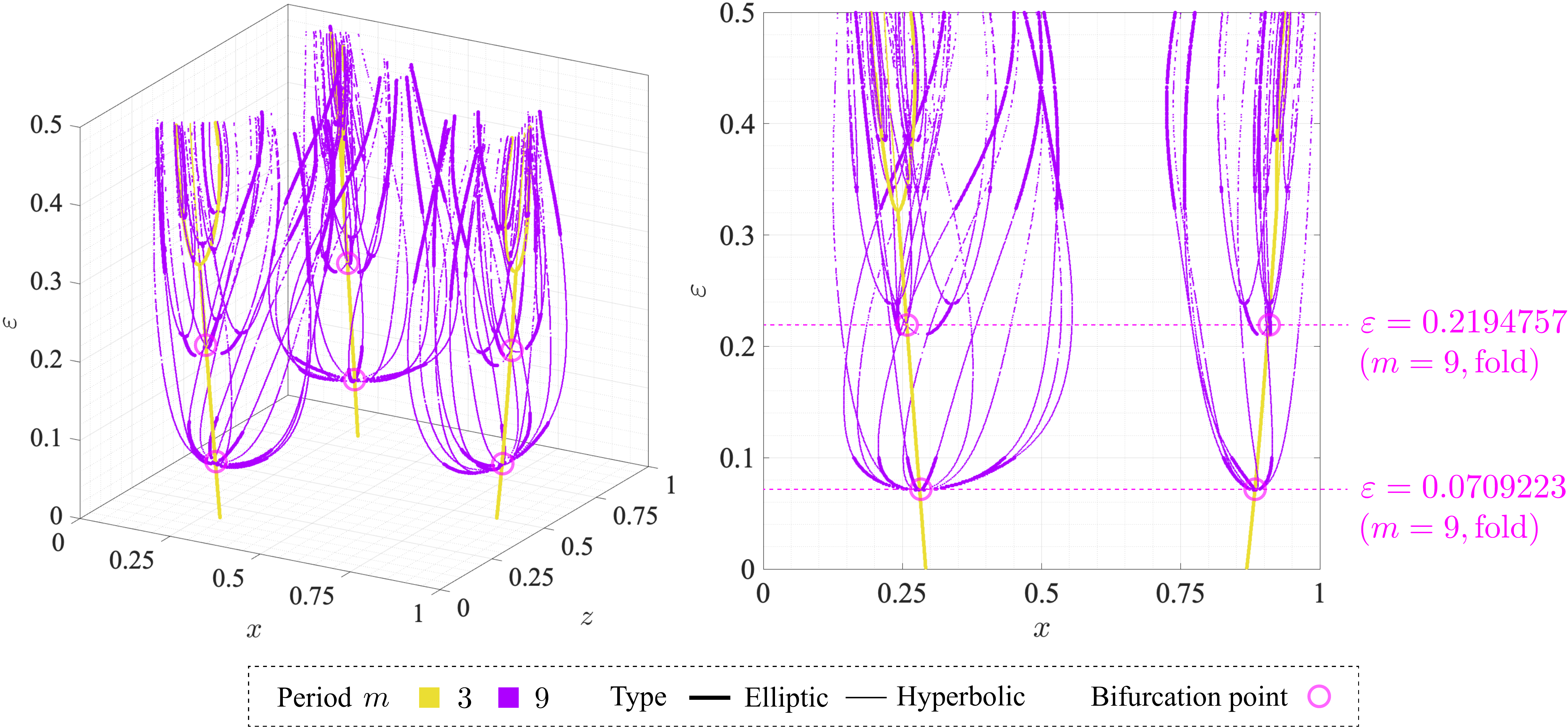}
\caption{$\varepsilon$-bifurcation diagram of 3 and 9-periodic points}
\label{fig:I234_MP3_9}
\end{center}
\end{figure}

\begin{figure}[H]
\begin{center}
\includegraphics[scale=0.25]{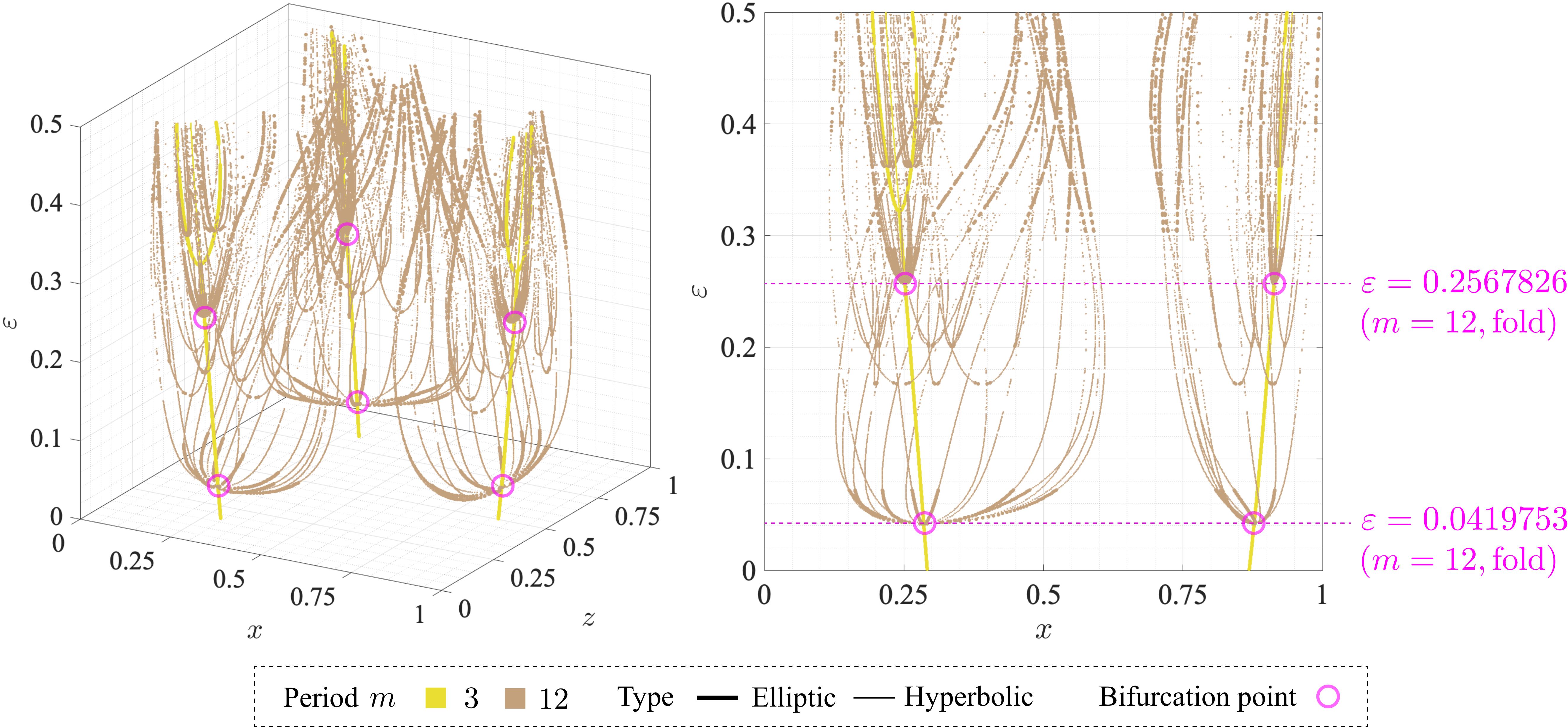}
\vspace{-2mm}
\caption{$\varepsilon$-bifurcation diagram of 3 and 12-periodic points}
\label{fig:I234_MP3_12}
\end{center}
\end{figure}

\vspace{-5mm}

\begin{figure}[H]
\begin{center}
\includegraphics[scale=0.25]{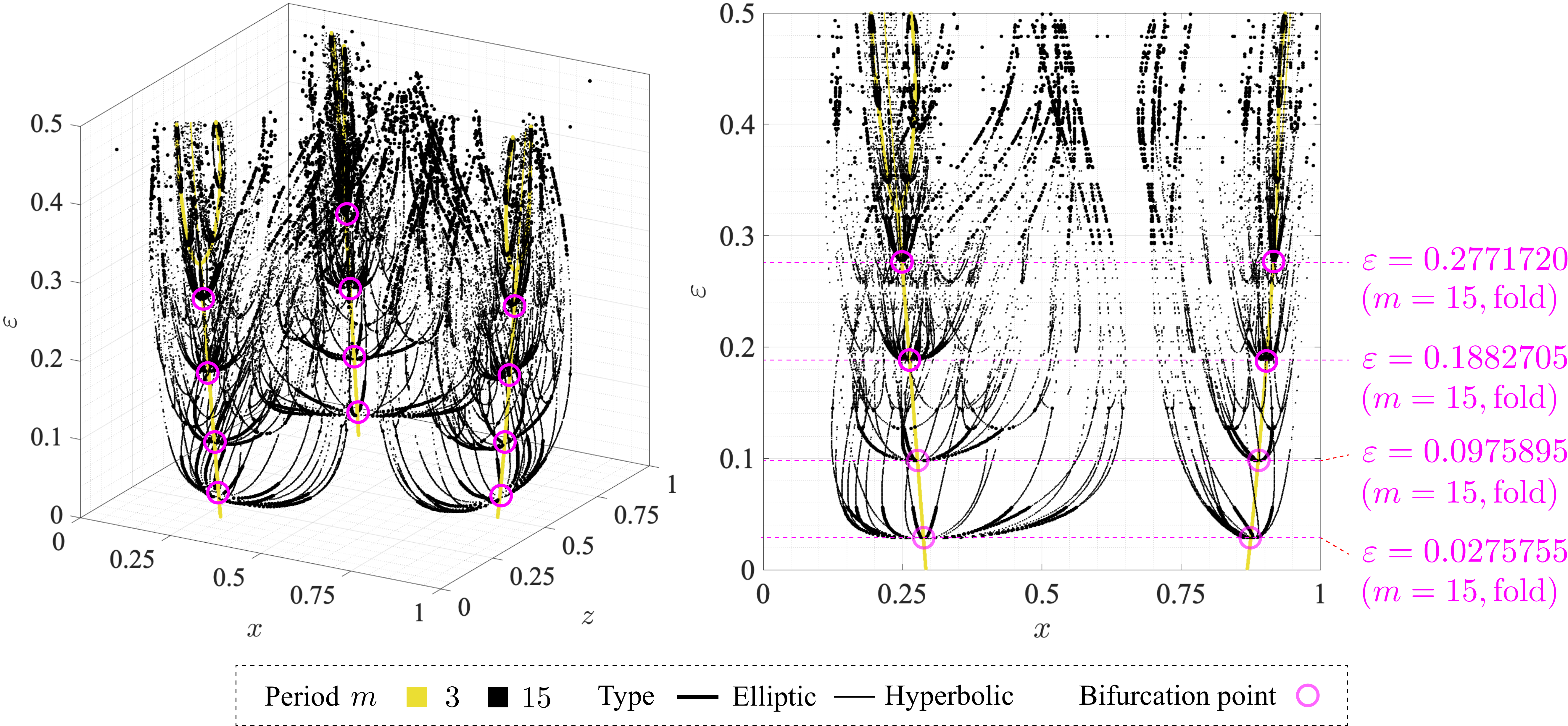}
\vspace{-2mm}
\caption{$\varepsilon$-bifurcation diagram of 3 and 15-periodic points}
\label{fig:I234_MP3_15}
\end{center}
\end{figure}

\vspace{-7mm}

\subsection{Bifurcations associated with other KAM islands}\label{Sec:others}
So far we have explored the bifurcations associated with KAM islands $I_1, I_2, I_3,$ and $I_4$. 
In this subsection, we investigate those which seem to be associated with other islands.
Here, we especially focus on the bifurcations of 5-periodic points with resonance condition $|n/m|=1/5$ 
and those of 4 and 8-periodic points with $|n/m|=1/4$. 

\paragraph{Fold bifurcations of 5-periodic points.}
Fig.\ref{fig:others_MP5} indicates the $\varepsilon$-bifurcation diagram of 5-periodic points of which 
resonance condition is $|n/m|=1/5$. It follows that the 5-periodic points bifurcate
at $\varepsilon=0.0469599$ and $\varepsilon=0.0832330$ in a fold bifurcation. 
Now let us take a look at how the symmetry of the stable 5-periodic orbits vary by the two bifurcations. 
Fig.\ref{fig:others_MP5_before_point} shows the 5-periodic points at $\varepsilon=0.04$, where it follows that 
the elliptic and hyperbolic 5-periodic points exist five each on the Poincar\'e section $\Sigma^{\theta_0}$. 
As is illustrated in Fig.\ref{fig:others_MP5_before_orbit}, the projection of the orbit of the elliptic 5-periodic points, namely, 
the stable 5-periodic orbit, is symmetric with respect to the horizontal and vertical center lines of the cell. 
However, when we increase $\varepsilon$, each elliptic 5-periodic point vary to a hyperbolic 

\begin{figure}[H]
\begin{center}
\includegraphics[scale=0.25]{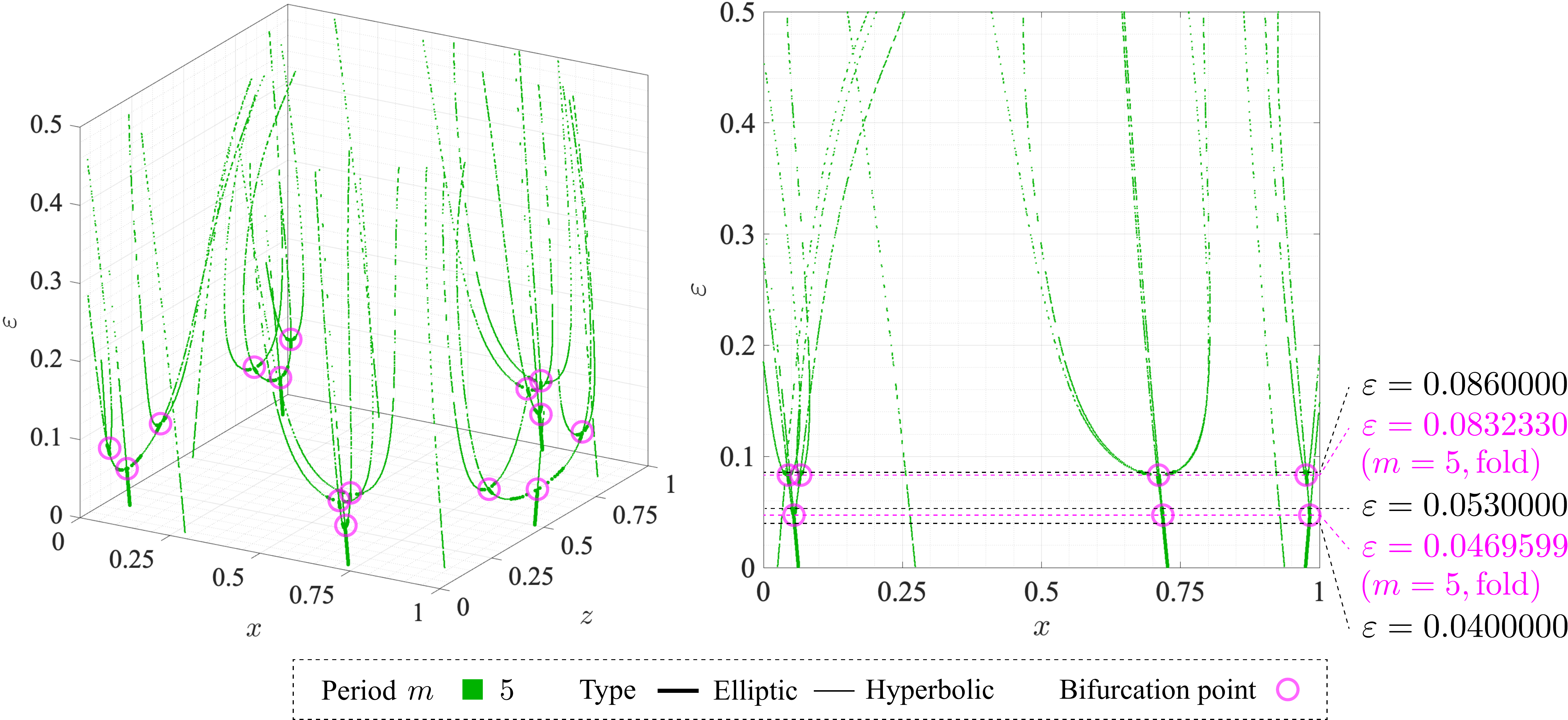}
\caption{$\varepsilon$-bifurcation diagram of 5-periodic points with $|n/m|=1/5$}
\label{fig:others_MP5}
\end{center}
\end{figure}

\vspace{-5mm}

\begin{figure}[H]
\begin{center}
\subfigure[5-periodic points]
{\includegraphics[scale=0.25]{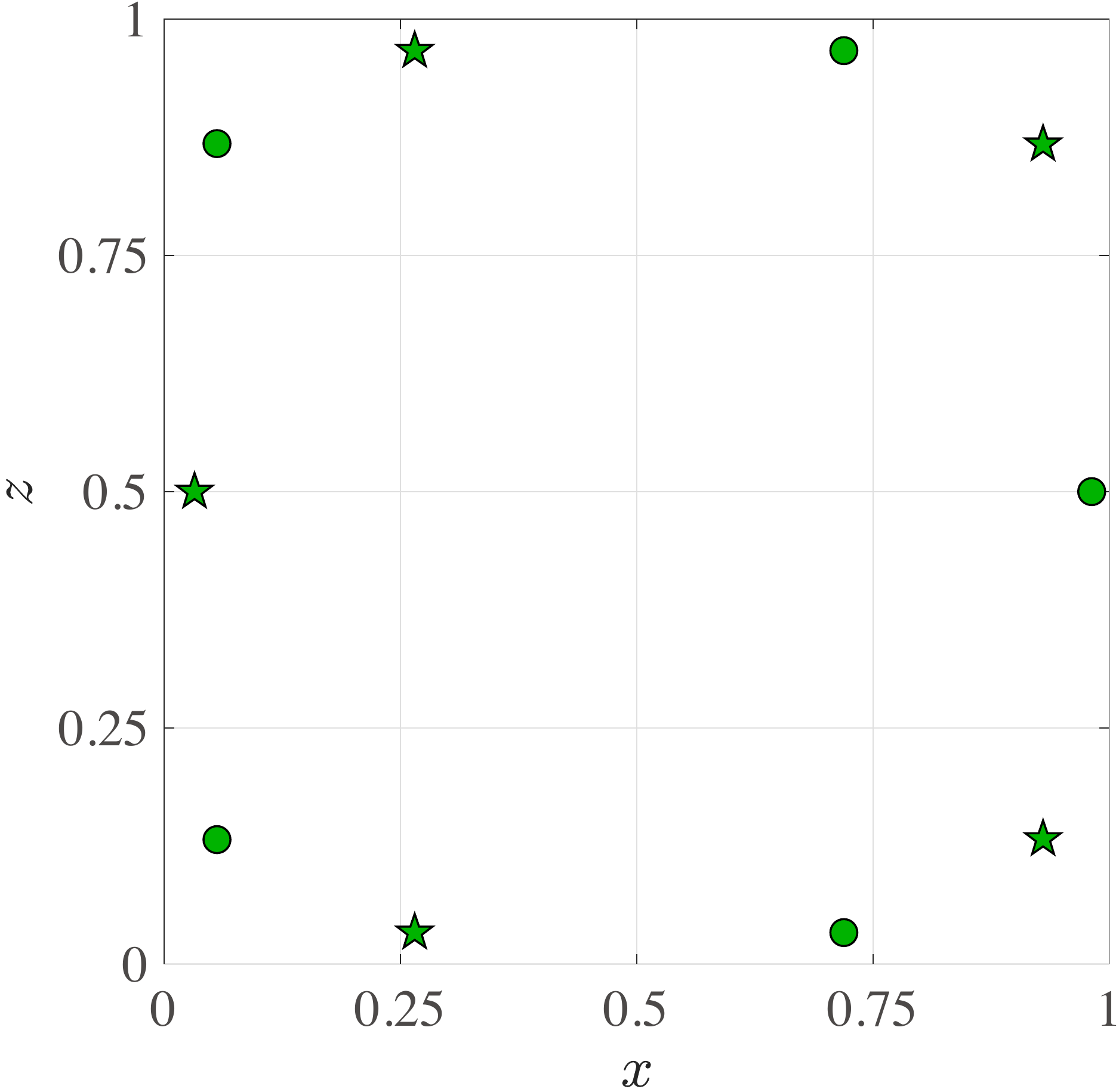}
\label{fig:others_MP5_before_point}}
\hspace{5mm}
\subfigure[Stable 5-periodic orbit]
{\includegraphics[scale=0.25]{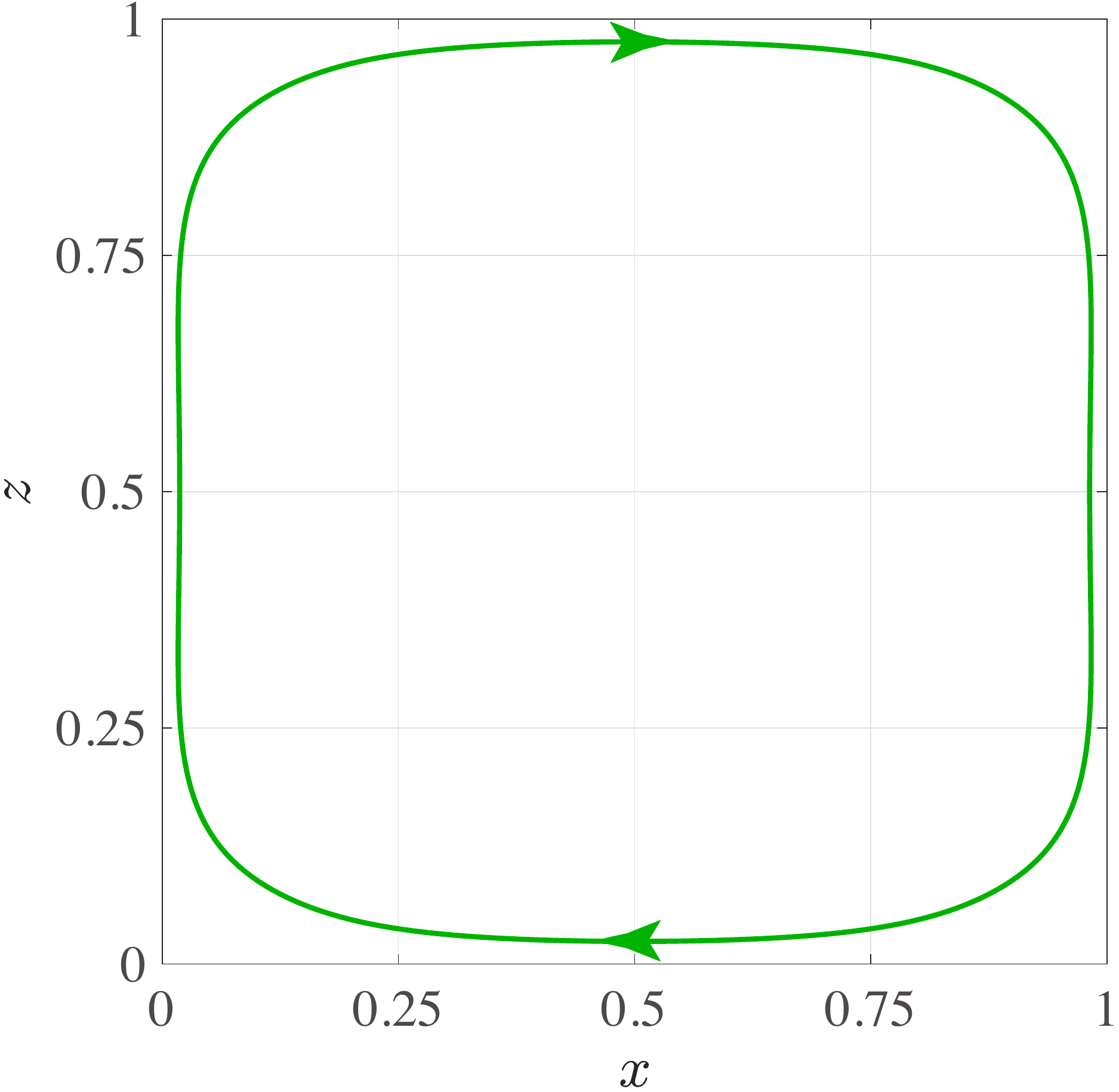}
\label{fig:others_MP5_before_orbit}}
\caption{5-periodic points and the projection of their orbits at $\varepsilon=0.04$}
\label{fig:others_MP5_before}
\end{center}
\end{figure}

\vspace{-5mm}

\begin{figure}[H]
\begin{center}
\subfigure[5-periodic points]
{\includegraphics[scale=0.25]{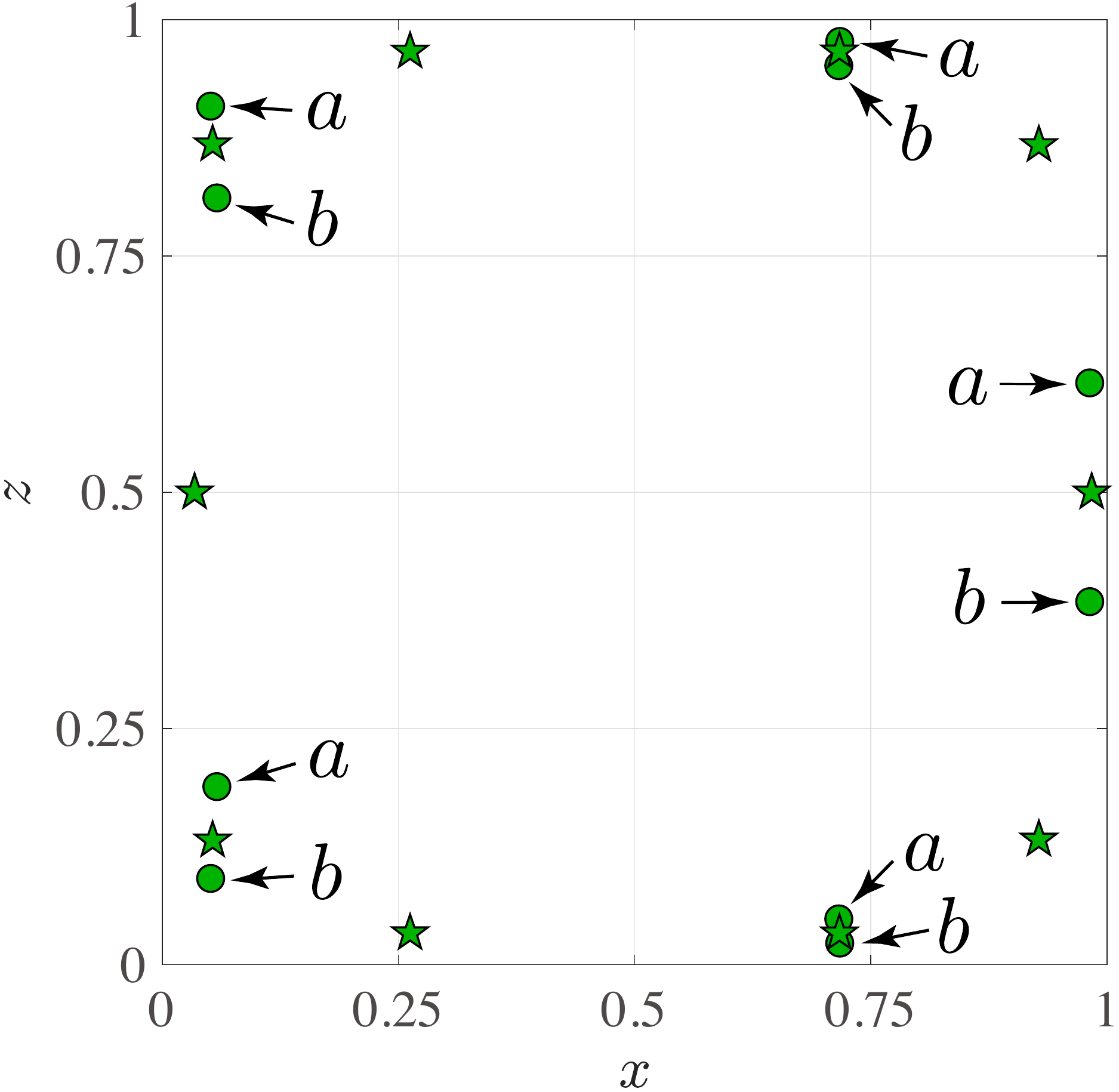}
\label{fig:others_MP5_after1_point}}
\hspace{5mm}
\subfigure[Stable 5-periodic orbit $c_a$]
{\includegraphics[scale=0.25]{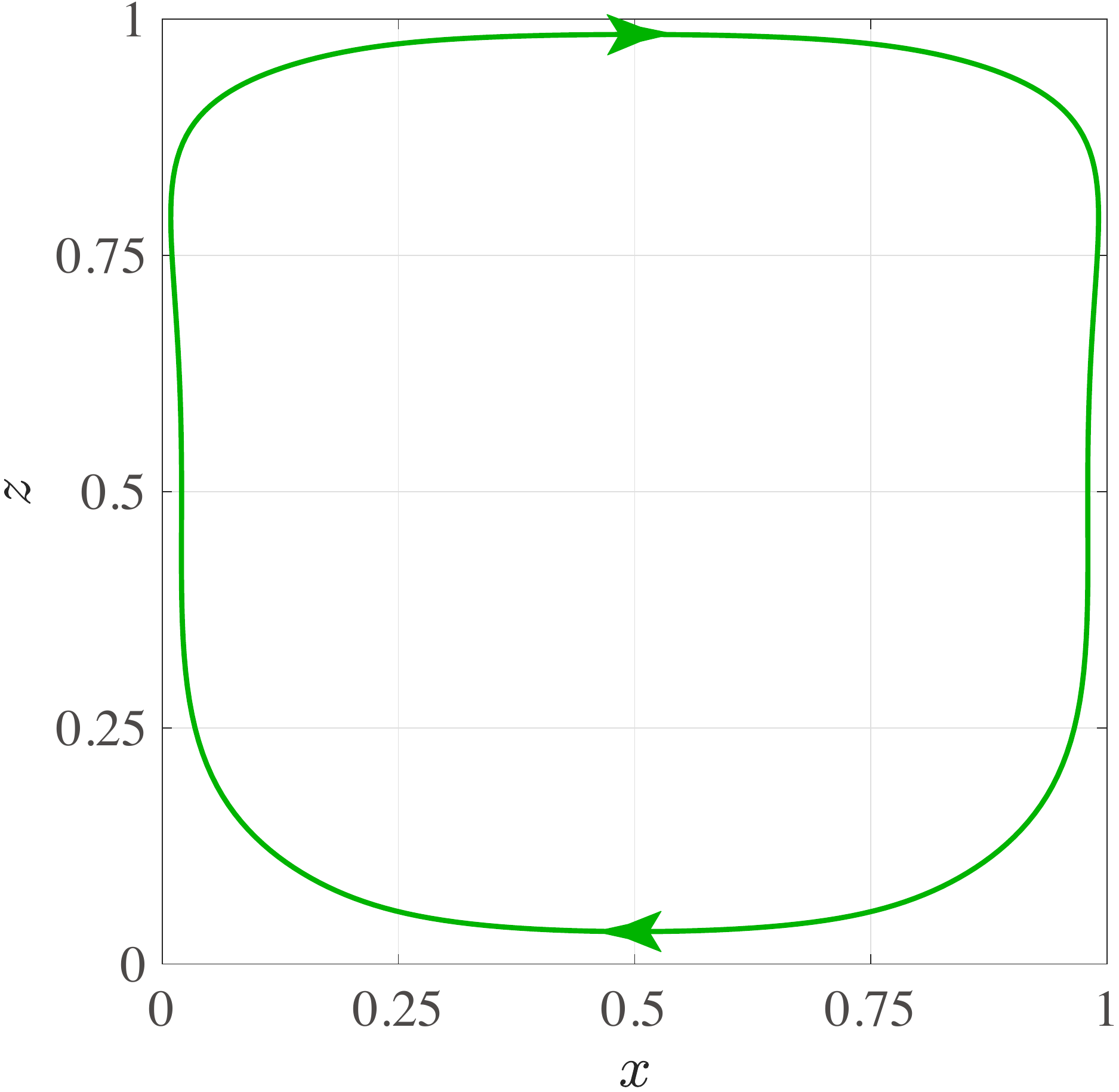}
\label{fig:others_MP5_after1_orbit}}
\caption{5-periodic points and the projection of their orbits at $\varepsilon=0.053$}
\label{fig:others_MP5_after1}
\end{center}
\end{figure}

\begin{figure}[H]
\begin{center}
\subfigure[5-periodic points]
{\includegraphics[scale=0.25]{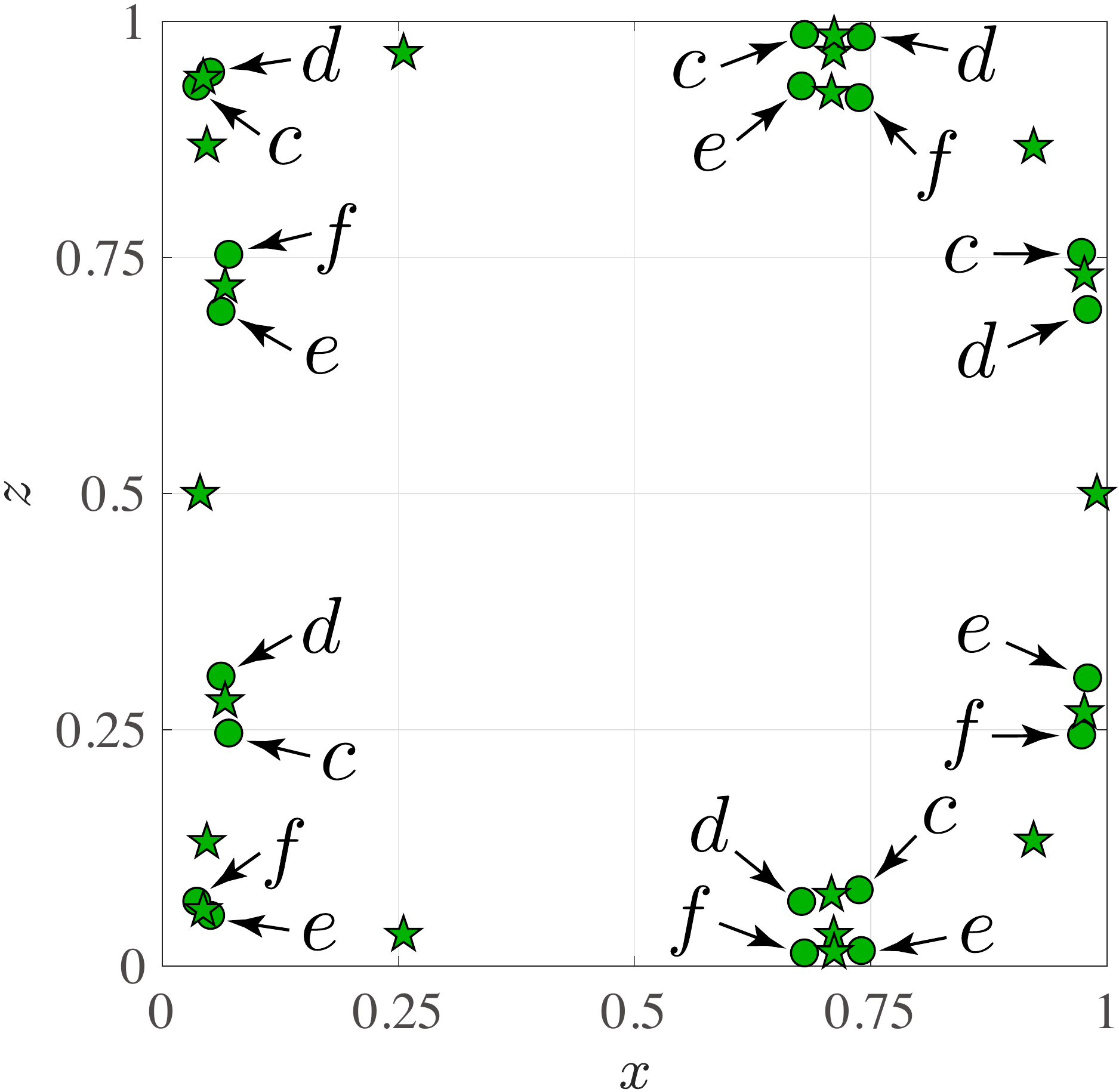}
\label{fig:others_MP5_after2_point}}
\hspace{5mm}
\subfigure[Stable 5-periodic orbit $c_c$]
{\includegraphics[scale=0.25]{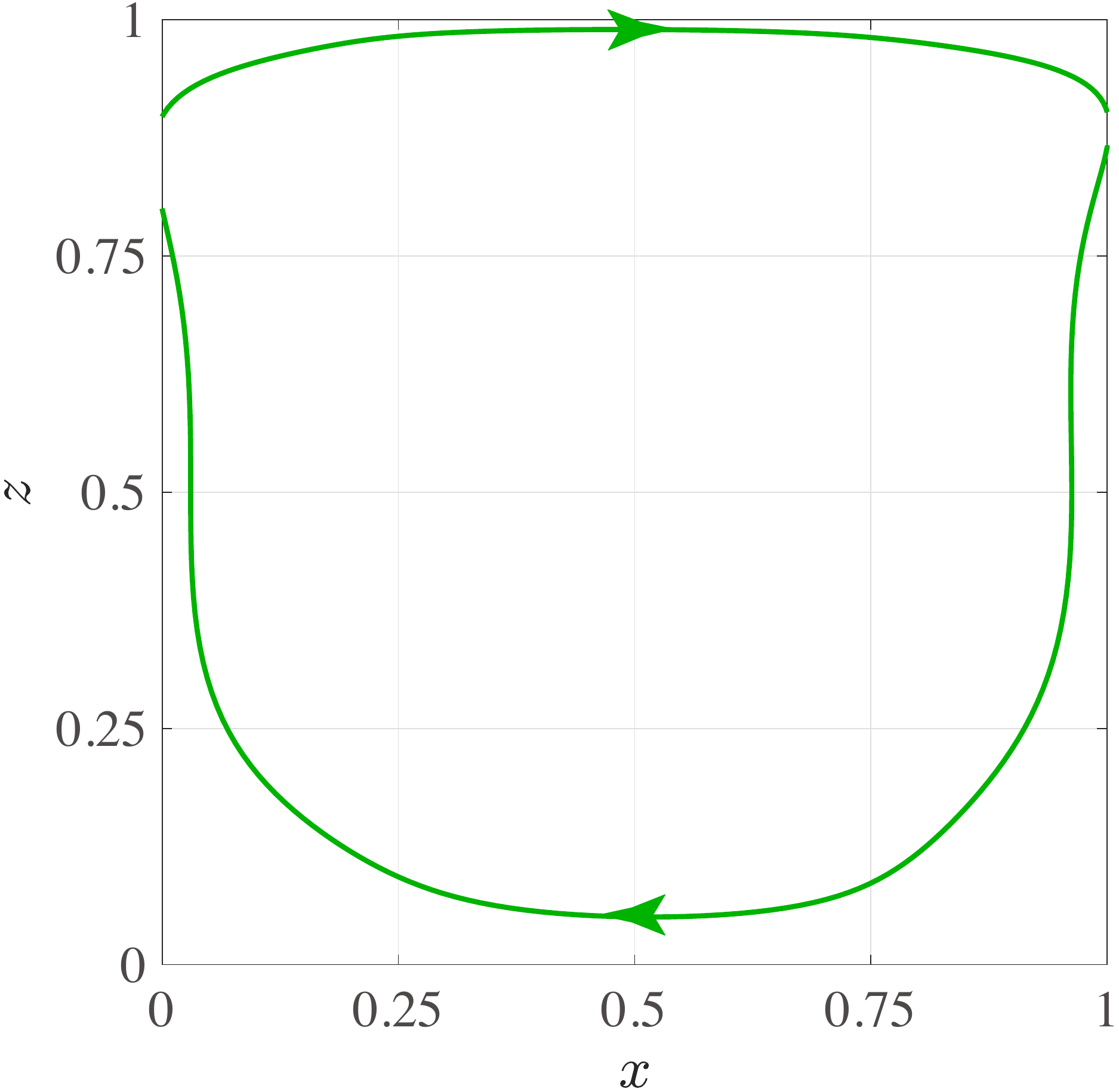}
\label{fig:others_MP5_after2_orbit}}
\caption{5-periodic points and the projection of their orbits at $\varepsilon=0.086$}
\label{fig:others_MP5_after2}
\end{center}
\end{figure}

\noindent
5-periodic point and two new elliptic 5-periodic points appear in the neighborhood 
by the fold bifurcation at $\varepsilon=0.0469599$; the case $\varepsilon=0.053$ is depicted in Fig.\ref{fig:others_MP5_after1_point}. 
We name the projection of the two new stable 5-periodic orbits as $c_a$ and $c_b$, 
and label the associated elliptic 5-periodic points as $a$ and $b$ respectively in order to show the correspondence, 
where $c_a$ and $c_b$ is symmetric with each other with respect to the horizontal center line.
From $c_a$ in Fig.\ref{fig:others_MP5_after1_orbit}, it follows that the orbit looses one of its symmetric property  
and become only symmetric with the vertical center line. 

When the amplitude $\varepsilon$ is increased further, the ten elliptic 5-periodic points once vary to hyperbolic ones, 
but return to elliptic ones. After that, they bifurcate again in a fold bifurcation at $\varepsilon=0.0832330$, 
which denotes that twenty elliptic 5-periodic points appear by the bifurcation. 
Fig.\ref{fig:others_MP5_after2_point} illustrates the 5-periodic points at $\varepsilon=0.086$. 
We denote the projection of the four new stable 5-periodic orbits by $c_c, c_d, c_e$, and $c_f$,  
and label the associated elliptic 5-periodic points as $c, d, e,$ and $f$ respectively in order to show the correspondence, 
where $c_c, c_d, c_e$, and $c_f$ are symmetric with each other with respect to the horizontal or vertical center line.
From $c_c$ in Fig.\ref{fig:others_MP5_after2_orbit}, it follows that the orbit looses its symmetry 
and become asymmetric with respect to the horizontal and vertical center lines of the cell. 
Hence, the stable 5-periodic orbit of which projection is originally symmetric with respect to the horizontal and vertical 
center lines of the cell become asymmetric by the two fold bifurcations. 
Furthermore, we can see that the number of 5-periodic orbits increases by the bifurcations 
and also that they become unstable when $\varepsilon$ is large enough, 
which denotes that the fluid transport become more complex. 




\paragraph{Flip bifurcations of 4-periodic points.}
Next, let us take a look at the bifurcations of 4-periodic points. 
Fig.\ref{fig:others_MP4_8}  illustrates the $\varepsilon$-bifurcation diagram of 4 and 8-periodic points 
of which resonance condition is $|n/m|=1/4$. When the amplitude of the perturbation is $\varepsilon=0.05$, 
elliptic and hyperbolic 4-periodic points appear eight each, as is shown in Fig.\ref{fig:others_MP4_8_before_point}. 
It follows that stable and unstable 4-periodic orbits appear two each. 
We name the projection of the two stable 4-periodic orbits as $c_a$ and $c_b$, and label the associated elliptic 4-periodic 
points as $a$ and $b$ respectively in order to show the correspondence. 
As is shown in Fig.\ref{fig:others_MP4_8_before_orbit}, $c_a$ is only symmetric with respect to the horizontal 
center line of the cell, which denotes that $c_b$ is symmetric with $c_a$ with respect to the vertical one. 

Now, we consider varying the amplitude $\varepsilon$. When the amplitude is increased from 
$\varepsilon=0$ to $\varepsilon=0.5$, the hyperbolic 4-periodic points do not seem to bifurcate. 
In contrast, it is clarified in our computation that the elliptic 4-periodic points bifurcate in a flip bifurcation 
at $\varepsilon=0.0545763$. At the bifurcation point, each elliptic 4-periodic point varies to a hyperbolic one 
and two new 8-periodic points appear in the neighborhood of each 4-periodic point, 
where Fig.\ref{fig:others_MP4_8_after_point} depicts them at $\varepsilon=0.059$. 
It follows that each stable 4-periodic orbit varies to one unstable 4-periodic orbit and one stable 8-periodic orbit 
by the flip bifurcation. Since there are two stable 4-periodic orbits before the bifurcation, 
two new unstable 4-periodic orbits and two new stable 8-periodic orbits are generated by the bifurcation. 
We denote the projection of the former two orbits by $c_c$ and $c_d$ and that of the latter two by $c_e$ and $c_f$.
Then, we label the associated periodic points as $c, d, e,$ and $f$ respectively in order to show the correspondence. 
Note that $c_c$ and $c_d$ as well as $c_e$ and $c_f$ are symmetric with each other 
with respect to the vertical center line. From $c_c$ and $c_e$ in Fig.\ref{fig:others_MP4_8_after_orbit1} and  
Fig.\ref{fig:others_MP4_8_after_orbit2}, it follows that the symmetric axes of all the orbits from $c_a$ to $c_f$ are the same, 
which is the horizontal center line. 
In addition, the resonance conditions of the 4 and 8-periodic orbits are both $|n/m|=1/4$. 
It follows that the symmetric axis of the projection and the resonance condition of the periodic orbits do not vary by the bifurcation. 
Furthermore, it is observed that most of the 4 and 8-periodic orbits become unstable when $\varepsilon$ is large enough, 
which denotes that the fluid transport become more complex.

\begin{figure}[H]
\begin{center}
\includegraphics[scale=0.25]{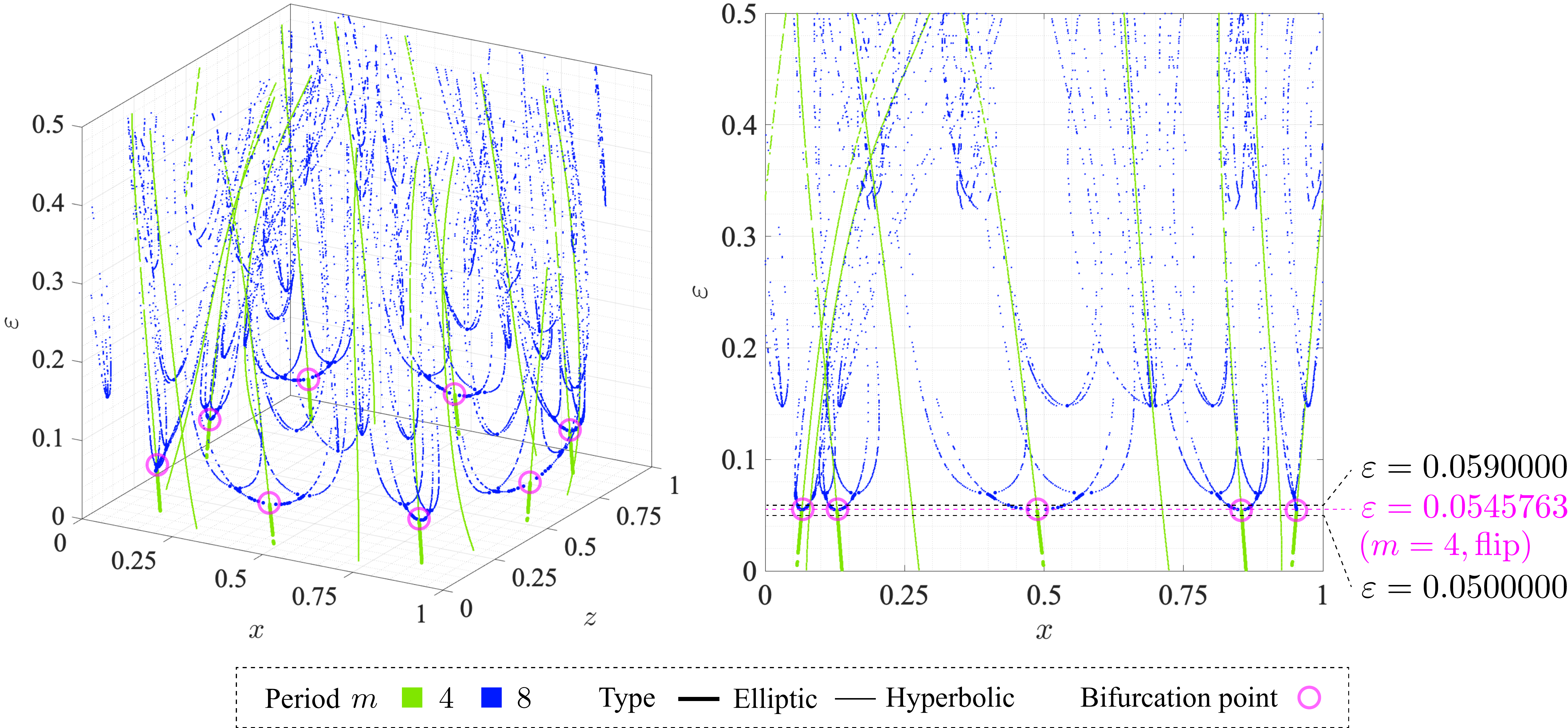}
\caption{$\varepsilon$-bifurcation diagram of 4 and 8-periodic points with $|n/m|=1/4$}
\label{fig:others_MP4_8}
\end{center}
\end{figure}

\vspace{-5mm}

\begin{figure}[H]
\begin{center}
\subfigure[4-periodic points]
{\includegraphics[scale=0.25]{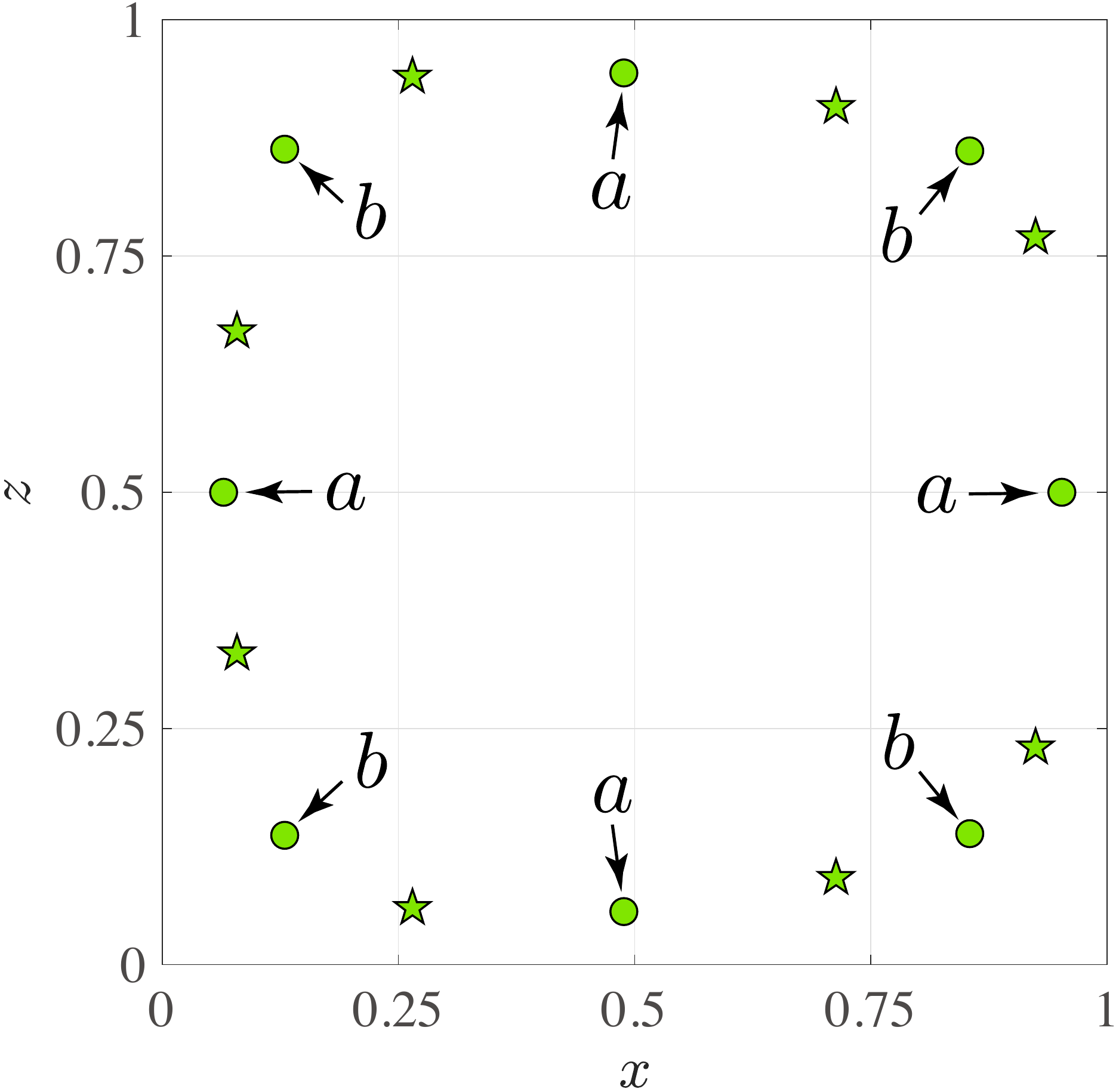}
\label{fig:others_MP4_8_before_point}}
\hspace{5mm}
\subfigure[Stable 4-periodic orbit $c_a$]
{\includegraphics[scale=0.25]{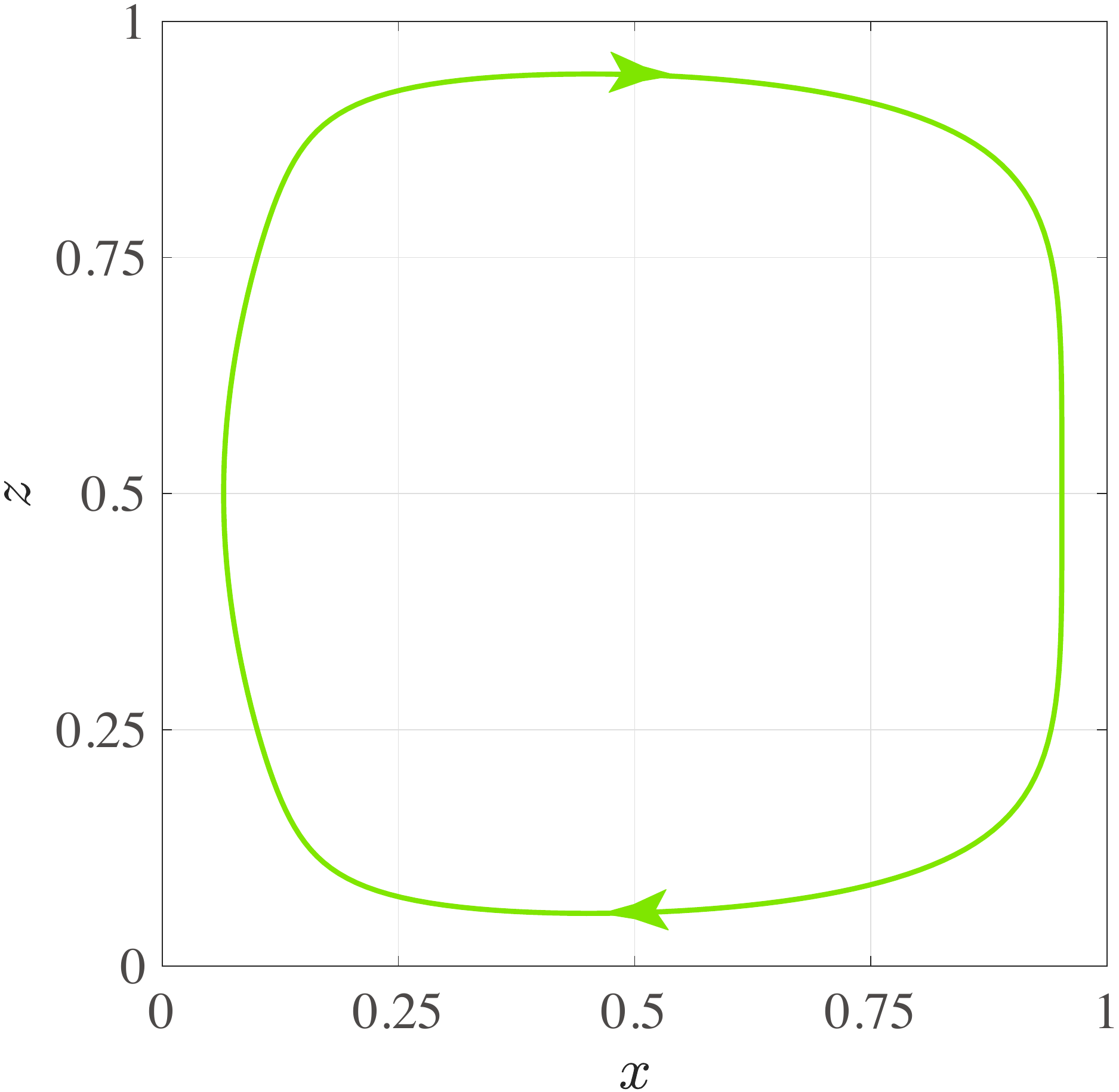}
\label{fig:others_MP4_8_before_orbit}}
\caption{4-periodic points and the projection of their orbits at $\varepsilon=0.05$}
\label{fig:others_MP4_8_before}
\end{center}
\end{figure}

\begin{figure}[H]
\begin{center}
\subfigure[4 and 8-periodic points]
{\includegraphics[scale=0.25]{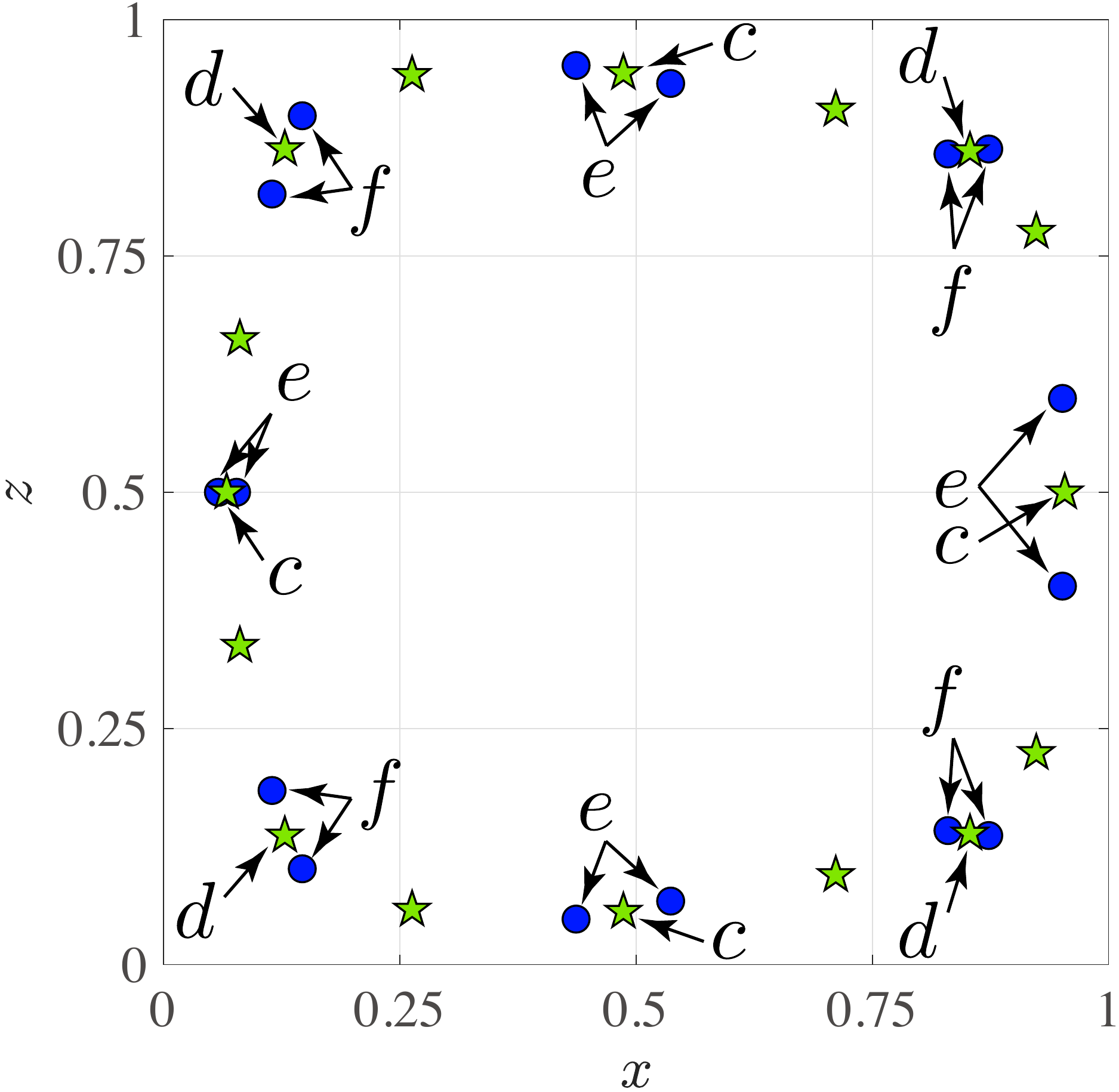}
\label{fig:others_MP4_8_after_point}}
\hspace{5mm}
\subfigure[Unstable 4-periodic orbit $c_c$]
{\includegraphics[scale=0.25]{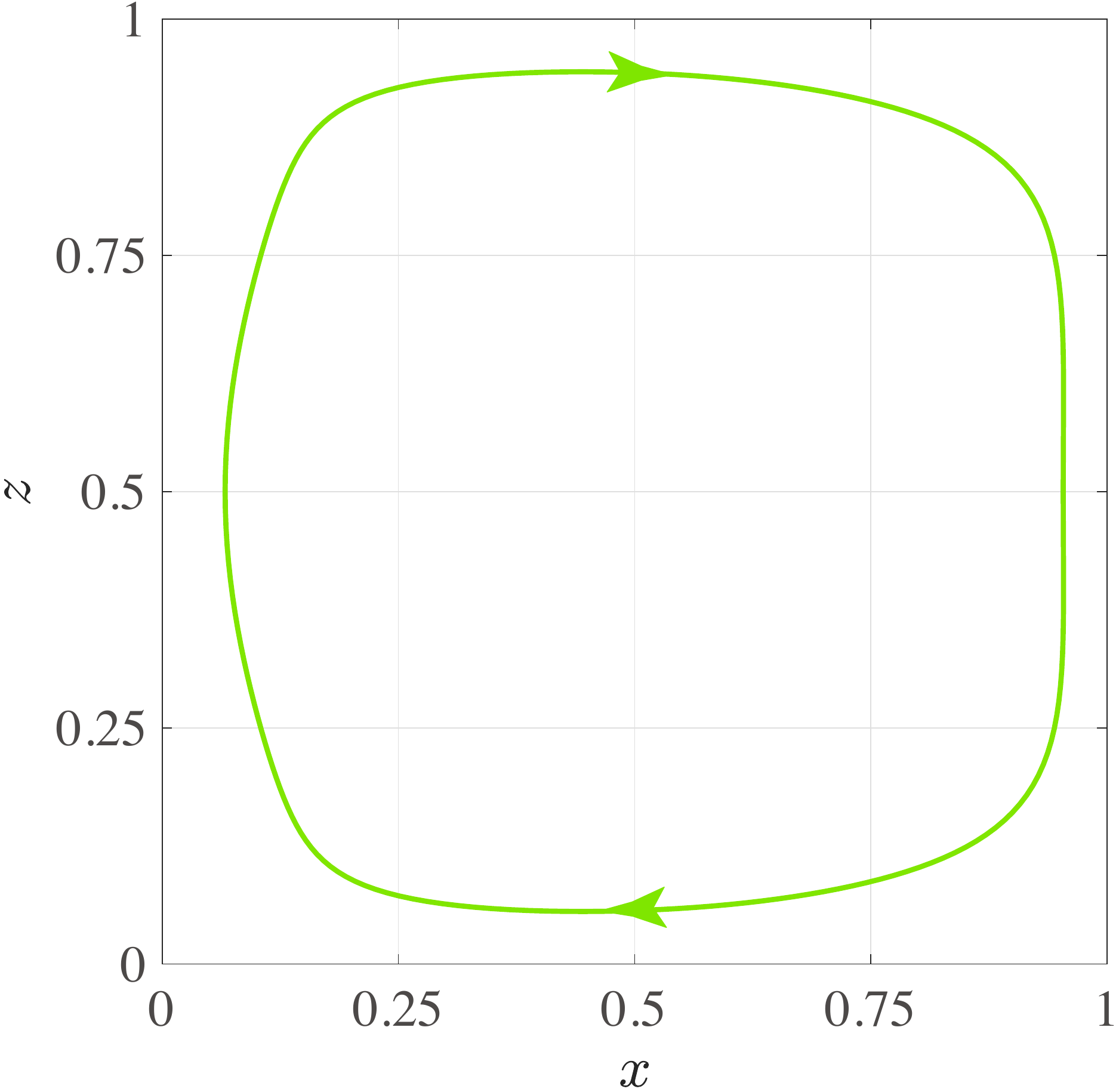}
\label{fig:others_MP4_8_after_orbit1}}
\subfigure[Stable 8-periodic orbit $c_e$]
{\includegraphics[scale=0.25]{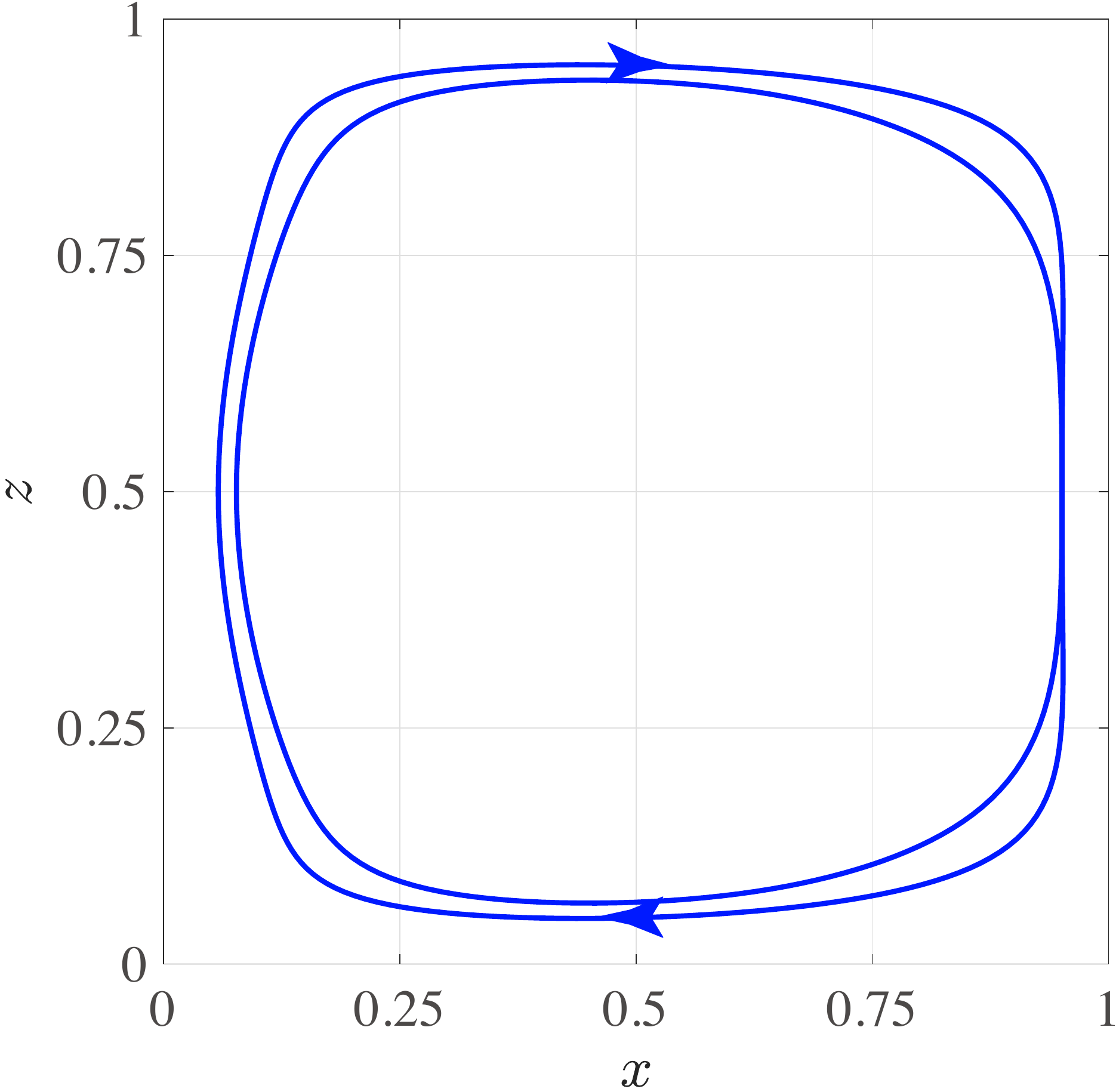}
\label{fig:others_MP4_8_after_orbit2}}
\caption{4 and 8-periodic points and the projection of their orbits at $\varepsilon=0.059$}
\label{fig:others_MP4_8_after}
\end{center}
\end{figure}

\vspace{-8mm}

\section{Conclusions}
\label{Sec:conclusion}

In this paper, we have numerically explored the global structures of periodic orbits appeared in a two-dimensional 
perturbed Hamiltonian model of Rayleigh-B\'enard convection. First we have detected the periodic points on the Poincar\'e section  and then analyzed the associated periodic orbits from the perspective of resonances and symmetries. Furthermore, we have clarified the global bifurcations regarding the periodic orbits 
associated with the parameter $\varepsilon$ which is the amplitude  of the perturbation. Thus, we have gained the following results: 

\begin{itemize}
\item KAM tori associated with elliptic $m$-periodic points have twisted structures in the extended phase space 
$\mathcal{M} = M \times S^1$, which denotes that each region of the KAM islands are mapped to the same region 
after $m$ times of Poincar\'e maps. 
From a physical point of view, they are transported periodically as a kind of vortex in the Lagrangian description. 
\item We propose a theorem regarding the symmetries of $n/m$-resonant orbits; namely, 
if the projection of an $m$-periodic orbit onto the phase space $M$ is symmetric with respect to the 
horizontal and vertical center lines of a cell, the period $m$ and the winding number $n$ of the orbit are both odd. 
It follows that $m$-periodic orbits appear in symmetric pairs when either $m$ or $n$ is even.
\item When the amplitude $\varepsilon$ of the perturbation is increased, the $m$-periodic points associated with the main KAM island $I_1$ 
disappear one after another by fold bifurcations and seem to vary to an elliptic 1-periodic point at the center of $I_1$.
\item When $\varepsilon$ is increased, $3l$-periodic points $(l=2,3,4,5)$ are generated one after another by fold bifurcations 
around the elliptic 3-periodic points at the center of KAM islands $I_2, I_3$, and $I_4$, 
where the 3-periodic points themselves also bifurcate in fold and flip bifurcations after that. 
\item Periodic points associated with other islands also bifurcate one after another 
and most of them vary to unstable ones, when $\varepsilon$ is increased. 
Some of them generate more orbits as in the fold bifurcations of 5-periodic points, 
while some others generate orbits with larger periods as in the flip bifurcations of 4-periodic points. 
Hence, the bifurcations of periodic points that may not be associated with $I_1$ 
may be the main factor that makes the fluid transport complex when $\varepsilon$ is increased.
\end{itemize}

\noindent \paragraph{Acknowledgements.} 
M.W. is partially supported by Waseda University (SR 2021C-137), 
Waseda Research Institute for Science and Engineering ‘Early Bird - Young Scientists’ community (BD070Z004400) and the MEXT "Top Global University Project". H.Y. is partially supported by JSPS Grant-in-Aid for Scientific Research (17H01097), JST CREST (JPMJCR1914), Waseda University (SR 2021C-134, SR 2021R-014), the MEXT "Top Global University Project", and the Organization for University Research Initiatives 
(Evolution and application of energy conversion theory in collaboration with modern mathematics).

%


\end{document}